\documentclass{article}

\usepackage[margin=1in]{geometry}
\usepackage[utf8]{inputenc}
\usepackage{amsmath,amssymb,amsthm,amsfonts}

\usepackage{hyperref}
\usepackage{cleveref}
\usepackage{thm-restate}
\usepackage{enumitem}
\usepackage{subcaption}

\usepackage{placeins}
\usepackage[T1]{fontenc}
\usepackage{authblk}
\usepackage[sort&compress]{natbib} 

\renewcommand{\cite}[1]{\citep{#1}}

\newlist{romanenum}{enumerate}{1}
\setlist[romanenum]{
    label=(\roman*)., 
    leftmargin=*    
}

\newcommand{\mX}{\mathcal{X}}
\newcommand{\mY}{\mathcal{Y}}
\newcommand{\bbE}{\mathbb{E}}
\newcommand{\bbI}{\mathbb{I}}
\newcommand{\bias}{\text{bias}}
\newcommand{\ebias}{\text{BIAS}}
\newcommand{\mae}{\text{MAE}}
\newcommand{\acc}{\text{acc}}
\newcommand{\eacc}{\text{ACC}}
\newcommand{\fid}{\text{fid}}
\newcommand{\efid}{\text{FID}}
\newcommand{\YSimplex}{\triangle_\mY}
\newcommand{\pref}{p_\textrm{ref}}
\newcommand{\pagg}{p_\textrm{agg}}
\newcommand{\pmarg}{p_\textrm{marg}}
\newcommand{\given}{\mid}
\usepackage{csquotes}

\usepackage[textsize=tiny]{todonotes}

\newcommand\blfootnote[1]{
	\begingroup
	\renewcommand\thefootnote{}\footnote{#1}
	\addtocounter{footnote}{-1}
	\endgroup
}

\definecolor{DarkBlue}{rgb}{0.1,0.1,0.5}
\hypersetup{
	colorlinks=true,       
	linkcolor=DarkBlue,          
	citecolor=DarkBlue,        
	filecolor=DarkBlue,      
	urlcolor=DarkBlue,          
	pdftitle={},
	pdfauthor={},
}

\usepackage[group-separator={,}]{siunitx}
\title{Addressing Discretization-Induced Bias in Demographic Prediction}

\author[1]{Evan Dong}
\author[2]{Aaron Schein}
\author[3]{Yixin Wang}
\author[4]{Nikhil Garg}
\affil[1]{Cornell University, zed5@cornell.edu}
\affil[2]{University of Chicago, schein@uchicago.edu}
\affil[3]{University of Michigan, yixinw@umich.edu}
\affil[4]{Cornell Tech, ngarg@cornell.edu}

\begin{document}

\maketitle

\begin{abstract}
    Racial and other demographic imputation is necessary for many applications, especially in auditing disparities and outreach targeting in political campaigns. The canonical approach is to construct continuous predictions -- e.g., based on name and geography -- and then to \textit{discretize} the predictions by selecting the most likely class (argmax). We study how this practice produces \textit{discretization bias}. In particular, we show that argmax labeling, as used by a prominent commercial voter file vendor to impute race/ethnicity, results in a substantial under-count of African-American voters, e.g., by 28.2\% points in North Carolina. This bias can have substantial implications in downstream tasks that use such labels.
    
    We then introduce a \textit{joint optimization} approach -- and a tractable \textit{data-driven thresholding} heuristic -- that can eliminate this bias, {with negligible individual-level accuracy loss}. Finally, we theoretically analyze discretization bias, show that calibrated continuous models are insufficient to eliminate it, and that an approach such as ours is necessary. Broadly, we warn researchers and practitioners against discretizing continuous demographic predictions without considering downstream consequences.\blfootnote{We thank David Rothschild, Eliza Hong, Erica Chiang, Eugene Kim, Gabriel Agostini, Sidhika Balachandar, Sophie Greenwood, members of the Cornell AI Policy and Practice Initiative, and anonymous reviewers of the ACM Conference on Fairness, Accountability, and Transparency (ACM FAccT) for fruitful conversations and feedback. A version of this paper has been accepted to the ACM FAccT 2024 conference and will appear as an extended abstract.}  
\end{abstract}

\section{Introduction}

Knowing demographic characteristics of individuals, particularly race and ethnicity, is necessary for many important applications, for example: auditing disparities such as in lending and policing \citep{bureau2014using,andrus2021we,chen2019fairness,pierson2020large,edwards2019risk,zhang2018assessing}, making more fair and accurate decisions in voterfile-based polling and turnout campaigns \citep{fraga2016redistricting,fraga2016candidates,fraga2018turnout,grumbach2020race,hersh2015hacking}, and decision-making more broadly \citep{ghosh2021fair,deluca2022validating,grumbach2020race,diamond2019effects}. 

However, individual-level demographic data is not always available, and so researchers and practitioners \textit{impute} (i.e., \textit{predict}) these characteristics using available individual-level data, including in all the above-cited examples. Crucially, both individual-level \textit{and} aggregate distributional prediction errors matter for many downstream tasks: in auditing the effect of voter ID laws on voting, for example, classifying a white \textit{individual} as Black affects disparity estimates, but so does counting \textit{overall} fewer Black potential voters in a geographic area. 

The standard approach is: for $N$ individuals, the user observes features $x_1, \dots, x_N$ (e.g., name, location) but not the desired labels $y_1, \dots, y_N$ (e.g., race/ethnicity). The user has access to a trained predictive model $q(y, x) \approx \Pr(y \mid x)$, for example, Bayesian Improved Surname Geocoding \citep{elliott2009using}, its variants \citep{fiscella2006use,voicu2018using,imai2022addressing,chintalapati2023predicting,greengard2024improved}, or other machine-learning based approaches \citep{jain2022importance, decter2022should,ghosh2021fair,lee2017name}. A downstream task then uses discrete \textit{decisions} $\hat{y_1},\dots,\hat{y}_N$ -- for example, to send targeted voter outreach information to individuals based on their predicted demographics, or to audit whether minoritized individuals are treated disparately.

A large literature thus pursues accurate probability models $q(y, x)$ (\citet{chin2023methods} surveys this literature; see also \citet{greengard2024improved}). Inaccurate models are known to affect downstream tasks \citep{decter2022comparing}; for example, \citet{baines2014fair,zhang2018assessing} observe that inaccurate models lead to an \textit{over-}estimate of disparities that rely on such labels. 

We, in contrast, ask: what is the right way for the user to produce individual discrete labels $\hat y_i$, given access to a model that produces continuous $q(y,x_i)$? What are the downstream consequences of different types of errors caused by the choices? A seemingly obvious choice is to assign for each point the argmax of its predictive distribution---i.e., $$\hat{y}_i \leftarrow \arg\max_y q(y, x_i)$$ for each $i$. Indeed, this approach is adopted by many users, especially practitioners \cite{TargetSmart_2022,fraga2018turnout,fraga2016candidates}. While this decision rule seems unassailable for producing a {single} label, it may lead to what was recently dubbed \textit{argmax bias} \cite{twitterargmax} in the algorithmic fairness literature, when the empirical distribution of decisions overrepresents the most likely label. Intuitively, if a group overall composes a small fraction of the population, even a Bayesian optimal predictive model may rarely assign plurality probability for that group, even for individuals of that group. More generally, as we show, discretization methods have different individual and aggregate error properties, and so one's discretization method should depend on the downstream task.

Our empirical context centers on administrative voter records data, (known as \textit{voterfile data}) in the United States, which is widely used across academia and industry \citep{ansolabehere2011gender,barber2022400,ghitza2020voter,fraga2016redistricting,fraga2018turnout}. Such data compiles the information collected by state voter registries and includes demographic information of registered voters, such as race or gender. Some state registries collect and make available demographic information that is self-reported by voters, however many do not. As a result, campaigns and academics frequently use {imputed} race/ethnicity in tasks such as survey targeting, get-out-the-vote messaging, and measuring disproportionate impacts of voter roll purges \cite{hersh2015hacking}. We analyze imputed labels in a widely-used commercial voter file from TargetSmart\footnote{We received a research license to this data from PredictWise, a campaign analytics firm.} -- which essentially uses argmax -- as well as the North Carolina voter registry directly.

We find that the algorithmically imputed race/ethnicity of registered voters is substantially skewed white: e.g., in North Carolina, the algorithmically labeled fraction of \textit{African-American} voters in the voter file is 28.2\% less than the fraction of registered voters who self-report as \textit{African-American}. This bias is due to both continuous model miscalibration (16.3\% points) and argmax bias (11.9\% points). This pattern extends to every state -- argmax discretization \textit{undercounts} \textit{voters of color in 48 out of 50 states} in comparison to the fraction implied by the continuous scores $q(y, x_i)$. Other discretization rules have their own error patterns. The commonly used \textit{threshold} rule (only classifying individuals for whom $q(y, x_i) \geq t$, e.g., for threshold $t = 80\%$) \citep{bureau2014using,adjaye2014using,chen2019fairness,diamond2019effects,pierson2020large} 

exacerbates this bias, especially geographic skews, while leaving many points unclassified. \textit{Randomizing} decisions ($i$ is labeled as $y$ with probability $q(y, x_i)$) correctly matches the predictive distribution, but is relatively inaccurate in individual labels. 

Next, we introduce \textit{optimization-based}  approaches to generate discrete labels, that construct an integer program that outputs a label for each data point $i$, while balancing multiple error objectives. We primarily consider two objectives: (a) individual data point-level \textit{accuracy} (labeling each data point according to its true class) and (b) overall population-level \textit{distributional fidelity}: where the marginal class distribution of decisions $\hat y_i$ should be close to some desired distribution (e.g., the marginal distribution implied by the model $q$, or, when available, external data on the true marginal distribution). For example, if two data points each have a probability 50\% for each of the two classes, it may {deterministically} assign them different labels so that the overall class balance of decisions is even. Empirically, we find that our approach can eliminate discretization-induced bias, \textit{even without using any information beyond model predictions $q(y, x)$}, with a negligible loss in individual-level accuracy. We further develop an efficient heuristic that can achieve similar performance, that can be interpreted as a \textit{data-driven thresholding} approach, informed by the optimization solutions on a small subset of the data.   

Finally, we theoretically characterize \textit{decision-making} (discretization) bias, as distinct from \textit{predictive modeling bias}: how does amplification of the most likely class depend on the properties of the discretization procedure and the continuous scores, respectively? In \Cref{thm:informationargmax}, we show that argmax bias emerges even when the continuous predictive model is itself \textit{unbiased} (calibrated), and is tightly connected to \textit{predictive uncertainty}. In \Cref{thm:pareto}, we analyze \textit{decision-making} approaches to reduce bias: how should discrete decisions $\hat y_i$ be made from a Bayes optimal model of continuous probabilities $q (y, x_i)$? We show that a joint optimization approach such as ours is \textit{necessary} for optimal discretization in the presence of a distributional objective and predictive uncertainty.

Putting things together, we caution against the use of existing discrete labels, as commonly distributed in commercial voter files, for sensitive applications where such bias would affect results. We further answer the question: what should a researcher or practitioner do to prevent discretization-induced bias?

\begin{enumerate}[label=(\arabic*),leftmargin=2em]
\item When possible, use the \textit{continuous} scores instead of single discrete labels, as in \cite{chen2019fairness,mccartan2023estimating,deluca2022validating,kallus2022assessing,fraga2023reversion,bureau2014using}. For example, some auditing tasks can be conducted with sufficiently calibrated scores. However, for some tasks, such as individual-level outreach decisions,  individual-level discrete decisions might either be required or be easier to integrate into practitioner analysis pipelines.

    \item Improve the continuous model $q$ when possible -- e.g., by acquiring more data or changing the modeling approach. This is the predominant approach pursued in the demographic prediction literature \cite{chin2023methods}. However, this approach is {insufficient}. As we make precise in \Cref{thm:informationargmax}, improving continuous predictive model accuracy reduces the bias; {however}, even unbiased Bayes optimal predictors induce argmax bias, unless predictions have zero error (can \textit{perfectly} identify the true label for each data point). 

    \item When discretizing, consider the downstream desiderata -- e.g., is accuracy among the labeled set most important, or label representativeness? -- and evaluate how label error may affect results. For example, \citet{fraga2016candidates,fraga2018turnout} count voters by imputed race and account for variance induced by the imputation process; they further discuss that bias as we measure here would lead to conservative estimates of differences in voter turnout induced by changes such as redistricting, though do not empirically quantify the potential effect on the results.
    
    \item Finally, the optimization-based approaches developed here can flexibly discretize labels that have the error properties that best match the downstream desiderata. For example, we illustrate approaches that maximize accuracy while matching group distributions both overall and per geographic subarea. 
    
\end{enumerate}

\subsection{Methods and data summary}
\label{sec:methodssummary}
We briefly summarize our primary methods and data, with additional details in the model and methods sections. Consider a trained model (such as BISG) that outputs a vector of label probabilities $q(x_i)$ for each data point $x_i$. A (potentially randomized) \textit{decision (discretization) rule} $D: \YSimplex^N \to \mY^N$ is a rule that, given a set of probability vectors $\{q(x_i)\}_{i = 1}^N$ associated with $N$ data points, assigns each data point $i$ a label $\hat y_i \in \mY$.

\subsubsection{Decision desiderata} There are two high-level {desiderata}: (1) \textbf{decision accuracy}, i.e., making the correct decision \textit{for each data point}; (2) making sure that the \textit{distribution} of assigned labels $\{\hat y_i\}$ matches some desired distribution, a notion that we will capture both via (a) the \textbf{bias} for each class $y$ and (b) full \textbf{distributional fidelity}. For example, we may want to assign a correct demographic label for each potential voter (for more effective personalized get-out-the-vote campaigns) but also have the overall label distribution match the demographic makeup of the state so that we're not under-counting minoritized groups when conducting overall analyses. 

\textbf{(1)} To measure of \textbf{decision accuracy}, we use (one minus the) 0-1 loss, i.e., given true labels $y_i$ and decisions $\hat y_i$: $\acc(y_{1:N}, \hat y_{1:N}) = \frac1N\sum_i^N\mathbb{I}\left[\hat y_i = y_i\right].$ 

\textbf{(2a)} To measure \textbf{bias}, we use the distance between the marginal distribution of labels and a \textit{reference} distribution $p_{\textrm{ref}}$. Let $\hat p_{\textrm{marg}}(y)$ denote the fraction of datapoints with the label $y$:
\[
\hat p_{\textrm{marg}}(y, \hat y_{1:N}) = \frac{1}{N}\sum_{i=1}^N \mathbf{1}[\hat{y}_{i} = y].
\]
Then, {bias} is how much a class $y$ is amplified by labels relative to the reference:
\begin{align*}
	\bias(y,  \hat y_{1:N}, \pref) \triangleq \hat p_{\textrm{marg}}(y, \hat y_{1:N}) - \pref(y).
\end{align*}

\textbf{(2b)} As a summary statistic, we will further consider \textbf{distributional fidelity}, the negative sum of the absolute value of bias across classes (equivalently, the negative of the $\ell_1$ distance between the label and reference distributions).
\[\fid(p_{\textrm{ref}}, \hat y_{1:N}) \triangleq  -{\sum_{y \in \mY} \Big|\bias(y, \hat y_{1:N}, \pref)}\Big|.\]
\vspace{.25em}

\paragraph{Aggregate posterior reference distribution} What is an appropriate reference distribution $\pref$? Of course, if we knew the \textit{true} distribution of labels (e.g., the true fraction of each group in the voter file), we could measure and optimize fidelity with respect to it. However, this information is not often known. In this paper, we will thus primarily compare to the \textbf{aggregate posterior}: for a given set $\{x_i\}$, what would the decision distribution be if we could make \textit{continuous} decisions corresponding to the continuous classifier $q$,
\begin{align*}
	p^q_{\textrm{agg}}(y, \{x_i\}) = \frac1N\sum_i q(y, x_i).
\end{align*}
When $q$ is Bayes optimal, this choice counts the ``correct'' number of decisions for each class if one was not forced to make discrete decisions, and approaches the true distribution $\Pr(y)$ as $N \to \infty$. 
This approach thus (1) isolates bias due to the discretization process as opposed to miscalibration in the continuous classifier $q$; (2) can be optimized for and empirically measured, given predictions from $q$; i.e., this approach requires no more information than the argmax rule. 

Notably, one can also calculate conditional aggregate posteriors -- for example, the fraction of labels for each group \textit{within each county}, if we observe such geographic information for each data point.

\subsubsection{Overview of decision rules}

\begin{figure}[tbh!]
	\centering
	\begin{subfigure}{.45\textwidth}
		\centering
		\includegraphics[width=1.16\textwidth]{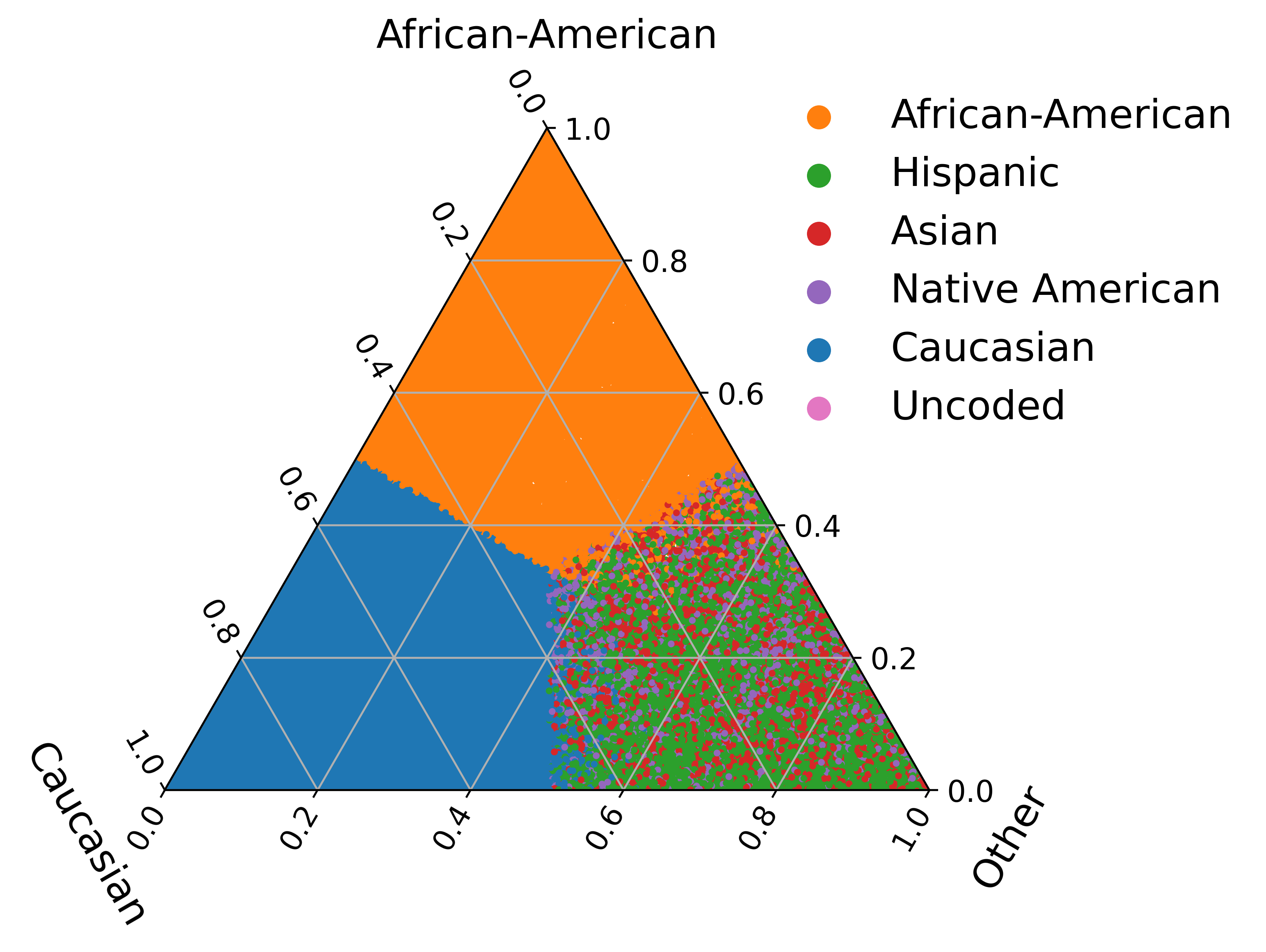}
		\caption{Argmax}
		\label{fig:simplex-argmax}
	\end{subfigure}
	\hfill
	\begin{subfigure}{.45\textwidth}
		\centering
		\includegraphics[width=\textwidth]{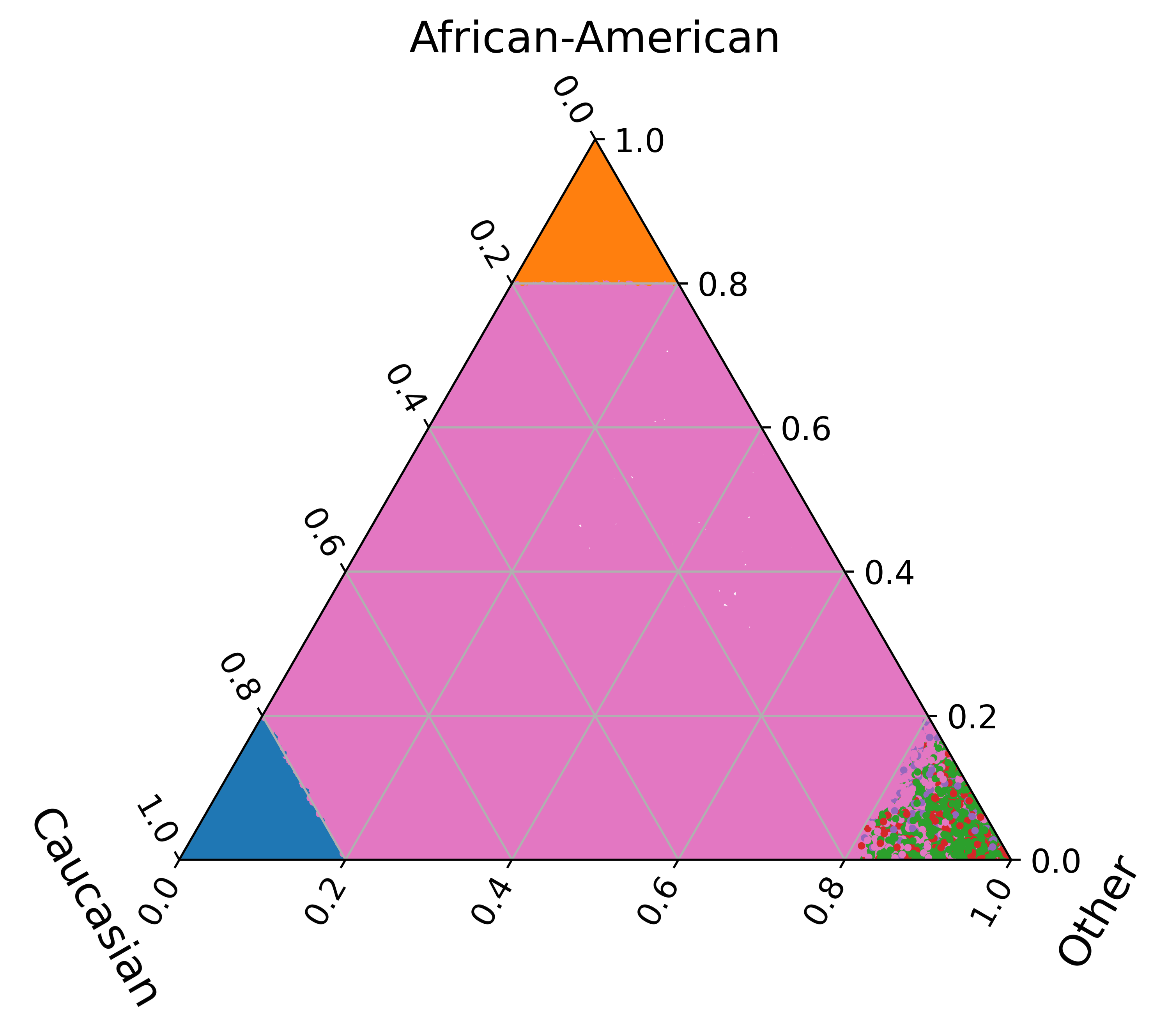}
		\caption{Threshold at probability $0.8$
		}
		\label{fig:simplex-threshold}
	\end{subfigure}    
	
	\begin{subfigure}{.45\textwidth}
		\centering
		\includegraphics[width=\textwidth]{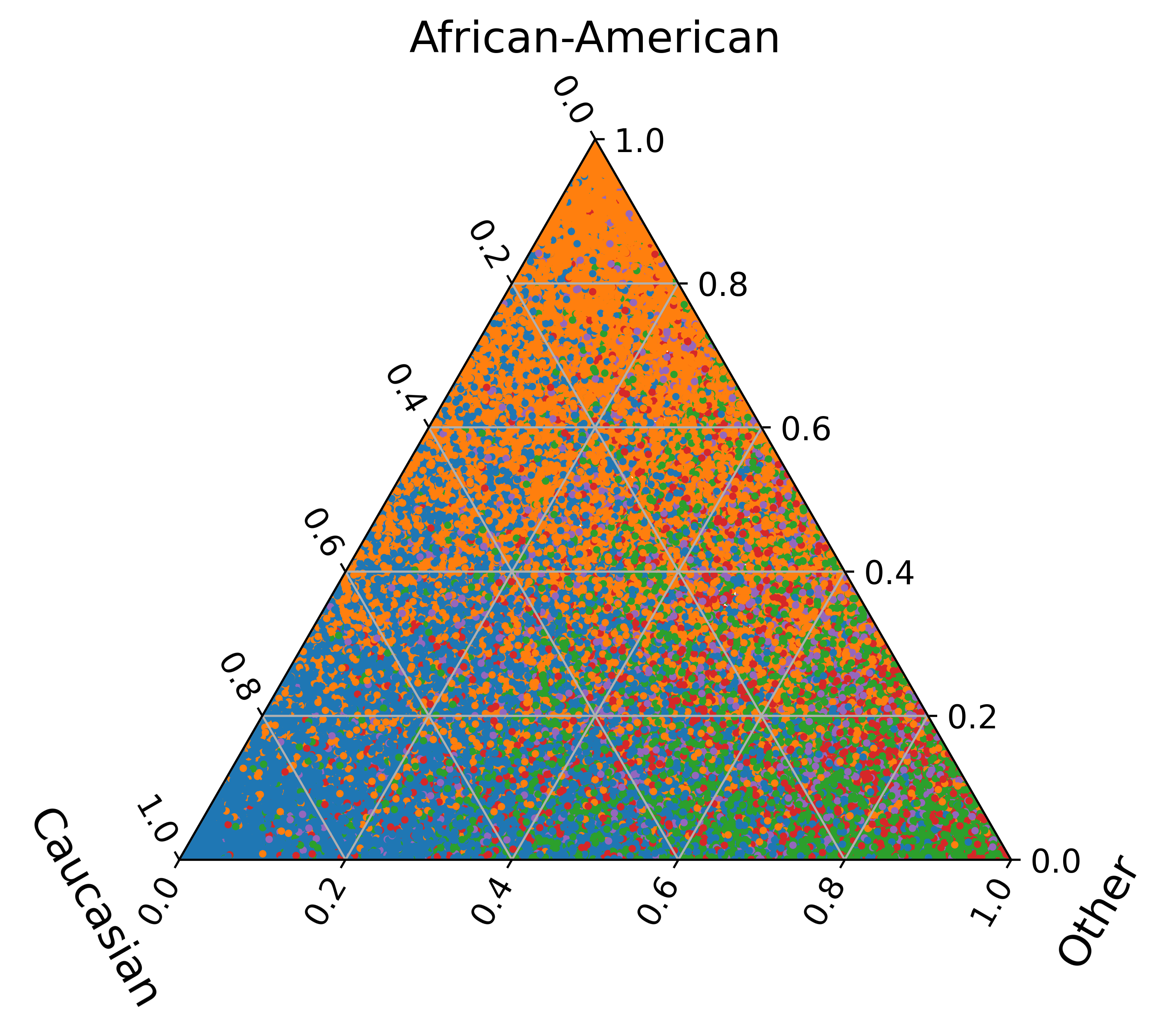}
		\caption{Thompson Sampling
		}
		\label{fig:simplex-thompson}
	\end{subfigure}
	\hfill
	\begin{subfigure}{.45\textwidth}
		\centering
		\includegraphics[width=\textwidth]{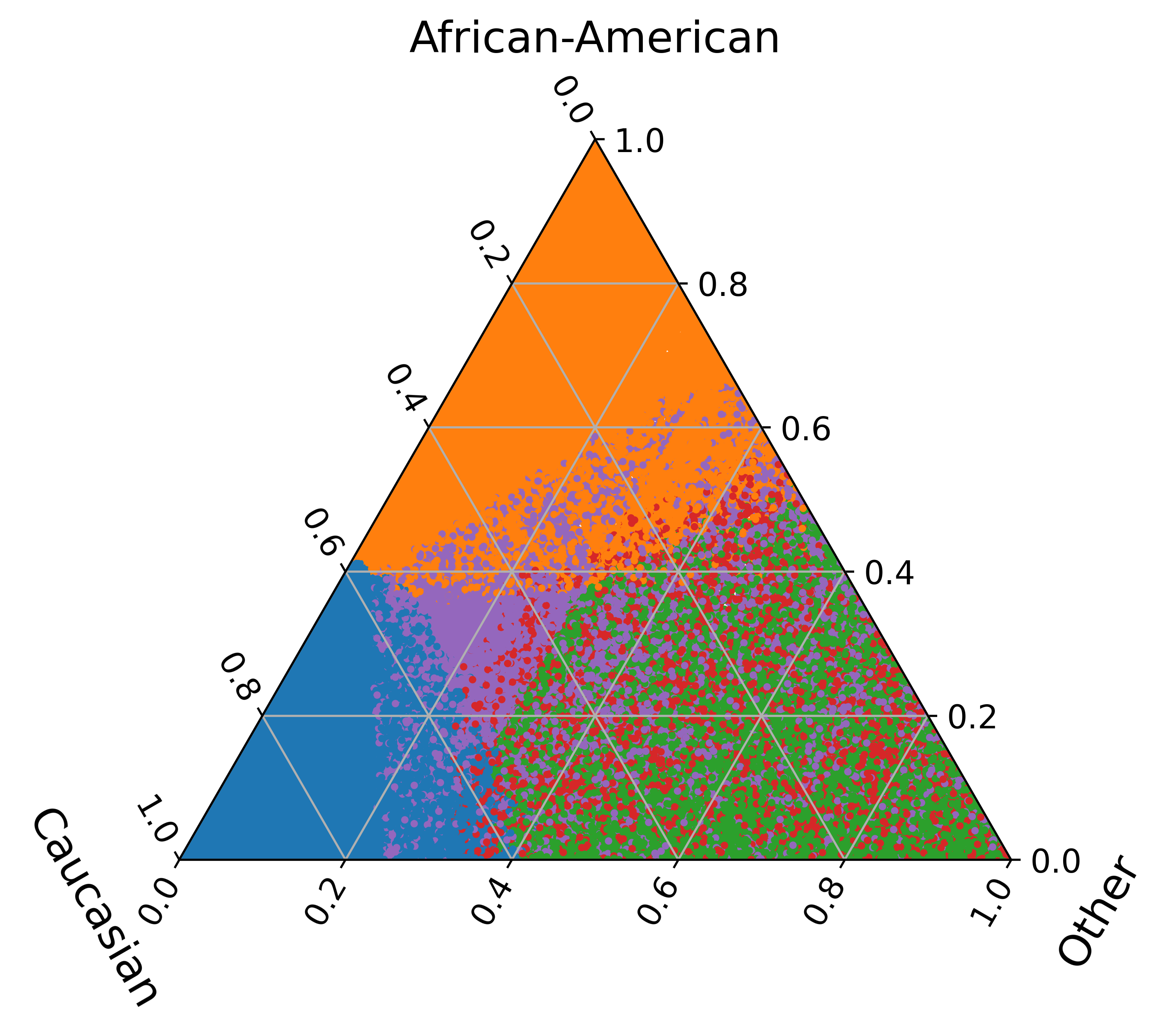}
		\caption{Aggregate Posterior matching
		}
		\label{fig:simplex-matching}
	\end{subfigure}    
	
	\caption{Comparison of different discretization methods. Each subfigure shows a 3-dimensional probability simplex, where individual points are colored according to the label assigned by the corresponding method. For example, in (a), all the blue points are assigned the \textit{Caucasian} label, which has the highest class probability according to the continuous model $q$ for that data point. We use a sample of points from the voter file, and \textit{Hispanic}, \textit{Asian}, and \textit{Native American} probabilities are aggregated into \textit{Other}. Our approach in \Cref{fig:simplex-matching} (posterior matching) matches the class distribution while maintaining individual data point level accuracy.}
	\label{fig:simplex-main}
\end{figure}

Several rules are used in academia and in practice:

\begin{description}
\item [Argmax] rule assigns, for each data point, the most likely\footnote{Suppose ties are broken according to some consistent ordering over the classes. In the event of ties, we will use $\arg\max_y q(y, x)$ to denote the \textit{single} class that is in the argmax set and is first in the consistent ordering among classes in the argmax set.} class: $D_{\textrm{argmax}}(\{q(x_i\}) \triangleq \{\hat y_i = \arg\max_y q(y, x_i)\}$.

\item [Threshold at $t$] rule assigns a label only if the argmax class has a probability of at least $t$, i.e., $\hat y_i = \arg\max_y q(y, x_i)$ if $\max_y q(y, x_i) \geq t$, otherwise $i$ is left uncoded. 

\item [Thompson sampling] for each data point $i$, \textit{samples} a class $y$ based on the probabilities $q(y, x_i)$, i.e., $\hat y_i = y$ with probability $q(y, x_i)$, and so in expectation will have labels matching the aggregate posterior. 

\end{description}

These rules are all \textit{independent}: the assigned label for a data point $i$ depends on $q(x_i)$ but not on other data points or their probability vectors $\{q(x_j)\}_{j \neq i}$.

In contrast, our proposed approach also includes {joint} decision optimization rules, which depend on the entire set of data points, as a solution to discretization bias.

\begin{description}
	\item [Integer optimization] rules directly optimize given objectives. We consider the rule that corresponds to solving the following optimization problem, for a given $\gamma$, data points $\{x_i\}$, and reference $p_{\textrm{ref}}$:
	\begin{equation}
		D^\gamma(\{q(x_i)\}, p_{\textrm{ref}}) = {\arg\max}_{\{\hat y_i\}} \gamma \left(\frac{1}{N}\sum^N_{i=1} q(\hat y_i, x_i)\right) + (1-\gamma) \fid(p_{\textrm{ref}},\hat y_{1:N}). \label{eq:optdecisionrules}		
	\end{equation}
	The rule balances, parameterized by $\gamma \in [0, 1]$, row-level scores (accuracy as measured by $q$) and distributional fidelity.
	
	\item [Aggregate posterior matching] A special case of the joint optimization decision rules defined in \Cref{eq:optdecisionrules} is as $\gamma \to 0$, leading to a rule that maximizes accuracy subject to the constraint that distributional fidelity be as high as possible. We call such rules \textit{matching} solutions, e.g., \textit{aggregate posterior matching}. For our results, we use aggregate posteriors calculated at multiple geographic levels: nationally, throughout a state, or for each County within the state.
	
	\item [Data-driven thresholding heuristic] One downside of the optimization based approaches is that they may be computationally expensive for large datasets (requiring an integer constrained variable for each data point). Thus, we develop and evaluate the following heuristic: solve the integer optimization for a tractable batch of data. Then, on that batch, train a machine learning model with features being the scores $q(y, x_i)$ and the labels being the optimization outputs $\hat y_i$. Finally, apply this model to the rest of the dataset. Once trained, this approach is as simple as existing thresholding or argmax approaches, and so we refer to it a \textit{data-driven thresholding heuristic} (with vector thresholds implicitly defined by the machine learning model). In this work, we train this approach to approximate the aggregate posterior matching solution and integer programs.

\end{description}

\Cref{fig:simplex-main} illustrates, for several decision rules, labels as a function of the class probabilities. Intuitively, like Thompson sampling, aggregate posterior matching assigns a proportional number of data points to each class, which may require individual points to not be labeled by its most likely class. However, its labels are more in line with argmax labeling than is Thompson sampling. 

\subsubsection{Data and code} We evaluate a widely used commercial data file that provides both predicted continuous race/ethnic probabilities and a single discrete generated label for each individual. The voter file provides continuous probabilities for people along $K=5$ predicted race/ethnicity categories: \textit{African-American}, \textit{Asian}, \textit{Caucasian}, \textit{Hispanic}, and \textit{Native American}.\footnote{These categories are imperfect, do not disambiguate race from ethnicity, are socially constructed, and there is large uncaptured granular variation; to the extent that demographic-based analyses are important, we support the collection of more granular, self-reported data \cite{movva2023coarse,hanna2020towards, read2021disaggregating, kauh2021critical, wang2022towards}. We note that the voter file data treats \textit{Hispanic} as mutually exclusive with racial categories, and furthermore that racial self-identification as Native American differ from legal tribal citizenship.} The labels also include an \textit{Uncoded} category. Crucially, according to both the data dictionary and verified by us, the discrete labels are derived using the \textit{argmax} rule, except for the \textit{Uncoded} category.\footnote{The data dictionary states that an individual is marked as \textit{Uncoded} when the model is not confident enough to make a prediction. From our analyses, it seems that the process of an individual being \textit{Uncoded} cannot be simply described or reverse engineered as a threshold based rule; some uncoded rows have individual class probabilities of over 90 percent. We believe that validation from external data sources may play a role in an individual being \textit{Uncoded}.} 

In our main text analysis, we relabel these \textit{Uncoded} points via the argmax rule.\footnote{We replicate all results in \Cref{sec:dropuncoded} when we instead exclude these voters. As \textit{Uncoded} data points are disproportionately predicted in probability to be voters of color, excluding these makes the population \textit{less} diverse, \textit{increasing} measured bias; results are qualitatively identical.} 

We analyze the universe of available predictions and state voter files ($N=\num{261547234}$) across all 50 states and the District of Columbia. Furthermore, some states (such as North Carolina) additionally make available each voter's \textit{self-identified} race/ethnicity, as filled in on the voter registration form. Our data includes this information where available alongside the algorithmically generated predictions and discrete labels; in some of our analyses, we use self-report data as \textit{ground truth labels}\footnote{Though we refer to these self-report labels as ground truth, we note that prior work has found that individuals may not be consistent in their reporting over time \cite{fraga2016candidates}.} to measure true accuracy and to calculate a ground truth reference distribution $\pref(y)$. For these analyses in the main text, we filter our dataset of model predictions to include North Carolina only. This analysis contains $N=\num{6374636}$ individuals, after excluding the $\num{2312356}$ individuals in NC without self-reported race and the $\num{138541}$ voters who self-report race as \textit{Other}.\footnote{Of course, excluding such individuals from consequential analyses could introduce further bias, speaking to the fraught nature of demographic imputation and coarse racial/ethnic data.} In the Appendix, we replicate our primary analyses using fully public data and predictive modeling methods from \citet{greengard2024improved}, as well as for other states in the commercial voter file. Our code is available at \href{https://github.com/evan-dong/demographic-prediction-argmax-bias}{https://github.com/evan-dong/demographic-prediction-argmax-bias}. 

\section{Empirical Results}
\label{sec:empirics}

\begin{figure}[tb]
	\centering
	\begin{subfigure}{.3\textwidth}
		\centering
		\includegraphics[width=.95\linewidth]{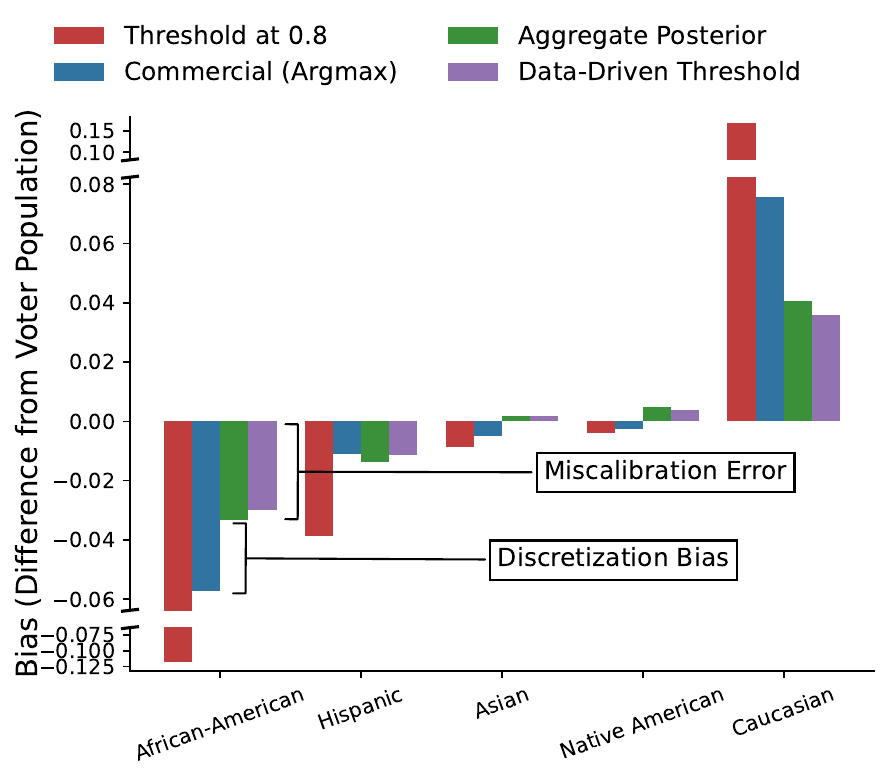}
		\caption{Differences in predicted and self-reported counts in NC.}
		\label{fig:votercounts}
	\end{subfigure}
	\hfill
	\begin{subfigure}{.3\textwidth}
		\centering
		\includegraphics[width=.95\linewidth]{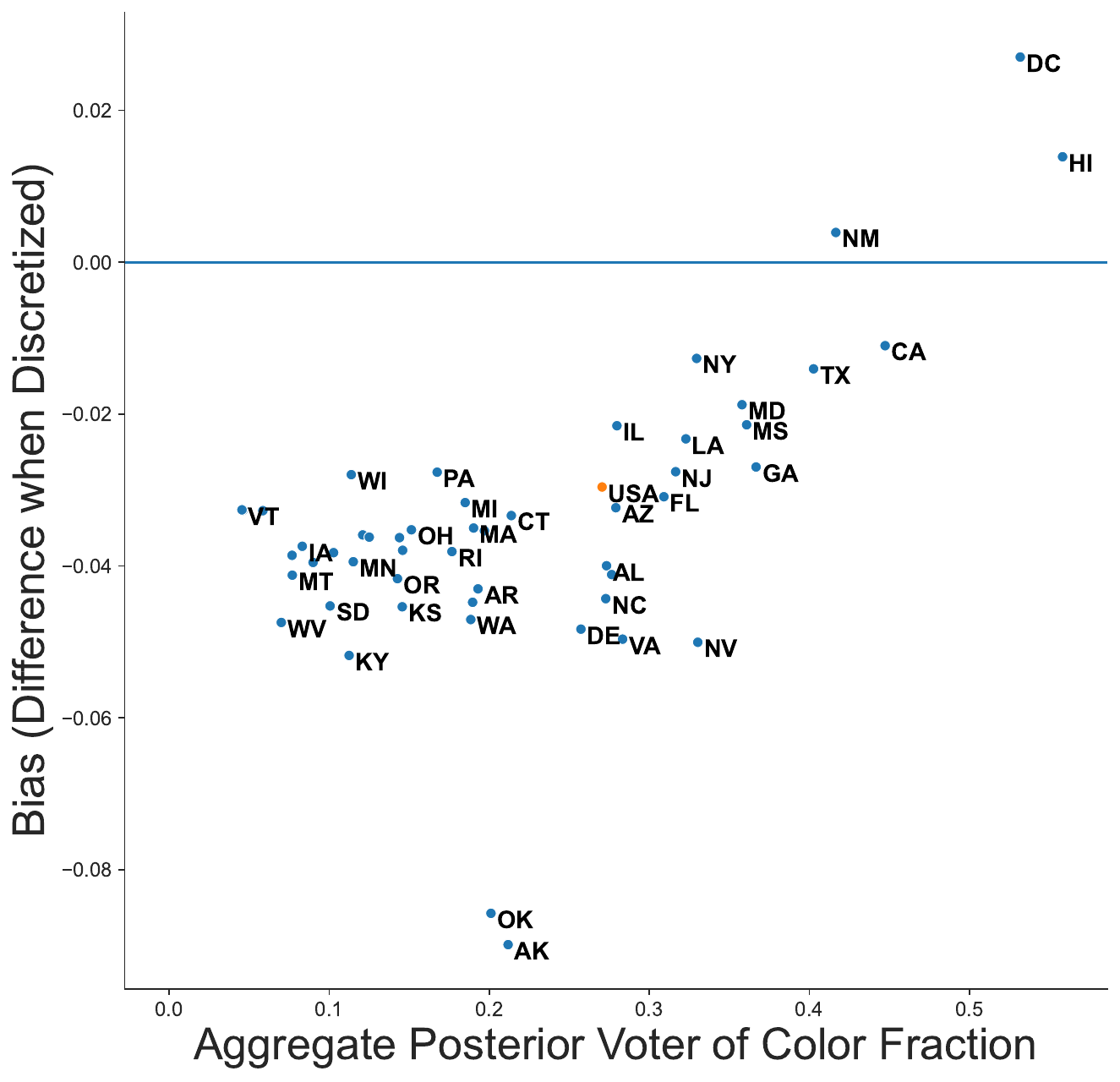}
		\caption{Per-state discretization bias.}
		\label{fig:pocrep}
	\end{subfigure}
	\hfill
\begin{subfigure}{.3\textwidth}
	\centering
	\includegraphics[width=.95\linewidth]{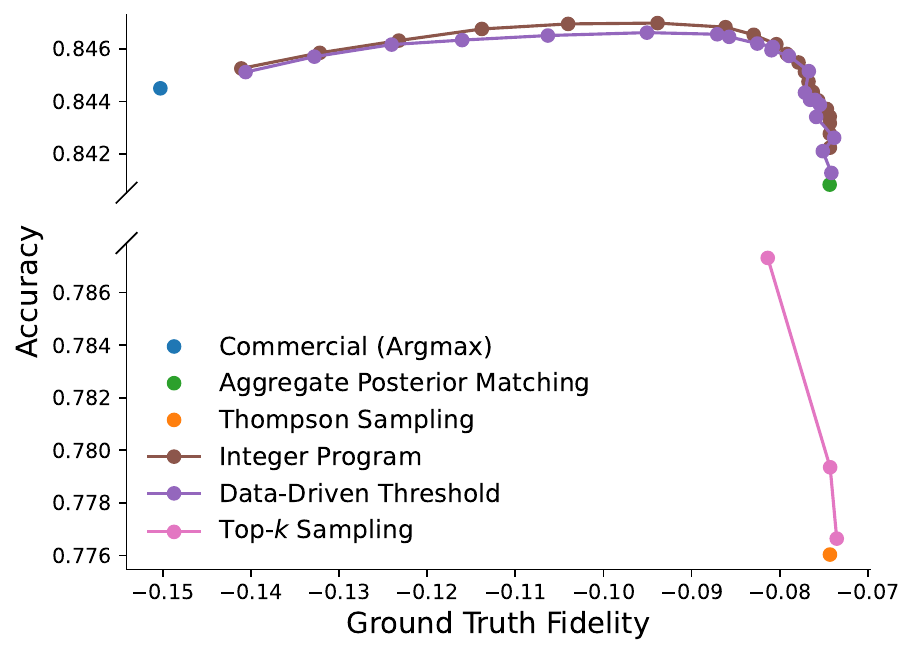}
	\caption{Accuracy versus Fidelity of rules in NC.}
	\label{fig:pareto-NC}
\end{subfigure}
	\caption{Bias (undercounting of voters of color) in the voter file. (a) In North Carolina where ground truth self-reported data is available, the difference in counts for each group between each discretization method and the ground truth. The argmax method substantially undercounts African Americans in particular, with threshold methods further magnifying such bias. The bar marked ``Aggregate Posterior'' corresponds to the bias of both Thompson Sampling and Aggregate  method; as these methods directly reflect proportions from the continuous model, its undercounting is due to the model's miscalibration. The additional bias of argmax is thus the bias caused by discretizing the model scores. (b) In all states plus Washington, DC, the discretization bias (difference between the fraction of voters of color in the discrete labels and the aggregation posterior fraction). Points below the horizontal line indicate a comparative underrepresentation of voters of color compared to the aggregate posterior -- only in DC, Hawaii, and New Mexico does the discretization lead to an increase of the count of voters of color. Note that in DC and Hawaii, Caucasian is not the most common class: African Americans (DC) and Asian (HI) are, respectively, and these classes are over-represented by argmax. New Mexico (NM) is the one exception where the argmax decision rule under-represents the most common class compared to the aggregate posterior. In (b), some similarly clustered states are left unlabeled for visual clarity, and the overall effect in full voter file is marked in orange. In (c), we plot the performance of different decision rules according to our two metrics. Notice that our optimization-based rules Pareto-dominate the sampling approach, and that the data-driven threshold approximates the full curve of integer programs quite closely. }
\label{fig:voterbias}
\end{figure}

\begin{table}[tb]
	\begin{center}
		\begin{tabular}{|c|c|c|c|}
			\hline
			& \text{Accuracy with} & \text{Fidelity to} & \text{Fidelity to} \\
			& \text{Ground Truth} & \text{Ground Truth} & \text{Aggregate Posterior} \\
			\hline\hline
            Threshold at 0.8 ($67.9\%$ not dropped) & \textbf{0.928} & -0.352 & -0.293 \\ 
            Integer Program, $\gamma = 0.9$ & 0.846 & -0.079 & -0.019 \\ 
            Commercial (Argmax) & 0.844 & -0.150 & -0.091 \\ 
            Aggregate Posterior Matching & 0.841 & \textbf{-0.074} & \textbf{-0.000} \\ 
            Data-Driven Threshold Matching Heuristic& 0.841 & -0.074 & -0.002 \\ 
            \text{County-Conditional Aggregate Posterior Matching} & 0.834 & -0.074 & -0.000 \\   
            Thompson Sampling & 0.776 & -0.074 & -0.000 \\ 
            \hline
        \end{tabular}
		\caption{In North Carolina, the performance of different decision-making discretization rules. Fidelity is compared both to the true marginal distribution (i.e., the self-reported voter population) and the aggregate posterior; differences between the two reflect probabilistic score miscalibration. Both sampling-based approaches are \textit{strictly dominated} by joint decision-making rules, which achieve equal or better fidelity with substantially better accuracy. Furthermore, both matching approaches substantially improve on distributional fidelity (to either ground truth or aggregate posterior) than the argmax rule, with negligible loss in accuracy. Note that the aggregate posterior matching approach in particular uses no more data than the argmax approach. Metrics for the threshold rule are calculated only on labeled points, dropping $32.1\%$ of the dataset.}
		\label{tab:table-NC} 
	\end{center}
\end{table}

\begin{figure}[!ht]
	\centering
    \begin{subfigure}{.45\textwidth}
        \centering
        \includegraphics[width=1.25\linewidth]{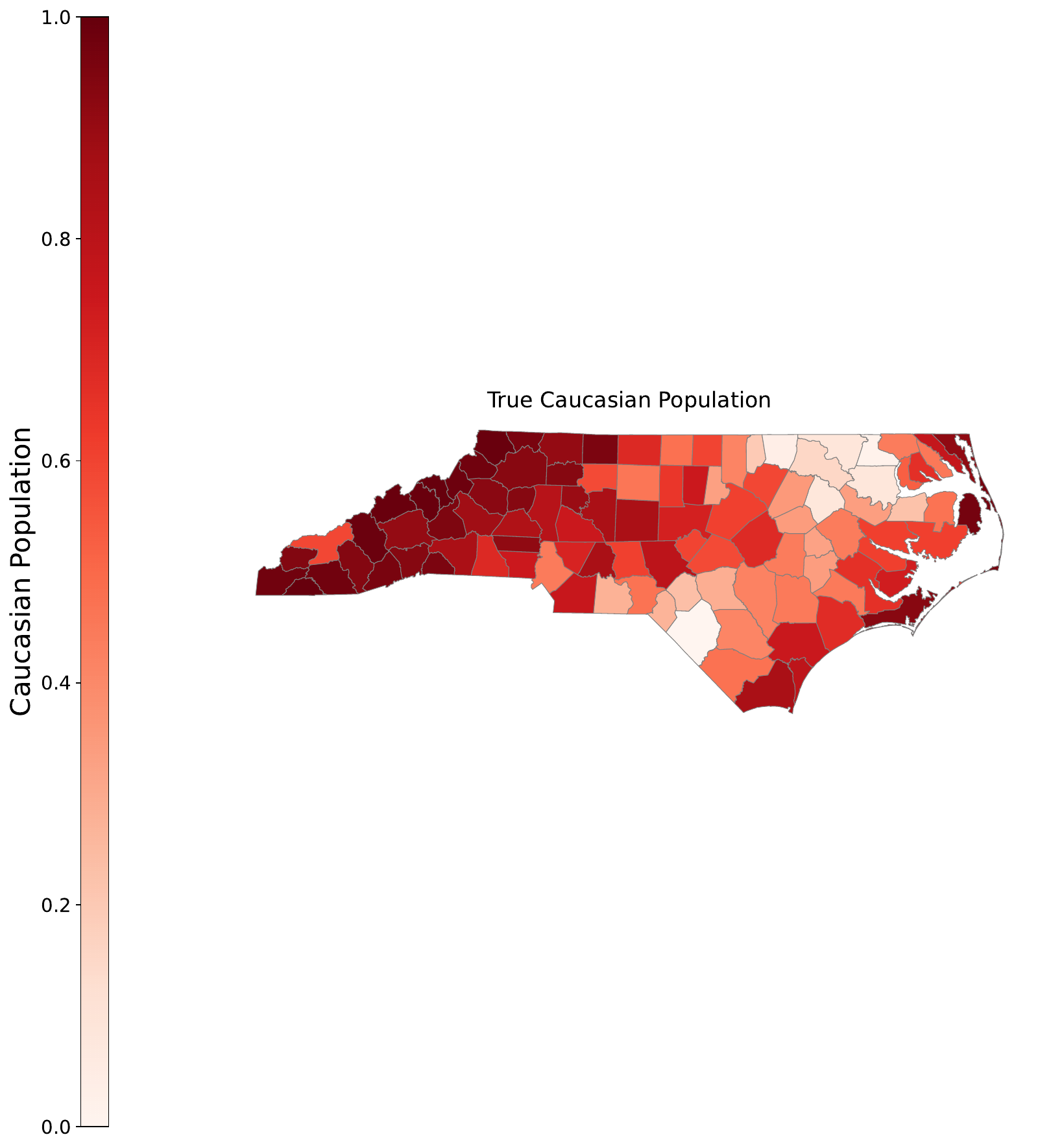}
        \caption{Voter file population.}
        \label{fig:popmap}
    \end{subfigure}	
    \hfill    
    \begin{subfigure}{.45\textwidth}
        \centering
    \includegraphics[width=.85\linewidth]{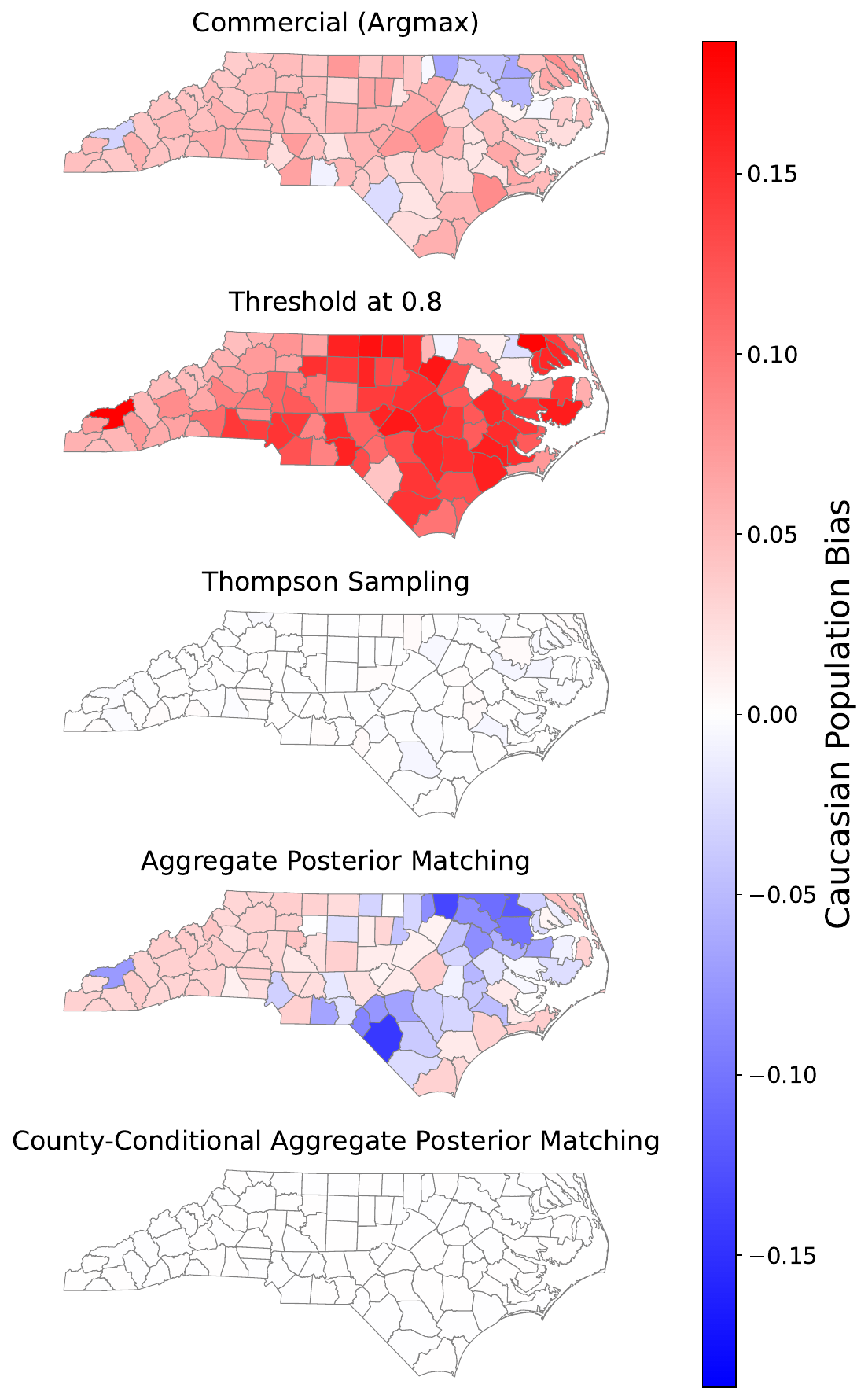}
	\caption{Bias compared to per-county aggregate posterior. }
	\label{fig:biasmap}
    \end{subfigure}
	\caption{Per-county Caucasian bias in North Carolina, under different discretization methods -- the underrepresentation of voters of color (i.e., bias overrepresenting the white population) when compared to the aggregate posterior. The most commonly used methods of thresholding and argmax further cause \textit{geographic} skews -- in many parts of the state, very few rows are classified as non-Caucasian, and the bias is largest in counties with an already skewed population. In contrast, Thompson sampling and County-conditional aggregate posterior matching have no bias when compared to the aggregate posterior.}
	\label{fig:maps}
\end{figure}

\paragraph{Discretization bias and Voter of Color Representation}
Argmax labeling substantially undercounts voters of color. \Cref{fig:votercounts} shows -- for North Carolina voters for whom self-identified labels are available -- the bias (relative count) of each decision rule compared to this ground truth. \textit{Every group except} \textit{Caucasian} is under-counted by the argmax rule, and this under-counting is substantial. The predicted population differs by over four hundred thousand \textit{African-American} voters out of a self-reported population count of $\num{1447516}$; i.e., the discrete labels count over 28\% fewer African Americans than self-reported figures. More precisely, compared to ground truth population distribution for each group, the labeled distribution has fewer \textit{African-American} ($-28.2\%$), \textit{Asian} ($-20.1\%$), \textit{Hispanic} ($-14.9\%$), and \textit{Native American} ($-32.8\%$) voters while inflating the \textit{Caucasian} population by $+10.1\%$. Furthermore, a large source of this undercounting is due to the \textit{discretization} rule (argmax) in addition to miscalibration in the continuous classifier. As the predictions from Thompson Sampling and Aggregate Posterior Matching avoid this discretization bias, our framework allows a comparison of the two sources: the continuous classifier's miscalibration causes $16.3\%$ of the undercounting, with the remaining $11.9\%$ caused by argmax discretization. (The classifier calibration curves using NC data are given in Appendix \Cref{fig:calibration-nc}). 

These patterns extend to the other states. 
Since we do not have \textit{self-identified ground truth} labels in every state, \Cref{fig:pocrep} compares the fraction of voters labeled as a person of color\footnote{While defining the boundaries of race and whiteness is always contentious, given the data at hand and these commercially used labels, our operational definition of voters of color includes \textit{Hispanic} voters of any race.} (any label except \textit{Caucasian}) by the vendor, versus the aggregate posterior fraction according to the same vendor's continuous probabilities. In \textit{nearly every state}, the fraction of people labeled as a person of color is \textit{less than} the aggregate posterior -- i.e., the argmax discretization rule comparatively undercounts people of color everywhere except New Mexico, Hawai'i, and DC. Notice the general trend in voter of color representation: the larger the white majority population, the larger the discretization bias.

\paragraph{Mitigating bias with optimization-based decision rules} 
These biases can be mitigated via changing the discretization rule, with negligible loss in individual-level accuracy. \Cref{fig:votercounts} shows that matching to the aggregate posterior (which does not require any more information than argmax) substantially reduces bias, and \Cref{fig:pareto-NC} shows accuracy versus fidelity to the ground truth for each rule. \Cref{tab:table-NC} lists, for each decision rule, the accuracy with respect to the ground truth as well as distributional fidelity to both the ground truth distribution and the aggregate posterior (additional rules are shown in \ref{sec:methodsdetails}).

Remarkably, our joint optimization-based rules have similar accuracy as the argmax rule (0.36\% more 0-1 loss), with \textit{substantially less} bias (from 15\% to 7.4\% with respect to ground truth). This is true for the matching rules (that exactly match the distribution implied by the continuous models), the data-driven thresholding heuristic, and integer optimization rules that balance accuracy and distributional fidelity. Furthermore, the sampling-based approaches are highly suboptimal: Thompson sampling is substantially less accurate, with comparable distributional fidelity and bias to the joint rules.

Two aspects of the results are especially striking. (1) Our optimization-based decision rules outperform standard rules even though they are not given the {true} marginal distribution over classes, i.e., when it uses just the information given by the continuous classifier (and so can be implemented in every state, even when true population information is not available, just like the argmax rule). (2) Even the extreme fidelity-prioritizing choice of $\gamma \to 0$ optimized by the matching approaches induces negligible accuracy loss, and so it seems appropriate for any downstream application (and notably, it is the opposite choice than the implicit $\gamma = 1$ choice made by the current argmax approach).  

Together, these results show that commercial voter files undercount voters of color, and this undercounting can be mitigated simply by changing how vendors convert continuous probabilities to discrete labels. 

\paragraph{Geographic skews} 
Finally, we consider sub-state bias. \Cref{fig:maps} shows, for each county in North Carolina, the bias (in terms of white voters compared to the aggregate posterior) for each method, as well as the true fraction of white voters in the voter file. 

Argmax and Threshold rules in particular exacerbate geographic differences: majority-white areas are made even \textit{more} so through the imputation process. Intuitively, this occurs both because of conditional geographic miscalibration (models such as BISG use geography, and so a Black voter in a majority-white area is given a higher white probability than someone with the same name in a majority-Black area) and the discretization process (as discussed above, argmax and thresholding approaches amplify the majority class). Such geographic skews may be especially important in discrimination auditing settings, as they may lead to poor identification of groups in areas where they are especially in the minority and may face the most discrimination. 

As with overall bias, such geography conditional bias can be mitigated with our optimization-based discretization rules, with minimal loss in accuracy. For example, an optimization approach that matches the aggregate posterior \textit{for each county} can exactly match each county's demographic distribution according to the model, with minimal loss in overall accuracy. (We note, however, that matching overall distributions without matching to each county can further polarize labels, where most non-white voters are drawn from counties where such groups live.)

Finally, \Cref{fig:biasmap} also illustrates one benefit of Thompson sampling over all other discretization approaches: due to the sampling process, it naturally balances per-sub-group distributions without explicit optimization. This unbiased property may especially be important in some some auditing or polling settings, despite the individual level inaccuracy. 

\paragraph{Result robustness and replication} The appendix contains replications using public data and models from \citet{greengard2024improved}, who produce models from public voter file data that are calibrated for Caucasian and African Americans -- i.e., the aggregate posterior is correct in terms of the fraction of Caucasian and African American voters. As we theoretically analyze below, we find that even calibrated models lead to substantial discretization bias using standard approaches. We also show robustness of our results to alternative analysis choices, discretization methods, and subsets of the data. All results qualitatively replicate. 

\section{Theoretical characterization of argmax bias and joint decision-making}

Motivated by the empirical results, we now characterize \textit{decision-making} bias, as distinct from \textit{predictive modeling bias}: how does amplification of the most likely class depend on the properties of the discretization procedure and the continuous scores, respectively?

In \Cref{thm:informationargmax}, we show that discretization bias emerges even when the continuous predictive model is itself \textit{unbiased} (calibrated), and is tightly connected to \textit{predictive uncertainty}. 

As intuition, consider the no information regime: if the features provide no additional information over the overall class probability, then even for the Bayes optimal classifier we have $q(y, x_i) = \Pr(y)$ for all $i$, and so all points get assigned label $\hat y_i = z$, the argmax class. On the other extreme of full information, Bayes optimal classifiers induce no argmax bias amplification, as they can identify the correct class for each data point. In the middle, we show that bias amplification is bounded by the mean average error of the continuous classifier. 

In \Cref{thm:pareto}, we analyze \textit{decision-making} approaches to reduce bias: how should discrete decisions $\hat y_i$ be made from a Bayes optimal model of continuous probabilities $q (y, x_i)$? We show that independent decision-making is \textit{necessarily} suboptimal in the presence of a distributional objective -- any such policy, even those designed to meet distributional constraints (e.g., Thompson  sampling), is Pareto dominated in terms of expected accuracy and distributional fidelity. 

Proofs are in \Cref{sec:proofs}.

\subsection{Formal definitions}

In the model, there is a decision-making (classification) setting with discrete classes $y \in \mY$, where $| \mY | = K \geq 2$, with the \textit{prior} class distribution  $\Pr(y)$. There are $N$ data points at the decision (test) time. For each data point $i$, we observe features $x_i \in \mX$, where the data and the unobserved true label $y_i$ are drawn from the joint distribution $(X, Y) \sim F_{XY}$, where $F_X, F_Y$ are both non-constant. We distinguish between \textit{continuous predictions} $q(y, x_i) \in [0, 1]$ and \textit{discrete decisions} $\hat y_i \in \mY$. In our notation, we often omit distribution $F$. 
We also denote a set of data points as $x_{1:N}$.

\subsubsection{Continuous classifier and discretization rules} 

Suppose we have a machine learning classifier $q$ (e.g., trained on some other dataset) that makes continuous probability predictions for each class, given the data. Let $q(y, x)$ denote the probability placed on class $y$ given data $x$. For example, the Bayes optimal classifier is
\[
q^*(y, x) \triangleq \Pr( Y = y | X = x).
\]
We will refer to $p$ as the \textit{prior} and $q$ as our learned \textit{posterior}. Let the simplex of possible probability vectors over the support $\mY$ be $\YSimplex$, and let $q(x) \in \YSimplex$ be the vector of probabilities $(q(y, x))_y$.

Continuous classifiers are evaluated via their predictive performance. The expected \textbf{Mean Absolute Error (\mae)} of the classifier $q$ is:\footnote{Note that, though related, mean average error of the continuous predictor $q(y, x)$ is not equivalent to expected zero-one loss of the discrete decisions. The latter, our measure of decision accuracy, also depends on how the continuous predictor is discretized.} 
\begin{align*}
\mae(q) =\bbE_{F_{XY}}\left[\frac1N\sum_i ( 1 - q(y_i, x_i) )\right] = \int_{\mX} \sum_{y \in \mY} \left( 1 - q(y, x) \right) \Pr(y \given x) f(x) dx.
\end{align*}

Decision (discretization) rules are as defined in \Cref{sec:methodssummary}. Our theoretical results distinguish between \textit{independent} (e.g., argmax, threhsolding, Thompson sampling) and \textit{joint} (e.g., our optimization approach) rules. Formally, with independent rules, there exists a (potentially randomized) function $d$ such that $D_d (\{q(x_i)\})= \{\hat y_i =  d(q(y, x_i))\}$.

\subsubsection{Formal metrics and Pareto optimal decision-making}
For our theoretical results, we consider the expectations of our accuracy, fidelity, and bias metrics, calculated over the joint distribution of the data $F_{XY}$ over datasets of size $N$ (and any randomness in the decision rule). The \textit{expected} accuracy of a decision rule $D$ together with continuous predictor $q$ is thus
\begin{align}
	\eacc_N(D, q) = \bbE_{F_{XY}}\left[ \acc(y_{1:N}, D(\{q(x_i)\})) \right].
\end{align}
Analogously $\efid_N(D, q, p_{\textrm{ref}})$ is the expectation of fidelity and $\ebias_N(y, D, q, p_{\textrm{ref}})$ is the expectation of bias. 

For a given data distribution $F$, continuous classifier $q$, reference distribution $\pref$, and dataset size, we want to make effective decisions. Given the multiple desiderata, we thus want Pareto optimal decision-making, i.e., to adopt a rule $D$ that maximizes, for some $\gamma \in (0, 1]$,
\begin{align}
	O^\gamma_N(D, q, \pref) = \gamma \eacc_N(D, q) + (1 - \gamma) \efid_N(D, q, \pref).
	\label{eq:paretoexpects}\end{align}
The weight $\gamma$ is task-dependent, set by a practitioner's expertise for whether accuracy or distributional fidelity is more important. 

When clear from context, we omit $q$, $\pref$, and $N$ from the arguments for $\ebias$, $O^\gamma$, $\eacc$, and $\efid$. 

\subsection{Theoretical Results: Argmax bias and continuous classifier}

We first study the relationship between argmax bias and the properties of the continuous classifier. We show that discretization bias can occur even with unbiased (calibrated) continuous classifiers $q$, as long as it is uncertain (has non-zero error). 

The following result holds for all calibrated classifiers. A given classifier $q$ is \textbf{calibrated} when its continuous predictions are correct on average, 
$\Pr(Y=y|q(y, x) = c) = c$.\footnote{For all $c$ of non-zero measure: $\forall c \in \left\{c : \left|\int_{(x, y)} \bbI[q(y, x) = c] f(x, y) d(y, x)\right| > 0 \right\}$. Further note that calibration implies Bayesian consistency, i.e., the continuous model cumulatively assigns the proper mass to each class $y$: $\bbE_{X}[q(y, x)] = \Pr(y)$.} Note that the Bayes optimal classifier is calibrated.

\begin{restatable}{theorem}{information} 
\label{thm:informationargmax}
Argmax bias depends on predictive uncertainty, i.e., how much information features $x$ provide about the true class label. Consider calibrated classifier $q$ and the argmax decision rule $D_{\textrm{argmax}}$. Let $N>K$ and consider a reference distribution of either the aggregate posterior or the prior. Then, 

\begin{romanenum}
\item Suppose $x$ provides no information and $q(y, x) = \Pr(y)$ for all $x, y$. Then, amplification bias is maximized and all data points are classified as a plurality class: 
\[
    \ebias(y, D_{\textrm{argmax}}) = 
    \begin{cases}
        1 - \Pr(y), & \text{if } y = \arg\max_w p(w), \\
        -\Pr(y), & \text{otherwise}.
    \end{cases}
\]

\item Suppose the classifier is perfect, i.e., $\forall (x, y) \sim F_{XY}$, we have $q(z, x) =    
\begin{cases}
        1, & \text{if } z = y, \\
        0, & \text{otherwise}.
    \end{cases}$ 
Then, there is no amplification bias, $\ebias(y, D_{\textrm{argmax}}) = 0.$

\item More generally, in the intermediate regime, argmax bias is upper bounded by the predictive error of the 
classifier $q$:
\begin{align}
    \ebias(y, D_{\textrm{argmax}}) & \leq \mae(q). \label{eqn:biasinequality}
\end{align}
The bound is tight: there exist $F_{XY}$ and $q$  such that \Cref{eqn:biasinequality} holds with equality for the plurality class. 
\end{romanenum}
\end{restatable}

The first two parts directly follow from the definitions. If the features provide no information about $Y$, all decisions are identically the most common class; if the classifier is perfect, then each point is labeled according to its true class. The proof for the intermediate regime is more involved: we first construct an example in which bias equals $\mae$,\footnote{For an argmax class $z$, we have $q(z, x) \in \{\frac1K, 1\}$ when $D_{\textrm{argmax}}(q(x)) = z$ and $q(z, x) = 0$ otherwise.} and then show that any setting can be transformed into this example without violating the inequality.   

\begin{figure}[!h]
\centering
\begin{subfigure}{.45\textwidth}
    \centering
    \includegraphics[width=.7\linewidth]{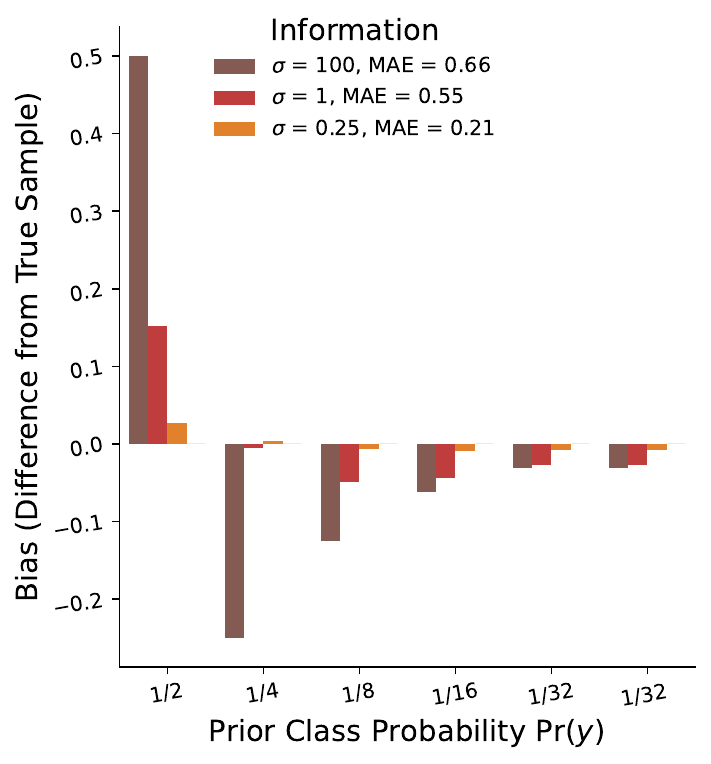}
    \caption{The argmax bias for each class.}
    \label{fig:info-bias}
\end{subfigure}
\hfill
\begin{subfigure}{.45\textwidth}
    \centering
    \includegraphics[width=.95\linewidth]{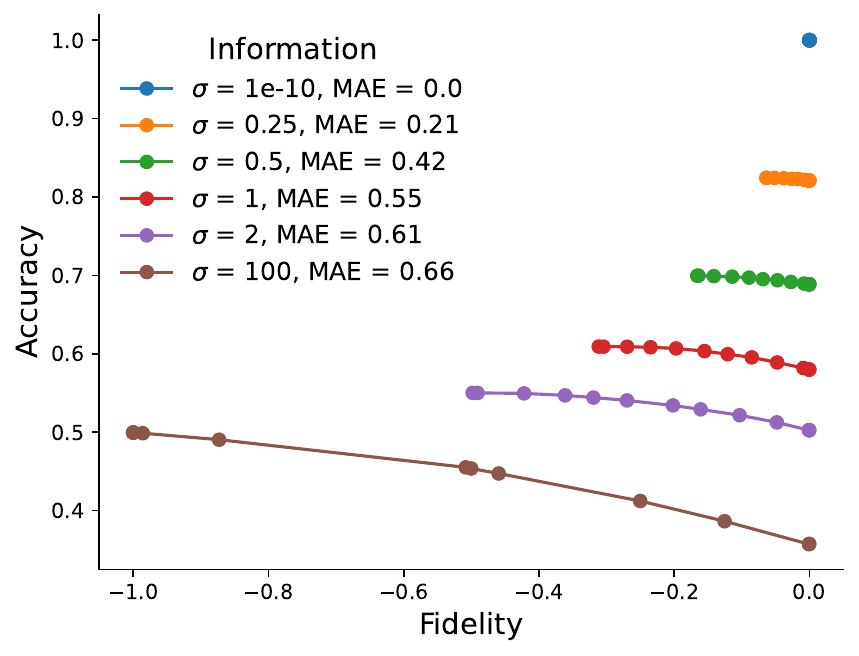}
    \caption{The Pareto curve of optimally-performing decision rules.}
    \label{fig:info-pareto}
\end{subfigure}
\caption{In simulation, discretized performance and argmax bias of the Bayes optimal classifier as a function of model accuracy, labeled by generating parameter $p$ and $\mae$. The results reflect \Cref{thm:informationargmax}: as classifier accuracy increases ($\mae$ decreases), argmax bias decreases, and both accuracy and distribution fidelity of optimal rules increase. Simulation details are described in \Cref{sec:simulationsetup}.}
\label{fig:info}
\end{figure}

We illustrate the main ideas of \Cref{thm:informationargmax} on a simulated dataset: argmax bias (and more generally, the entire Pareto curve) depends on classifier accuracy, the $\mae$.\footnote{Our data simulation process is described in greater detail in \Cref{sec:simulationsetup}; in short, we generate datasets with varying correlations between $X$ and $Y$ and calculate the Bayes optimal posterior $q^*$ in closed form.} 
\Cref{fig:info-bias} shows $\ebias(y, D_{\textrm{argmax}})$ for each class $y$ (labeled with its class probability $\Pr(y)$) as information changes. With little information, the argmax rule (completely) amplifies the highest probability class, and the bias decreases as classifier informativeness increases. Notice that in the lowest-accuracy setting, minority classes are missing from the output distribution entirely. \Cref{fig:info-pareto} shows the Pareto curves (with the argmax rule on one extreme of each Pareto curve). As information and therefore classifier accuracy increases, optimal decision-making improves both decision accuracy and distributional fidelity -- as the uncertainty of $y$ given $x$, as measured by the expected $\mae$, decreases, the Pareto curve shifts upwards.  

This result clarifies the relationship between bias induced by the continuous model versus the decision rule. Unlike the bias amplification observed in works such as \citet{leino2019feature}, argmax bias occurs even for a calibrated classifier, one that places appropriate mass on each class, due to the threshold nature of argmax discrete decisions. We note that the observation that perfect accuracy eliminates bias amplification is in the work of \citet{wang2021directional}, who further state that accuracy and bias are not always correlated; \Cref{thm:informationargmax}(iii) shows that bias is upper bounded by calibrated classifier error and that the bound is tight. 

Overall, \Cref{thm:informationargmax} suggests that improving model performance -- calibrating its predictions, and improving accuracy -- is aligned with mitigating argmax bias. However, it is insufficient: any less-than-perfect continuous model, even if unbiased by itself, can induce biased decisions. Empirically in the Appendix, we observe this in the voter file setting using the data and models of \citet{greengard2024improved}, who produce approximately calibrated continuous predictive models that nevertheless undercount voters of color when discretized.

\subsection{Results: Individual discretization versus joint optimization}

Above, we study how the continuous classifier's properties induce discretization bias, when using the argmax rule. We now study discrete decision-making: given an imperfect but Bayes optimal classifier, how can we make unbiased decisions? We show that joint decision-making, as in the program in \Cref{eq:optdecisionrules}, is necessary for Pareto optimality. 

For the result, we define the notion of \textbf{non-trivial} reference distributions $\mathcal P^q_N$: a reference distribution $p^q : \{x_i\} \to \YSimplex^N$ is in $\mathcal P_N$ if it assigns mass at least $\frac1N$ to the most likely class in the dataset, i.e., at least 1 data point would be discretized to the most likely class $\arg\max_y \sum_iq(y, x_i)$.\footnote{This technical condition eliminates the optimality of independent decision rules that trivially allocate every decision to some arbitrary fixed class, under a reference distribution that requires every point to be allocated to that class regardless of the continuous predictions $q$.} 

\begin{restatable}{theorem}{jointpareto} \label{thm:pareto}
Consider Bayes optimal $q$, and $N > K$. For all $F$: 
\begin{romanenum}
\item For every $\gamma$, $\pref$ there exists a joint decision-making rule that maximizes $O^\gamma(D, q, \pref)$ in \Cref{eq:paretoexpects}.
\item The argmax decision rule $D_{\textrm{argmax}}$ maximizes $O^1(D, q, \pref)$, i.e., is accuracy maximizing.
\item $D_{\textrm{argmax}}$ is the \textit{only} Pareto optimal independent decision rule for non-trivial $\pref \in \mathcal{P}$. No independent decision rule $D$ maximizes objective $O^\gamma$ for any $\gamma$, unless $\gamma = 1$ and $D$ agrees with $D_{\textrm{argmax}}$ with probability 1. 
\end{romanenum}
\end{restatable}

\noindent Parts $(i)$ and $(ii)$ follow directly from the classifier being Bayes optimal---maximizing the objective for each given dataset can be formulated as an integer program, and the argmax rule maximizes expected accuracy for each data point. To prove part (iii), we show that any other independent rule, with positive probability, disagrees with argmax on every point in the dataset -- and so we can improve both accuracy and fidelity by switching the decision for at least one point.

\Cref{fig:pareto-NC} visualizes \Cref{thm:pareto} on the voter file, showing the performance of the various decision rules that all use the same data and model $q$. Independent decision rules are substantially Pareto suboptimal compared to the optimal integer program -- except the argmax rule, which optimizes accuracy at substantial cost of distributional fidelity. Note that in practice, without the Bayes optimal predictor, other decision rules may outperform and strictly dominate argmax, as seen in the peak of the integer program curve in \Cref{fig:pareto-NC}. \Cref{fig:pareto-simulation} shows results fully congruent with \Cref{thm:pareto} on simulated data, in the appendix.

We note that the joint decision-making approaches do not necessarily use \textit{more information} than the argmax or other independent rules. For example, optimizing fidelity to the aggregate posterior simply requires the scores $q(y, x)$---the performance gain comes from jointly making decisions across points. On the other hand, we note that \Cref{thm:pareto}(ii) does not preclude other rules having higher accuracy than the argmax rule when the continuous model is not Bayes optimal, as shown in \Cref{fig:pareto-NC}. More generally, it is possible that, by incorporating domain knowledge such as the expected class distribution when available, that joint decision-rules can correct for model errors.\footnote{For example, in \citet{balachandar2024domain}, external information aids in estimating unbiased continuous predictive models when outcome data is missing.} 

\section{Discussion}

Individual-level demographic imputation is a key algorithmic task across important applications. While much of the literature focuses on improving continuous predictive models, we show the importance of the next stage of discretizing continuous probabilities into individual labels. In particular, we show that the standard approaches of argmax and thresholding substantially skew the distribution of labels, especially as they might correlate with geography, decreasing the fraction of labeled minority groups. We further show that while increasing predictive modeling accuracy reduces (worst-case) bias, it is generally insufficient to eliminate such bias, even with perfectly calibrated models. We develop optimization-based discretization approaches that can match desired class distributions, with negligible loss in accuracy. More generally, we caution against the use of standard discretization approaches, without considering the error properties that one's downstream task requires.

Finally, while our empirical analysis focuses on demographic imputation in voter files, similar considerations appear in many other demographic imputation and algorithmic fairness tasks. For example, matching to the reference distribution could be interpreted as \textit{group fairness} in algorithmic fairness settings, if each label $y$ corresponds to a group, or \textit{individual fairness} in recommendation settings, if each label corresponds to an individual producer or item. Our optimization approach is thus related to the approaches of \citet{10.1145/3460231.3478857} and \citet{zhao-etal-2017-men}, and our work provides a characterization of such approaches, especially supporting joint optimization based approaches over sampling in many applications. We discuss the connection to this literature further in \Cref{sec:related-appendix}.

\paragraph{Limitations and extensions }While we empirically evaluate our optimization-based approach to balance individual level accuracy and distributional fidelity, the approach is flexible to a range of objectives, such as precision, group specific false positives, and the count of individuals allocated to each class. Indeed, we emphasize that the choice of discretization method should be informed by the downstream task for which the labels will be used, and other considerations such as computational scalability and ethical considerations. For example, in some auditing tasks, it may be important to have high accuracy among the classified set (historically used to motivate thresholding approaches), along with geographic representation. In other settings, having the labels reflect the distribution even of \textit{unknown} groups (e.g., that a proportional number of Black individuals who drive a specific type of car are classified as Black) may be more important than individual-level accuracy, suggesting the use of Thompson sampling. In some settings, it is possible to use continuous labels directly or to use self-reported demographic data -- we recommend doing so when possible.

One limitation of our approach -- indeed of much of the demographic imputation literature -- is the use of coarse racial/ethnic categories, and we advocate for methods that better deal with the full complexity of identity, while respecting privacy considerations. However, the use of more granular categories would exacerbate argmax bias (since there would be more, smaller probability classes), increasing the importance of discretization approaches beyond argmax or thresholding.

\section{Methodological Details} \label{sec:methodsdetails}

\begin{description}
\item [Argmax]
The commercial dataset's presupplied discrete labels did not appear to break ties in any totally consistent order. As a result, when labeling \textit{Uncoded} data points (approximately $3.23\%$ of the dataset) with argmax, we broke ties in the following order:  \textit{white} $\succ$ \textit{African American} $\succ$ \textit{Hispanic} $\succ$ \textit{Asian}$\succ$ \textit{Native American}, which corresponds to both overall population size and is the least inconsistent with the provided labels. Tiebreaking has a negligible effect on results; there were a total of $\num{390}$ such data points.

\item [Integer optimization]
We solved integer programs using cvxpy \citep{diamond2016cvxpy}, an open-source Python package, and Gurobi \citep{gurobi}, used with a free academic license. In both simulation and empirical analyses, we calculated $\gamma$ values between $0.8$ and $0.99$. \Cref{fig:pareto-NC} uses increments of $0.01$.

\textit{Computational considerations in joint rules} We note that solving the optimization problem jointly over the full dataset may be computationally expensive; \citet{10.1145/3460231.3478857} introduce a Lagrangian-based approach. In our empirical application, we solve this optimization problem in evenly sized \textit{batches} of approximately $\num{10000}$. When using the aggregate posterior conditioned upon geographical location, these batches are grouped by county instead of completely random selection.

\item [Aggregate posterior matching]
Matching to an exact distribution (i.e., the fidelity-maximizing one) becomes an assignment problem, solvable in polynomial time as a maximum weight matching problem on a bipartite graph, as in \citet{10.1145/3539597.3570402}. We use a bipartite max weight algorithm in Python to solve this problem, rounding the distribution to the closest integer number of labels to the fidelity-maximizing distribution. 

\item [Data-driven thresholding heuristic]
We explored several different machine learning models to approximate aggregate posterior matching. Logistic regression, random forests, and support vector machines all achieve similar performance given one $N\approx\num{10000}$ training batch, with $\approx99\%$ of predicted labels being identical to those of the matching solution. Our reported results use a linear SVM, approximating both the aggregate posterior matching solution (\Cref{tab:table-NC} and \Cref{fig:pareto-NC}) and integer programs (\Cref{fig:pareto-NC}).

\end{description}

We further evaluate two rules in the Pareto curve diagram \Cref{fig:pareto-NC} and the appendix.

\begin{description}

\item [Top-\textit{k} sampling]
Top-\textit{k} sampling is used in many modern applications, including large language models \citep{holtzman2020curious}, and is equivalent to Thompson sampling between just the (renormalized) probabilities of the top $k$ classes, with the remaining classes set to probability 0. 

\item [True marginal population matching]
In some settings, \textit{true} aggregate population-level information may be available or well-estimated (e.g. through the census or surveys), or some other reference distribution may be chosen given normative motivations. We perform matching using the true makeup of self-reported race to demonstrate the effectiveness of our approach in utilizing additional information.  

\end{description}

\textit{Synthetic data generation.} \label{sec:simulationsetup}
Our synthetic data setting for \Cref{fig:info,fig:pareto-simulation} is designed to illustrate our main ideas. We use $K=6$ classes (i.e., $\mY = \{1 \dots 6\})$, with the prior frequencies being sequential negative powers of 2, with the final two classes having equal frequency (i.e., $\{1/2, 1/4, 1/8, 1/16, 1/32, 1/32\})$. 

We simulate the features $x$ as a $K$ dimensional vector; each dimension is drawn from a Normal distribution with variance $\sigma^2$, where the dimension corresponding to the true $y$ has mean $1$ and the remaining dimensions have mean $0$. Thus the variance $\sigma^2$ characterizes the informativeness of the data (lower $\sigma$ means lower predictor $\mae$ of the Bayes optimal classifier $q^*(y, x) = \Pr(y | x)$). {We note that this choice makes it easy to calculate the Bayes optimal classifier and to sample data. Data are generated by first sampling $y_{1:5000}$ from the prior, and then sampling $x_{1:5000}$ with the $\Pr(x|y)$ probabilities above.} We draw $N=5000$ points for each dataset. We calculate the Bayesian optimal posterior from the generating parameters and average results across 100 datasets. 

\newpage
\FloatBarrier
\bibliography{bib}

\clearpage
\appendix

\section{Extended Related Work} \label{sec:related-appendix}

While we focus on demographic imputation as our primary application, the observation that algorithmic systems tend to \textit{amplify} the most likely class has reappeared with different names across computational fields several times, with several proposed solutions. To this literature, we contribute an analysis of how different discretization approaches have different error properties, suggesting the use of optimization-based approaches tuned to one's specific desiderata. 

\subsection{Argmax bias, calibration, and bias amplification across fields}  

Our work generally connects to work within algorithmic fairness and related fields. For example, if each class $y$ represents a (demographic) group, then distributional fidelity could represent a group fairness notion, with $p_{\textrm{ref}}(y) = \frac1K$ for all $y$, or corresponding to the class probabilities $p_{\textrm{ref}}(y) = \Pr(y)$.

\citet{twitterargmax} term ``argmax bias'' and study whether it affects image cropping algorithms used by Twitter (prioritizing male faces in image previews). In image labeling tasks, \citet{zhao-etal-2017-men} find that standard algorithms tend to over-label an agent as a \textit{woman} when they are \textit{cooking} in an image, with similar gender ``bias amplification'' occurring for other verbs such as \textit{shopping} and \textit{coaching}.  For example, \citet{zhao-etal-2017-men} find that when the fraction of training set images that are of women cooking $\Pr(\text{women})$ is greater than that of men cooking, the fraction of times an image labeler labels a person who is cooking as a woman is even \textit{greater}, $\Pr(\hat y_i = \text{women}) > \Pr(\text{women})$. 
Others then refine and expand on such bias \cite{wang2021directional}. Similarly, \citet{stanovsky-etal-2019-evaluating} find that standard machine translation algorithms amplify gender stereotypes (e.g., assign a gender-neutral \textit{doctor} as \textit{man} when translating to a language with gendered nouns). In recommendation systems, \citet{steck2018calibrated} proposes the notion of ``calibration,'' to counter the phenomenon that a user who prefers horror movies 70\% of the time and comedy movies 30\% of the time would {only} be recommended horror movies (this same notion of calibration is also present in algorithmic fairness settings). In biological engineering, deterministic binning of measured fluorescence systematically biases population estimates \cite{trippe2022randomized}. The fundamental phenomenon across settings is that absent sufficient context-specific information for a specific data point, learning systems that optimize for accuracy will default to the globally most likely outcome. Many papers study related biases in the various contexts \cite{peng2023reconciling,jacobs2021measurement,guo2021stereotyping,garg2021standardized,costanza2022audits,Birhane_2022_WACV,hall2022systematic,mansoury2020feedback,pmlr-v202-zhao23a,pmlr-v202-taori23a,jia2020mitigating,jagadeesan2023improved,jagadeesan2023supplyside,garg2018word}; see also Arvind Narayanan's summary in a Twitter thread \citep{narayanan2021tweet} on argmax bias. 

Given this widespread identification of the problem (and more generally, a focus on the algorithmic fairness notion of \textit{group fairness}), there has been a large literature studying its causes and proposing a range of solutions.

\textit{The role of the continuous model $q$.} The predominant focus in algorithmic fairness, like in demographic imputation, has been on the continuous predictive model $q (y, x) \approx \Pr(Y | X)$ and its inputs \cite{black2023toward}. The predictive model itself can be biased, i.e., for the argmax class $z$, we have that $q(z, x) > \Pr( Y = z | X = x)$. \citet{leino2019feature} show that bias amplification in the continuous modeling stage can be caused by inductive biases of stochastic gradient descent; they propose {feature selection} pre-processing as a solution. More generally, a large literature in fair machine learning aims to change continuous predictions $q(y, x)$ \cite{caton2020fairness}, such as by constrained optimization approaches during training (in-processing) or resampling training data (pre-processing). \citet{chen2018my} argue that ``unfairness induced by inadequate samples sizes or unmeasured predictive variables should be addressed through data collection, rather than by constraining the model,'' and \citet{cai2020fair} and \citet{10.1145/3306618.3314277} propose active information acquisition approaches to counter unfairness. 

\textit{The role of discrete decision-making.} Bias amplification can also occur in the discrete decision-making stage, as the term \textit{argmax bias} \cite{twitterargmax} suggests -- as discussed above, even if the predictive model is optimal, i.e., $q (y, x) = \Pr(Y | X)$, argmax discretization can amplify the most common class. Recognizing this cause, another set of technical approaches is to instead intervene at the \textit{decision point}, when going from a continuous score $q(y | x)$ to a discrete label or decision $\hat y$. \footnote{We note that there are also non-technical solutions. One especially relevant approach highlighted in the context of argmax bias is to delay algorithmic decisions or delegate to a human: if the problem is caused by discretizing continuous probabilities $\Pr(Y | X)$ into discrete decisions $\hat y$, then it can be avoided by not making discrete decisions. In the Twitter image cropping case, for example, a proposed solution was to allow users to manually crop images for previews \cite{twitterargmax}. As \citet{narayanan2021tweet} summarizes in the twitter thread: 
\begin{quote}
``\textit{I don't know any effective technical fixes for the argmax issue. But a really good approach for handling model uncertainty is to show the user the possible outputs and ask them to choose. But this goes against the mantra of frictionless design and so it's very rarely deployed.}''
\end{quote}
When possible, we believe that such an approach is effective and appropriate, as emphasized by \citet{twitterargmax}. However, we also believe that technical interventions at the point of decision-making under uncertainty are often necessary, and so should be improved.} 

We distinguish these technical solutions into two classes, \textit{independent} and \textit{joint}.

\begin{description}[leftmargin=.3em]
  
    \item [Independent decision-making] The first approach class is to change the mapping from $q(y, x_i)$ to $\hat y_i$, \textit{independently} for each data point $i$. By far the most commonly proposed technical fix in the algorithmic fairness community at the decision point is to \textit{sample} decisions -- for each data point $x$, probabilistically assign it a decision $y$ as a function of the continuous score $q(y, x)$. This approach can eliminate argmax bias -- each class occurs in the correct proportions -- but there may be substantial inaccuracy: a class with a posterior probability of only 10\% would be the assigned decision for about 10\% of the data points. Sampling is proposed prominently in the Twitter image cropping work that termed argmax bias \cite{twitterargmax}, and is also proposed as a solution to such biases by \citet{narayanan2021tweet}, in natural language processing \cite{hendricks2018women}, biological engineering \cite{trippe2022randomized}, in admission settings with heterogeneous information \cite{liu2021test}, and algorithmic fairness generally \cite{caton2020fairness}. 
    
    Despite sampling's popularity, there is a sense it is inadequate. As \citet {10.1145/3306618.3314277} highlight when proposing an information acquisition approach in algorithmic fairness, 
    \begin{quote}
    ``\textit{Recent work has proposed optimal post-processing methods that randomize classification decisions for a fraction of individuals, to achieve fairness measures related to parity in errors and calibration. These methods, however, have raised concern due to the information inefficiency, intra-group unfairness, and Pareto sub-optimality'' 
    }    
    \end{quote}

Another independent approach involves setting fixed thresholds and only discretizing a data point if the class probability is sufficiently high; if a data point does not meet any threshold, it is excluded
\citep{adjaye2014using,bureau2014using,zhang2018assessing,chen2019fairness}. This approach is often used for demographic imputation; \citet{adjaye2014using} pick different thresholds per group in a health care context. 

\item [Joint decision-making] An alternative approach is to \textit{jointly} make decisions across the entire dataset. \citet{10.1145/3460231.3478857} (for calibrated recommendations) and \citet{zhao-etal-2017-men} (to reduce bias amplification in NLP) formulate similar integer optimization problems that jointly take scores $\{q (y,x_i)\}$ and output labels that balance maximizing individual-level accuracy and matching a class distribution. \citet{10.1145/3539597.3570402} formulate the extreme case of prioritizing matching the class distribution as an efficient max flow problem. As opposed to sampling approaches (and independent approaches more generally), joint decision-making is rarely even mentioned as a solution to argmax bias or in algorithmic fairness broadly -- for example, it does not appear in a table of six solutions by \citet{twitterargmax}, or as a solution approach in a survey of fair machine learning \cite{caton2020fairness}, though sampling and thresholding approaches are extensively detailed. 
\end{description}

\vspace{.25em}
\textit{Our contribution to this literature.} What does our work contribute to this extensive literature? (1) We analyze decision-making bias as distinct from and in relation to predictive modeling biases, requiring decision-making interventions: decision-making bias amplification occurs even with Bayes optimal, unbiased predictive modeling unless predictions are perfect ($\Pr(Y = y \given X = x) = 1$ for the true class $y$, for each $x$). Thus, improving continuous predictions such as via active information acquisition is useful -- improving accuracy reduces (worst case) bias when using the argmax decision rule -- but it is generally insufficient. (2) Within the set of decision-making interventions, our work suggests a substantial departure from the algorithmic fairness status quo away from sampling based approach, and we suggest that optimization based approaches are relatively underused. More generally, we urge distinguishing between \textit{predictive modeling} and \textit{decision-making} biases in both auditing and bias mitigation, and urge more intentional choice of the discretization procedure: discretization is ultimately a highly context dependent decision, for which tradeoffs such as importance of accuracy, bias, computational scalability, other metrics, and non-quantitative concerns should be considered. 

\subsection{Demographic imputation}

Individual-level race and ethnicity data are used across applications, including public health, algorithmic fairness, political science, criminal justice, and economics -- especially for disparity auditing \cite{pierson2020large,chin2023methods,ghosh2021fair,chen2019fairness,deluca2022validating,grumbach2020race,diamond2019effects,garg2018word}. However, collecting ground truth self-reported data may face practical or legal barriers \citep{bureau2014using}. As a result, several methods have been developed to impute race/ethnicity using proxy variables, often name and geographic location \citep{fiscella2006use,elliott2009using,voicu2018using,imai2022addressing,chintalapati2023predicting}. 
For example, Bayesian Improved Surname Geocoding (BISG) \citep{elliott2009using}, a widely used method, was developed to estimate and assess racial disparities. Notably, the US Consumer Financial Protection Bureau uses BISG as a proxy to detect discriminatory lending practices, as creditors are generally prohibited from collecting race and ethnicity information \citep{bureau2014using,andrus2021we,chen2019fairness}. By and large, these demographic prediction methods are probabilistic classifiers that first produce continuous predictions; then, these continuous scores are discretized, either through an argmax rule or approximately so (as we verify in our commercial voter file), with probability thresholds \citep{bureau2014using,adjaye2014using,chen2019fairness,diamond2019effects}.   

Our empirical case study is in the context of voter registration data files, which form a backbone of public opinion and political science research for both academic and electoral ends; see \citet{ghitza2020voter} for an overview. Since not every US state makes public individual-level self-reported race/ethnicity information, commercially sold datasets come with predicted demographic data for each registered voter, generated through proprietary models \citep{ansolabehere2012validation}; these imputed demographic data are often provided both as continuous probabilities and discrete single labels \citep{ghitza2020voter}.

Many academic analyses (and from our experience, political campaign practitioners) use a voter file's discrete race and ethnicity labels \citep{fraga2016candidates,fraga2016redistricting,hersh2016primacy}. Similarly, many academic papers beyond voter file analyses use discrete labels, instead of continuous probability scores \citep{bureau2014using,adjaye2014using,diamond2019effects,zhang2018assessing,grumbach2020race,ghosh2021fair,pierson2020large}. We note that discretization is not strictly necessary in all applications; \citet{chen2019fairness} and \citet{mccartan2023estimating} present methods that utilize probabilistic outputs for disparity estimation, and \citet{deluca2022validating} sum probabilities directly when estimating the overall makeup of a population.  

\textit{Our contribution to this literature.} Our empirical results in \Cref{sec:empirics} show that the discretization process (typically argmax or closely related threshold rules) leads to substantial bias in a commonly used commercial voter file, leading to a severe undercounting of voters of color -- on top of the undercounting caused by continuous model miscalibration. We further find that a joint decision-making approach can eliminate discretization bias with negligible loss in individual-level accuracy. We thus caution against the use of existing discrete labels for sensitive applications where such bias would affect results.   Instead, we recommend that (a) data vendors produce discretized labels using a joint decision-making approach, as we do; (b) when possible, consumers of voter files directly use the \textit{continuous} scores instead of single discrete labels \cite{chen2019fairness,mccartan2023estimating,deluca2022validating}, or reconsider the use of imputed demographic data altogether. 

\section{Theoretical analysis}
\label{sec:proofs}

\subsection{Proofs}

Recall that a classifier $q$ is \textbf{calibrated} when its continuous predictions are correct on average, 
$\Pr(Y=y|q(y, x) = c) = c$.\footnote{For all $c$ of non-zero measure: $\forall c \in \left\{c : \left|\int_{(x, y)} \bbI[q(y, x) = c] f(x, y) d(y, x)\right| > 0 \right\}$.} Calibration implies Bayesian consistency, i.e., the continuous model cumulatively assigns the proper mass to each class $y$: $\bbE_{X}[q(y, x)] = \Pr(y)$.

\information*
\begin{proof}

\noindent\textbf{Proof for Part (i)}. The result immediately follows by definition of argmax. For the argmax class $z$, all points are classified to it and so $\hat p_{\textrm{marg}}(\cdot, z) = 1$, and so $\bias(z, \cdot, \cdot) = 1 - \pref(z)$ for any $\{x_i\}$. 

~\\\noindent\textbf{Proof for Part (ii)}. Similarly, the result immediately follows from the definition of the aggregate posterior or the prior. Due to perfect prediction, we have that $\bbE\left[\hat p_{\textrm{marg}}(\cdot, y)\right] = \Pr(y)$, for all $y$. With a consistent classifier (implied by our notion of calibration), we have that 
\begin{align*}
    \bbE[\pagg(y, \{x_i\})] &\triangleq \bbE\left[\frac1N \sum_i q(y, x_i)\right]
    = \Pr(y).
\end{align*}
~\\\noindent\textbf{Proof for Part (iii)}. The intermediate regime proof is more involved. For any continuous classifier $q$, denote the sets $S_a, S_b, S_c, S_d, S_e \subseteq \mX$ as
\[x \in   
\begin{cases}
        S_a & \text{if } q(z, x) = 1, \\
        S_b & \text{if } q(z, x) \in (\frac{1}{K}, 1)  \\
        S_c & \text{if } q(z, x) = \frac{1}{K} , \\
        S_d & \text{if } q(z, x) \in (0, \frac{1}{K})  \\
        S_e & \text{if } q(z, x) = 0 
    \end{cases} 
\] 
with probability mass $\ell \triangleq \int_{x \in X_\ell} f(x) dx$. Since the sets partition the entire probability space, we have $a+b+c+d+e=1$. 

For any class $y$, its amplification bias $\bias(y)$ is maximized when the most data points have $y$ has its decision, and so we suppose the argmax decision rule breaks ties in favor of $z$, i.e., the decision rule outputs $z$ when $q(z, x) \geq q(y', x)$ for all $y' \in \mY$. 

~\\\noindent \textbf{We start with an example such that \Cref{eqn:biasinequality} holds with equality. } Consider $q$ such that $b, d = 0$. Further suppose that for $x \in S_c$, we further have that $q(y, x) = \frac1K$, for all $y$. For $x \in S_e$, assume that there exists another class $y' \neq z$ such that $q(y', x) = 1$. 

Recall that, by assumption, $q$ is consistent, i.e.,
\begin{align}
     p(z) &= \bbE_{X}[q(z, x)] =  \int_{x \in S_a} 1 f(x) dx + \int_{x \in S_c} \frac1K f(x) dx + \int_{x \in S_e} 0 f(x) dx \\
     &= a + \frac{c}{K}.
\end{align} 

The key step is realizing that the \textit{calibration} requirement on $q$ gives us a bound for \mae, and relates bias to classifier accuracy. Recall that calibration means that 
\begin{align}
    & \bbE_{XY}[Y=y|q(y, x) = \delta] = \delta
\end{align}
which implies that $\bbE_{XY}[Y=z|x \in S_a] = 1$, and $\bbE_{XY}[Y=z|x \in S_e] = 0$, and $\bbE_{XY}[Y=z|x \in S_c] = \frac{1}{K}$. 

Equivalently, that $\int_{x\in S_a} f(x, z) = a$ and $\int_{x\in S_a} \sum_{y \neq z}f(x, y) = 0$. Similarly, $\int_{x\in S_e} f(x, z) = 0$, and $\int_{x\in S_c} f(x, z) = \frac{c}{K}$. 

Then, we have
\begin{align}
    \mae =& \int_{x}\sum_{y \in Y} f(x,y) |1-q(y, x)| dx\\
    =& \int_{x}\left[\left[f(x,z) |1-q(z, x)| \right] + \left[\sum_{y \neq z} f(x,y) |1-q(y, x)| \right]\right] dx\\
    =& \int_{x \in S_a}\left[\left[f(x,z) 0 \right] + \left[\sum_{y \neq z} f(x,y) 1 \right]\right] & \text{defns of } S_\ell\\
    &+ \int_{x \in S_e}\left[\left[f(x,z) 1 \right] + \left[\sum_{y \notin \{z, y'\}} f(x,y) 1 \right] + \left[ f(x,y') 0 \right]\right] \\
    &+ \int_{x \in S_c}\left[\left[f(x,z) \frac{K-1}{K} \right] + \left[\sum_{y \neq z} f(x,y) \frac{K-1}{K} \right]\right] \\
    =& 0 + 0 + c\left[\frac{K-1}{K}\right] & \text{calibration implications} \label{eqnpart:usecalibration}
\end{align}

Note that argmax (breaking ties in favor of $z$) assigns label $z$ for all $x \in S_a, S_c$. Thus, we have
\begin{equation}
p_\mathrm{argmax}(z) = a + c.    
\end{equation}
Putting things together, we have
\begin{align*}
    \bias(z) &= p_\mathrm{argmax}(z) - p(z) \\
    &= a + c - (a + \frac{c}{K})  = c \frac{K-1}{K}\\
     &= \mae. 
\end{align*}

\noindent \textbf{Proof that the above is also an upper bound in \Cref{eqn:biasinequality}, that that is the worst case example. }
Now we prove that this bound is the worst case for any calibrated (and thus consistent) $q_0$. The framework for the proof is, for any given setting (joint distribution $F_0$, calibrated and consistent $q_0$) and fixed class $z$, to transform the setting into the worst-case example above, where each transformation: (a) maintains the original marginal distribution $\Pr(y)$ of $F$; (b) either maintains or increases the mass of points $x$ assigned to class $z$ (and thus either maintaining or increasing \bias); and (c) does not increase \mae. Then, since the worst-case example satisfies the inequality, we have shown that the original example does as well. 

Our proof below is notationally simpler with continuous $x \in \mX$, such that there is no point mass ($F(\{x\}) = 0$, for all $x$). Thus, for discrete $\mX$, we first transform it by adding a continuous feature dimension, where the corresponding feature $x_{\text{new dimension}} \in [0, 1]$ is distributed uniform at random, independent of the other features and $y$. 

Recall that we assume that ties are broken consistently, with some ordering over the classes $y \in \mY$. Thus, we can consider the class $z$ such that any ties with any other class are broken in its favor: this class has the worst-case bias; for example for $F_0, q_0$ such that another class exhibits more bias, we can relabel the points (and so $\mae$ is fixed) such that $z$ has the largest bias, at least as large as the original example. Thus, proving the result for $z$ is sufficient.

We start with the case for $K = 2$, i.e., binary decision-making, and then show how the proof extends generally. The idea is that for any $x \in S_b$ (i.e., $q_0(z, x) \in (\frac{1}{K}, 1)$), we can transform the example such that $x$ becomes in either $S_a$ or $S_c$ while maintaining the original class distribution $\Pr(y)$ and the fraction of points discretized to $z$, and preserving $\mae$. Similarly, for any $x \in S_d$, we can move $x$ to either $S_c$ or $S_e$ while maintaining consistency, increasing bias, and preserving \mae. 

Step 1: Suppose mass $b$ of $S_b$ is positive, $b > 0$. Then, we move a set $S_{b \to a}$ to $S_a$ and a set $S_{b \to c}$ to $S_c$, respectively, maintaining consistency, where $S_{b \to a}$, $S_{b \to c}$ are a partition of $S_b$. In the binary example, we know that all $x$ are still discretized to $z$, since $q(z, x) > 1/2$ (and will remain so), and so $p_\mathrm{argmax}$ is fixed. Recall that consistency requires $p(z) = \bbE_{X}[q(z, x)] =  \int_{x \in \mX} q(z, x) f(x) dx$. We can choose sets $S_{b \to a}$, $S_{b \to c}$ to maintain consistency and the class distribution of $z$ by satisfying the following:
\begin{align}
     \int_{x \in S_b} q_0(z, x) f(x) dx &= \int_{x \in S_{b \to a}} 1 f(x) dx + \int_{x \in S_{b \to c}} \frac1K f(x) dx\\
    \iff \int_{x \in S_{b \to a}} q_0(z, x) f(x) dx + \int_{x \in S_{b \to c}} q_0(z, x) f(x) dx &= \int_{x \in S_{b \to a}} 1 f(x) dx + \int_{x \in S_{b \to c}} \frac1K f(x) dx \\
     \iff  \int_{x \in S_{b \to c}} (q_0(z, x) - \frac1K) f(x) dx &= \int_{x \in S_{b \to a}} (1 - q_0(z, x)) f(x) dx \label{eqn:maintainconsistency}
\end{align} 
Note that such sets exist because $F_0$ is a continuous distribution in $\mX$, and so expanding the set $S_{b \to c}$ monotonically and continuously increases the left-hand side, while monotonically and continuously decreasing the right-hand side. 

Next, we claim that this transformation preserves \mae. As we used in \Cref{eqnpart:usecalibration}, calibration implies that the transformed error $\mae_1$ within the subset $S_{b \to a}$ is $0$, and $\mae_1$ within the subset $S_{b \to c}$ is $F(S_{b \to c})\frac{K-1}{K}$. Recall that $\mae = \sum_{y\in\mY} \int_{x \in \mX} f(x,y) (1-q(y, x)) dx$. Then, 
\begin{align}
    \mae_1 &- \mae_0\\
    &= \sum_{y\in\mY} \left[\left[\int_{S_{b \to c}} f(x,y) \left[(1 - \frac1K) - (1-q_0(y, x))\right] dx\right] + \left[\int_{S_{b \to a}} f(x,y) \left[(1 - 1) - (1-q_0(y, x))\right] dx\right]\right]  \\       
    &= \sum_{y\in\mY} \left[\left[\int_{S_{b \to c}} f(x,y) \left[q_0(y, x) - \frac1K\right] dx\right] - \left[\int_{S_{b \to a}} f(x,y) \left[1-q_0(y, x)\right] dx\right]\right]  \label{eqnpart:provemaeconstant2part}\\       
    &= 0    
\end{align}
Where the last line follows from \Cref{eqn:maintainconsistency} (the sets were chosen to maintain consistency, and this also maintains \mae). Note that \Cref{eqn:maintainconsistency} only has this condition for the focal class $z$, whereas \Cref{eqnpart:provemaeconstant2part} sums over classes. Note that, in the $K=2$ example, maintaining consistency (and class distribution) of one class automatically maintains it for the other class, and that the sets $\{x: q(z, x) \in (1/K, 1)\}$ and  $\{x: q(y, x) \in (0, 1/K)\}$, $y \neq z$ are equivalent. 

Thus $\mae_1 = \mae_0$ and $\bias_1 = \bias_0$.

~\\Step 2: Similarly, now suppose mass $d$ of $S_d$ is positive, $d > 0$. Then, we move a set $S_{d \to c}$ to $S_c$ and a set $S_{d \to e}$ to $S_e$, respectively, maintaining consistency and $\mae$ while \textit{increasing} bias (in $S_d$, points are not discretized to $z$, but they are in $S_{d \to c}$ under the transformed example). 

As in Step 1, choose the sets to maintain consistency and the class distribution of $z$ by:
\begin{align*}
    \int_{x \in S_{d \to e}} (q_1(z, x) - 0) f(x) dx &= \int_{x \in S_{d \to c}} (1/K - q_1(z, x)) f(x) dx
\end{align*}
And then we have
\begin{align}
    \mae_2 &- \mae_1\\
    &= \sum_{y\in\mY} \left[\left[\int_{S_{d \to e}} f(x,y) \left[(1 - 0) - (1-q_1(y, x))\right] dx\right] + \left[\int_{S_{d \to c}} f(x,y) \left[(1 - \frac{1}{K}) - (1-q_1(y, x))\right] dx\right]\right]  \\       
    &= \sum_{y\in\mY} \left[\left[\int_{S_{d \to e}} f(x,y) \left[q_1(y, x)\right] dx\right] - \left[\int_{S_{d \to c}} f(x,y) \left[\frac{1}{K} - q_1(y, x)\right] dx\right]\right]  \\       
    &= 0,
\end{align}
and so $\bias_2 \geq \bias_1$ and $\mae_2 = \mae_1$.

Thus, we have
\begin{align*}
    \bias_0 &= \bias_1 \leq \bias_2 \\
    &= \mae_2 & \text{Equivalent to example for worst case where showed equality} \\
    &= \mae_1 = \mae_0.
\end{align*}
~\\\noindent Finally, we must show that the above argument holds for $K \geq 2$. For the $K=2$ case, the notation is simpler because transforming $q(z, x)$ also symmetrically transforms $q(y, x)$ for the class $y \neq z$, and that the sets $\{x: q(z, x) \in (1/K, 1)\}$ and  $\{x: q(y, x) \in (0, 1/K)\}$ are equivalent. Thus, if $x \in S_{b \to c}$, then, simultaneously, we were moving $q(y, x)$ from less than $1/2$ to exactly $1/2$. 

This symmetry no longer holds for $z$ and a \textit{fixed} $y \neq z$ because these factors depend on other classes as well (one can increase $q(z, x)$ while holding fixed $q(y, x)$). However, we show that we can still construct a sequence of transformations as follows.

While the mass of set $S_b$ is more than zero, note that for each $x \in S_b$, there exists \textit{at least one} class $y' \neq z$ such that $q(y', x) \in (0, 1/K)$. In the general case, then, we can construct sets $S_{b \to a}$, $S_{b \to c}$ where we change $q(z, x)$ as before while, simultaneously, moving $q(y', x)$ toward $0$ or $1/K$, to maintain consistency and $\sum_y q(y, x) = 1$ for all $x$. Note that the corresponding transformations all maintain $\mae$ and $\bias(z)$, while potentially increasing $\bias(y')$. Similarly, while the mass of set $S_d$ is non-zero, there exists  \textit{at least one} class $y' \neq z$ such that $q(y', x) \in (1/K, 1)$, and one can make the parallel transformations. 

Thus, for any $K$, these transformations maintain $\mae$ without decreasing $\bias$, and transform the setting to the worst-case example $F, q$, proving the result.
\end{proof}

\jointpareto*
\begin{proof} 
\noindent\textbf{Proof for Part (i)} The result immediately holds for joint decision-making rules, by construction. Consider the multi-objective optimization decision rules defined in \eqref{eq:optdecisionrules}: 

\begin{equation*}
D^\gamma(\{q(x_i)\}, p_{\textrm{ref}}) = {\arg\max}_{\{\hat y_i\}} \gamma \left(\frac{1}{N}\sum^N_{i=1} q(\hat y_i, x_i)\right) + (1-\gamma) \fid(p_{\textrm{ref}},\hat y_{1:N}).
\end{equation*}

Now, note that the objective can be decomposed using the tower property, as follows:
\begin{align*}
O^\gamma_N(D, q, \pref) &= \gamma \eacc_N(D, q) + (1 - \gamma) \efid_N(D, q, \pref) \\
&=  \bbE_{F_{\{X\}}} \left[ \bbE_{F_{\{Y\} \given \{X\}}} \left[\gamma \acc(y_i, D(\{q(x_i)\}))  + (1 - \gamma) \fid(\pref, D(\{q(x_i)\})) \given \{X_i\} = \{x_i\} \right]\right] \\
&=  \bbE_{F_{\{X\}}} \left[ \bbE_{F_{\{Y\} \given \{X\}}} \left[\gamma \frac1N\sum_i^N\mathbb{I}\left[D(\{q(x_i)\}) = y_i\right]  + (1 - \gamma) \fid(\pref, D(\{q(x_i)\})) \given \{X_i\} = \{x_i\} \right]\right] \\
&=  \bbE_{F_{\{X\}}} \left[ \bbE_{F_{\{Y\} \given \{X\}}} \left[\gamma \frac1N\sum_i^N q(D(\{q(x_j)\})_i, x_i)  + (1 - \gamma) \fid(\pref, D(\{q(x_i)\})) \given \{X_i\} = \{x_i\} \right]\right]
\end{align*}

Where the last line follows from, given Bayes optimal $q$, we have for all $\hat y_i$ 
\begin{align*}
    q(\hat y_i, x_i) &=  \bbE_{F_{Y \given X}}[ Y = \hat y_i | X = x_i].
\end{align*}
Thus, the optimization decision rule maximizes the inner expectation in the final line, for each given dataset $\{x_i\}$. Thus, the corresponding rule maximizes $O^\gamma_N(D, q, \pref)$, as desired. 

~\\\noindent\textbf{Proof for Part (ii)}

The result follows directly from the classifier being Bayes optimal -- for each row, you cannot do better than picking the most likely class for that row. 

Recall that $D(\{q(x_i)\})$ refers to the set of decisions $\hat y_{1:N}$ with dataset $x_{1:N}$, where the decisions can be jointly made across data points. Note that, for the Bayes optimal classifier $q^*$, we have
\begin{align*}
    \bbE_{F_{Y \given X}}\left[\mathbb{I}\left[  Y = y\right] \given X = x \right] &= q^*(y, x)
\end{align*}
Thus, we have:
\begin{align}
O_{N}^1(D, q^*, \cdot) &= \eacc_N(D, q^*) \\
&= \bbE_F\left[ \acc(y_{1:N}, D(\{q^*(x_i)\})) \right] \\
&= \bbE_F\left[ \frac1N \sum_i \mathbb{I}\left[ \hat y_i = y\right] \right] & \text{let } \hat y_i \triangleq D(\{q^*(x_j)\}_{j=1}^N)_i\\
&= \bbE_{F_{\{X\}}}\left[\bbE_{F_{\{Y\} | \{X\}}}\left[ \frac1N \sum_i \mathbb{I}\left[ Y_i = \hat y_i\right] \given \{X_i\} = \{x_i\} \right] \right] & \text{tower property} \\
&= \bbE_{F_{\{X\}}}\left[\frac1N \sum_i \bbE_{F_{Y | X}}\left[ \mathbb{I}\left[ Y = \hat y_i\right] \given X = x_i \right] \right] & \text{independent sampling of datapoints} \\
&= \bbE_{F_{\{X\}}}\left[\frac1N \sum_i q^*(\hat y_i, x_i) \right] \\
&\leq \bbE_{F_{\{X\}}}\left[\frac1N \sum_iq^*(\arg\max_y q^*(y, x_i), x_i) \right] \\
&\triangleq O_{N}^1(D_{\textrm{argmax}}, q^*, \cdot).
\end{align}
~\\\noindent\textbf{Proof for Part (iii)}. The proof for Part (iii) is more involved. Note that both components of the objective are non-negative, and from Part (i) we have examples of joint decision rules that are optimal for every data sample $\{x_i\}$. Thus, to prove that a rule is not Pareto optimal for a given $F$, it is sufficient to show there exists a dataset sample $\{x_i\} \sim F_{X}$ with positive probability for which the rule is not optimal for that sample. Since we are considering independent rules and a fixed $F, q$, we will denote $D(q(x))$ as the decision for datapoint $x$ for notational simplicity. 

\textbf{Binary setting, K = 2}. For clarity of exposition, we first prove the result in detail for a binary classification setting where $y \in \{0, 1\}$, and then extend the proof for all $K$. 

Consider independent decision rule $D$ that is not identical to the argmax rule. 

    \textbf{Case 1. Non-monotonic or random rules. } If a rule, with positive probability, is such that $q(1, x_1) > q(1, x_2)$ but $D(q(x_2)) = 1$ and $D(q(x_1)) = 0$, then it is suboptimal with positive probability. There exists some $D'$ where $D'(x) = D(x)$ for all $x$ except that $D(x_1) = 1$ and $D(x_2) = 0$. Then the marginal distribution $\pmarg$ remains the same and the expected accuracy of $D'$ is greater than that of $D$.

    \textbf{Case 2. Threshold rules} Thus, it is sufficient to consider rules with a threshold $p$, where $D_p(q(x)) = 1 \text{ if and only if } q(1, x) \geq p$, and $0$ otherwise. Note that $p = 0.5$ corresponds to the argmax decision rule. We now prove that for any other $p \neq 0.5$, with positive probability there exists a test sample such that the rule is suboptimal. Without loss of generality, assume $p < .5$, and a similar proof holds for $p > .5$. 

    By assumption, we have that $\exists \delta > 0$ such that $D_p$ disagrees with the argmax rule with positive probability $\delta$, i.e., $\Pr( q(1, x) \in (p, .5)) = \delta$. Then, with probability equal to $\delta^N > 0$, we have that the \textit{entire} data sample $x_{1:N}$ is such that $q(1, x_i) \in (p, .5)$, and so $D_p$ (which has $\hat y_i = 1, \forall i$) disagrees with the argmax rule on \textit{every} datapoint (which has $\hat y_i = 0, \forall i$). By the assumption of non-trivial reference distribution $\pref \in \mathcal P_N$, we have 
\begin{align*}
    p_{\textrm{ref}}(0, \{x_i\}) \geq \frac1N
\end{align*}
and so we can improve both accuracy and distributional fidelity by switching a decision on at least one point from $\hat y_i = 1$ to $\hat y_i = 0$. 
    
~\\\noindent\textbf{Extending argument for all $K \geq 2$.} Now, we show that the argument holds for any $K$, for any $F, q$ -- either an independent decision rule $D$ is equivalent to the argmax rule (with probability 1, they make the same decisions), or it is not on the Pareto curve. The argument is similar to the above, and is as follows. 

Consider a class $y$ such that there exists a set $X \subseteq \mX$ with $F(X) > 0$ such that, for $x \in X$, the argmax rule discretizes $x$ to $y$ but the rule $D$ does not: $D_{\textrm{argmax}}(q(x)) = y$, but $D(q(x)) \neq y$. Such a class $y$ and set $X$ exists since $D$ is not equivalent to $D_{\textrm{argmax}}$. 

Now, consider a dataset sample such that, $x_i \in X$ for all $i$. This happens with positive probability $F(X)^N$. For this sample, $D$ disagrees with $D_{\textrm{argmax}}$ on \textit{every} data point $x_i$, and note that $D_{\textrm{argmax}}(q(x_i)) = y$ for all $x_i$. By the assumption of non-trivial reference distribution $\pref \in \mathcal P_N$, we have 
\begin{align*}
    p_{\textrm{ref}}(y, \{x_i\}) \geq \frac1N,
\end{align*}
i.e., matching the reference distribution would allocate at least one data point to $y$. Thus, we can improve both accuracy and fidelity to the reference distribution by switching the decision of at least one point to $y$ (in particular, finding a class $y'$ that is over-allocated compared to the reference distribution). 
\end{proof}

\subsection{Simulating the Pareto Curve of Methods}

\Cref{fig:pareto-simulation} demonstrates \Cref{thm:pareto} with a simulated, Bayes optimal predictor. Data is generated as described in \Cref{sec:methodsdetails}, using $\sigma=0.5$ $(\mae = 0.42)$.
    \begin{figure}[!h]
    \centering
    \includegraphics[width=0.45\linewidth]{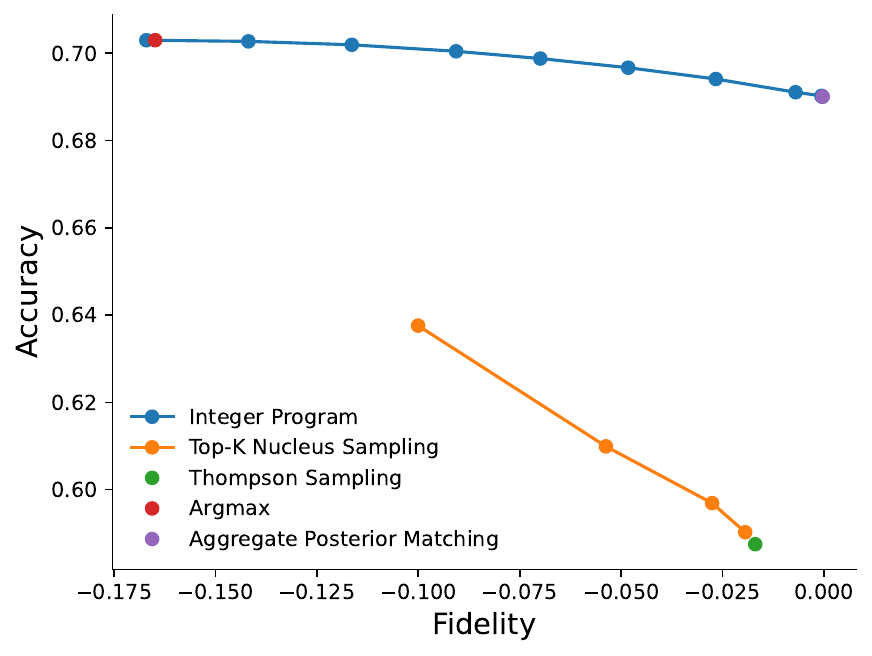}
    \caption{In simulation, the performance of various commonly used decision rules. The results illustrate \Cref{thm:pareto}: the argmax rule maximizes accuracy, but is the only Pareto optimal independent rule. Simulation details and decision rules are described in \Cref{sec:simulationsetup}.}
    \label{fig:pareto-simulation}
    \end{figure}

\subsection{Simulating Worst-Case Bias}

We now draw data according to the worst-case example discussed in the proof of \Cref{thm:informationargmax}. \Cref{fig:worstcasebias} replicates \Cref{fig:info-bias} from the main text, with this alternative synthetic data approach. Note that the bias of the most common class is equal to $\mae$, as we prove in \Cref{thm:informationargmax}. Class $1$ corresponds to the plurality class $z$, while the other classes all have equal probability. In the language of our proof, we have $e = 0$ and $a = 1-c$, and we vary $c$ in the plot.

\label{sec:worstcasesimulation}
\begin{figure}[!h]
    \centering
    \includegraphics[width=0.5\linewidth]{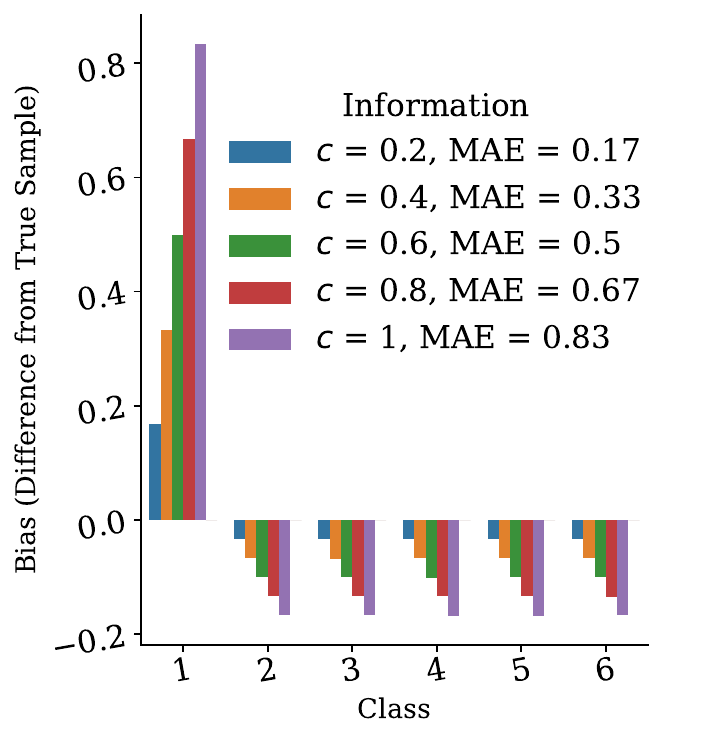}    
    \caption{Simulated bias for each class for the worst-case example in which bias of the most common class is equal to the $\mae$.}
\label{fig:worstcasebias}
\end{figure}

\FloatBarrier
\section{Additional empirical analyses}

\subsection{Model calibration} \label{sec:calibration}

\Cref{fig:calibration-nc} shows the calibration curve of the voter file's continuous scores, assessed in North Carolina for individuals for whom self-reported ground truth data is available. 

\begin{figure}[!h]
    \centering
    \includegraphics[width=0.5\linewidth]{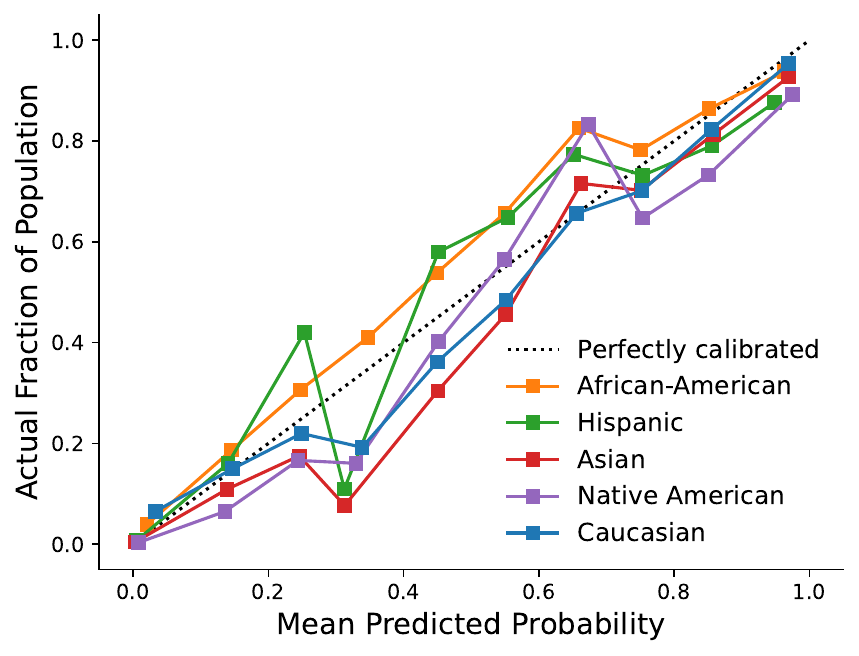}    
\caption{Calibration plot for the continuous scores in the voter file, assessed using North Carolina individuals for whom self-reported ground truth data is available. Mathematically, the $x$ axis corresponds to (binned) $q(y, x)$, and the $y$ axis corresponds to $\bbE[Y = y | q(y, x) = c]$. Thus, a perfectly calibrated classifier would be on the $y = x$ line for each group. Note that \textit{African-American} is fully above the line, meaning that for any predicted fraction $c$, a \textit{higher} proportion of the population is actually \textit{African American}. Thus, the classifier is under-predicting the group, leading to even the aggregate posterior under-counting the group in \Cref{fig:votercounts}. Conversely, \textit{Caucasian} is primarily below the line, indicating that the continuous model is over-predicting the group.}
\label{fig:calibration-nc}
\end{figure}

\subsection{Extended Analysis} 

As described in \Cref{sec:methodsdetails}, we empirically examine and plot additional discretization methods, such as Top-\textit{k} sampling. Results are in \Cref{fig:votercounts-extended,tab:table-extended}. These methods are also included in our replication analyses in \Cref{sec:public,sec:states,sec:dropuncoded}.

\begin{figure}[!h]
	\centering
	\includegraphics[width=0.5\linewidth]{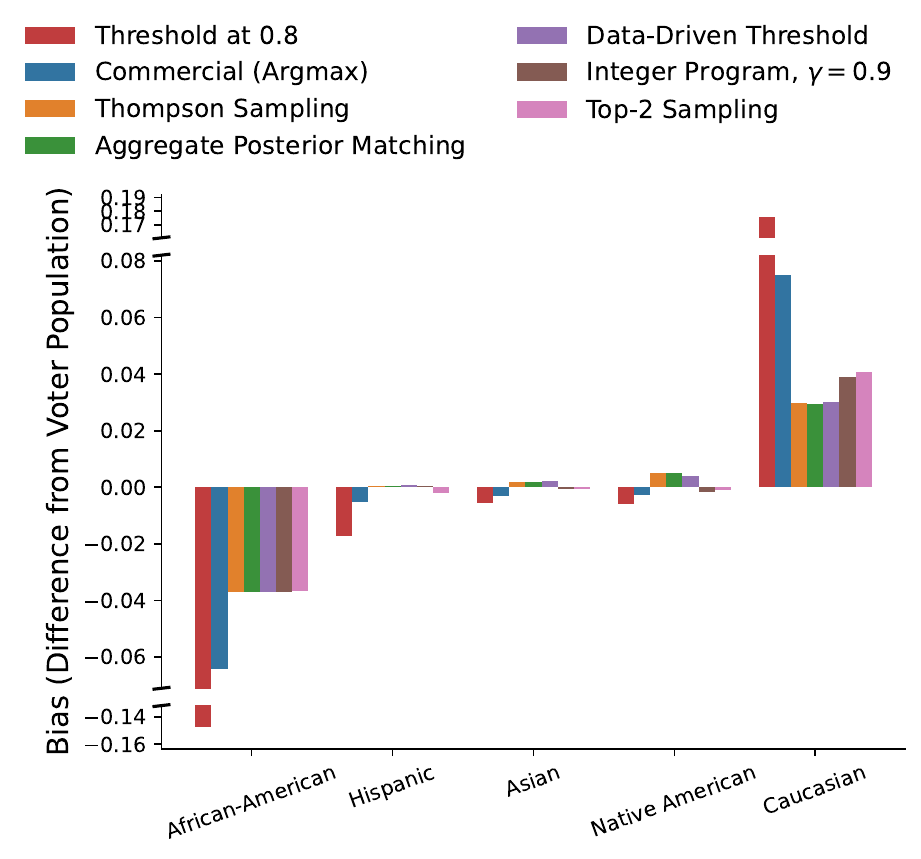}
	\caption{Extension of \Cref{fig:votercounts} with additional decision-making rules plotted.}
	\label{fig:votercounts-extended}
\end{figure}

\begin{table}[tb]
	\begin{center}
		\begin{tabular}{|c|c|c|c|}
			\hline
			& \text{Accuracy with} & \text{Fidelity to} & \text{Fidelity to} \\
			& \text{Ground Truth} & \text{Ground Truth} & \text{Aggregate Posterior} \\
			\hline\hline
			Threshold at 0.8 (Among $67.9\%$ Not Dropped) & \textbf{0.928} & -0.352 & -0.293 \\ 
			Integer Program, $\gamma = 0.9$ & 0.846 & -0.079 & -0.019 \\         Commercial (Argmax) & 0.844 & -0.150 & -0.091 \\ 
			Aggregate Posterior Matching & 0.841 & -0.074 & \textbf{-0.000} \\ 
			Data-Driven Threshold Matching Heuristic & 0.841 & -0.074 & -0.002 \\         
			\text{County-Conditional Aggregate Posterior Matching} & 0.834 & -0.074 & -0.000 \\
			Thompson Sampling & 0.776 & -0.074 & -0.000 \\ 
			Top-2  Sampling & 0.787 & -0.081 & -0.023 \\ 
			Ground Truth Marginal Matching & 0.841 & -0.000 & -0.074 \\ 
			\text{County-Conditional Ground Truth Marginal Matching} & 0.840 & \textbf{-0.000} & -0.074 \\ \hline
		\end{tabular} 
		\caption{Extension of \Cref{tab:table-NC} with additional decision-making rules plotted. Note that access to the ground truth distribution makes bias minimization trivial, while achieving accuracy comparable to the aggregate posterior.}
		\label{tab:table-extended} 
	\end{center}
\end{table}

\subsection{Replication on Public Data} \label{sec:public}

We replicate our analysis and results from publicly available data in North Carolina on imputation models developed by \citet{greengard2024improved}, whose data and code are publicly available for replication \citep{DVN/QIM4UF_2023}. Their data in North Carolina ($N=\num{4233012}$) uses $K = 6$ race/ethnicity categories: 
(non-Hispanic) \textit{Black}, \textit{American
Indian and Alaskan Native (AIAN)}, \textit{Asian and Pacific Islander (API)}, \textit{White}, and \textit{Hispanic}, with \textit{Hispanic} encompassing voters of any racial background.

We examine two sets of models and predicted probabilities used by \citet{greengard2024improved}: Bayesian Improved Surname Geocoding (BISG) (the standard race prediction algorithm), as well their novel contribution, a raking-based improved model. Their implementation of BISG uses voter file information. 

\Cref{fig:calibration-public} shows calibration plots for the two models provided; BISG in particular is calibrated for both White and Black groups. In \Cref{fig:votercounts-public}, we observe that argmax bias persists, even when a model may be well-calibrated. Notably, even when aggregate probabilities may overrepresent \textit{Black} voters and underrepresent \textit{White}, argmax discretization \textit{still} leads to undercounting for the former and overcounting for the latter. However, because the Raking predictive model comparatively under-predicts \textit{White} individuals, all discretization approaches that use the aggregate posterior under-count the \textit{White} class. \Cref{fig:pareto-public} shows that, with more calibrated models, argmax is accuracy-maximizing over joint rules, more closely resembling the Bayes optimal setting in simulation \Cref{fig:pareto-simulation}. Lastly, \Cref{fig:maps-appendix} replicates \Cref{fig:maps}.

\begin{figure}[!h]
    \centering
    \begin{subfigure}{.45\textwidth}
    \centering
    \includegraphics[width=\linewidth]{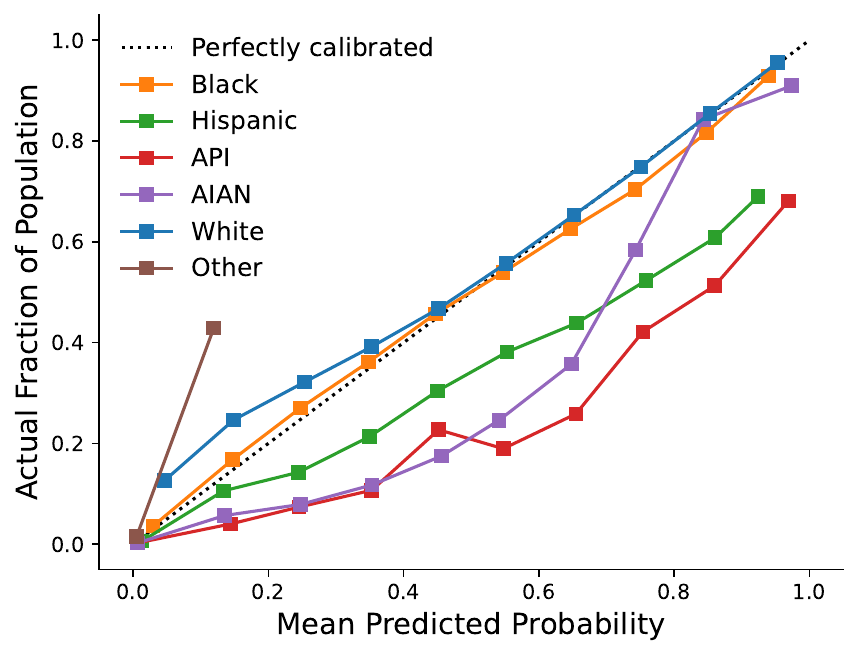}
    \caption{BISG.}
    \label{fig:calibration-bisg}    
    \end{subfigure}
    \hfill
    \begin{subfigure}{.45\textwidth}
    \centering
    \includegraphics[width=\linewidth]{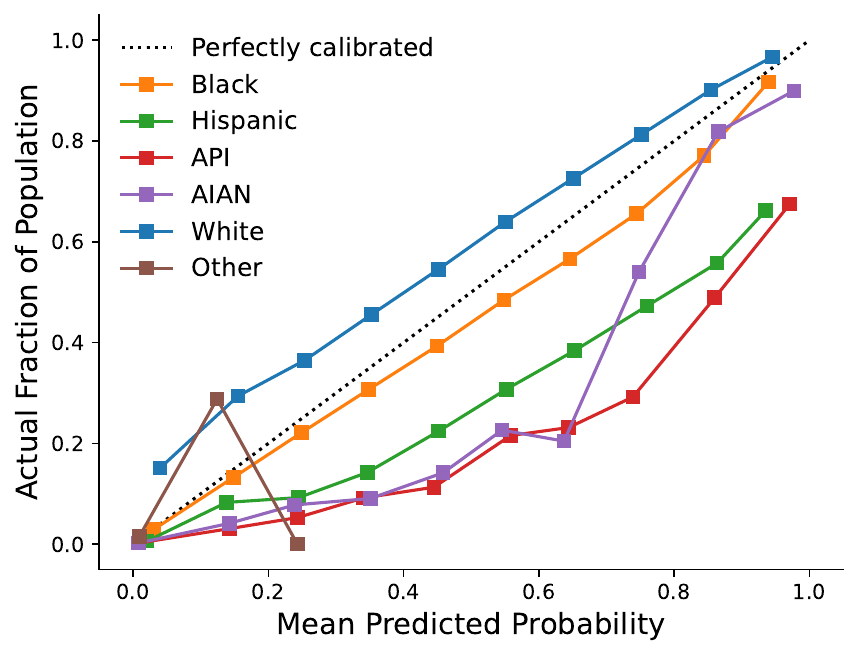}
    \caption{Raking.}
    \label{fig:calibration-raking}    
    \end{subfigure}
    \caption{Replicating \Cref{fig:calibration-nc} with the data and predicted probabilities in \citet{greengard2024improved}.}
    \label{fig:calibration-public}
\end{figure}

\begin{figure}[!h]
    \centering
    \begin{subfigure}{.45\textwidth}
    \centering
    \includegraphics[width=\linewidth]{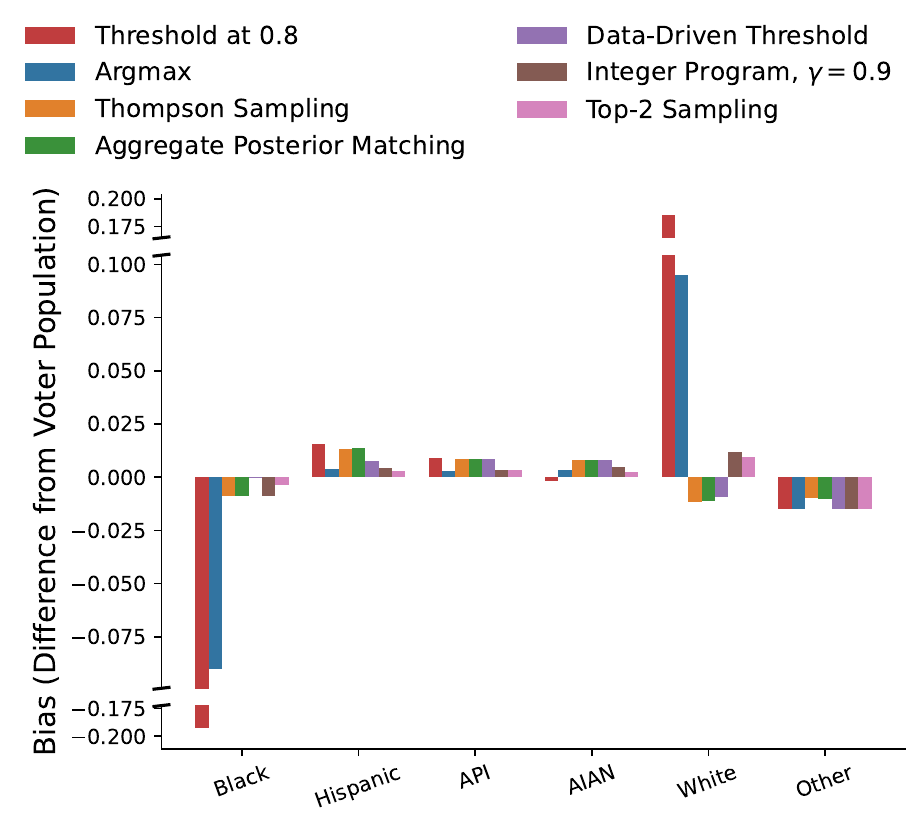}
    \caption{BISG.}
    \label{fig:votercounts-bisg}
    \end{subfigure}
    \hfill
    \begin{subfigure}{.45\textwidth}
    \centering
    \includegraphics[width=\linewidth]{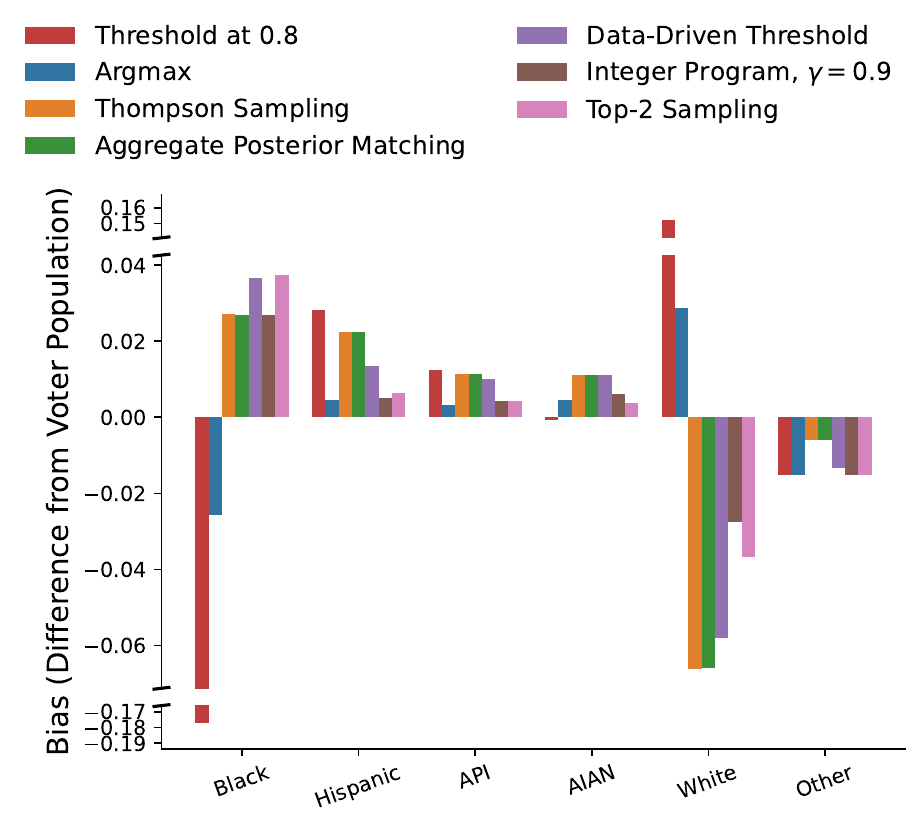}
    \caption{Raking.}
    \label{fig:votercounts-raking}
    \end{subfigure}
    \caption{Replicating \Cref{fig:votercounts} with the data and predicted probabilities in \citet{greengard2024improved}.}
    \label{fig:votercounts-public}
\end{figure}

\begin{figure}[!h]
    \centering
    \begin{subfigure}{.45\textwidth}
    \centering
    \includegraphics[width=\linewidth]{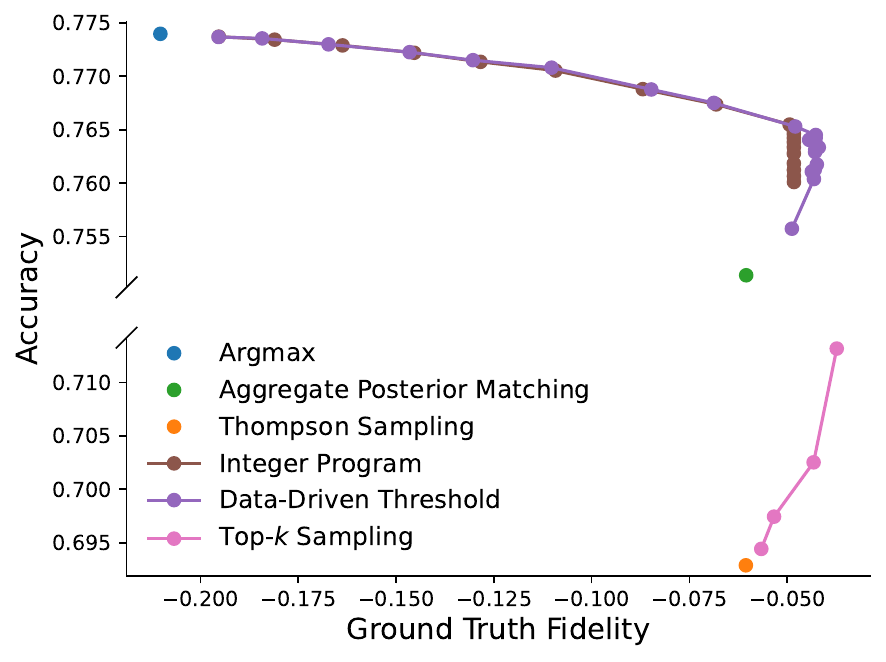}
    \caption{BISG.}
    \label{fig:pareto-bisg}
    \end{subfigure}
    \hfill
    \begin{subfigure}{.45\textwidth}
    \centering
    \includegraphics[width=\linewidth]{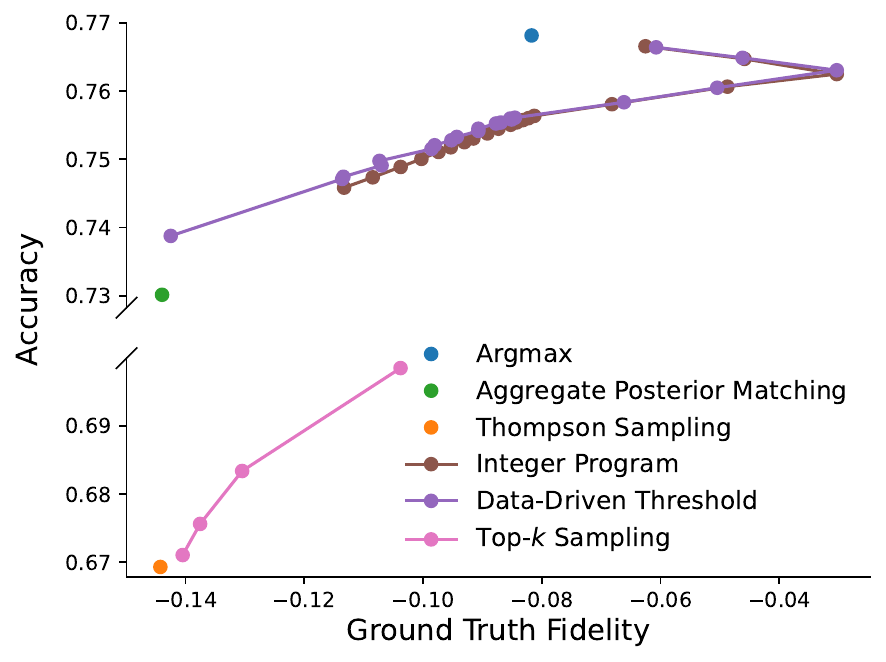}
    \caption{Raking.}
    \label{fig:pareto-raking}
    \end{subfigure}
    \caption{Replicating \Cref{fig:pareto-NC} with the data and predicted probabilities in \citet{greengard2024improved}. Note that in \Cref{fig:pareto-raking}, the aggregate posterior is further from the ground truth distribution than the argmax solution, as the predictive model underrepresents the white majority (see \Cref{fig:calibration-raking}), and so moving methods that increase fidelity to the Aggregate Posterior actually decrease fidelity to the Ground Truth.}
    \label{fig:pareto-public}
\end{figure}

\begin{figure}[!h]
	\centering
    \begin{subfigure}{.3\textwidth}
        \centering
        \includegraphics[width=1.25\linewidth]{images/popmap_ap.pdf}
        \caption{White population.}
        \label{fig:popmap-public}
    \end{subfigure}	
    \hfill    
    \begin{subfigure}{.3\textwidth}
        \centering
    \includegraphics[width=.85\linewidth]{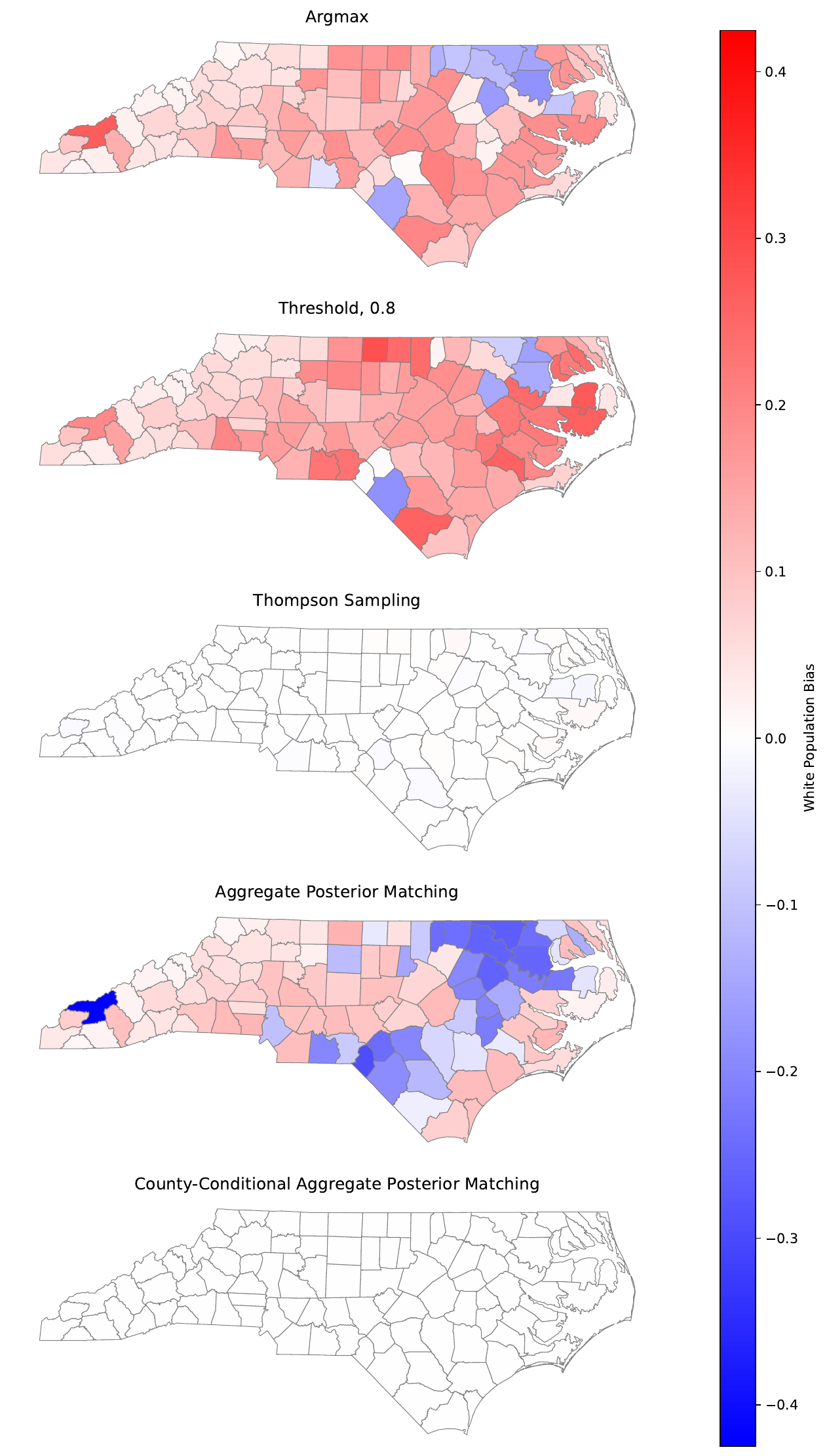}
	\caption{BISG.}
	\label{fig:biasmap-bisg}
    \end{subfigure}
    \begin{subfigure}{.3\textwidth}
        \centering
    \includegraphics[width=.85\linewidth]{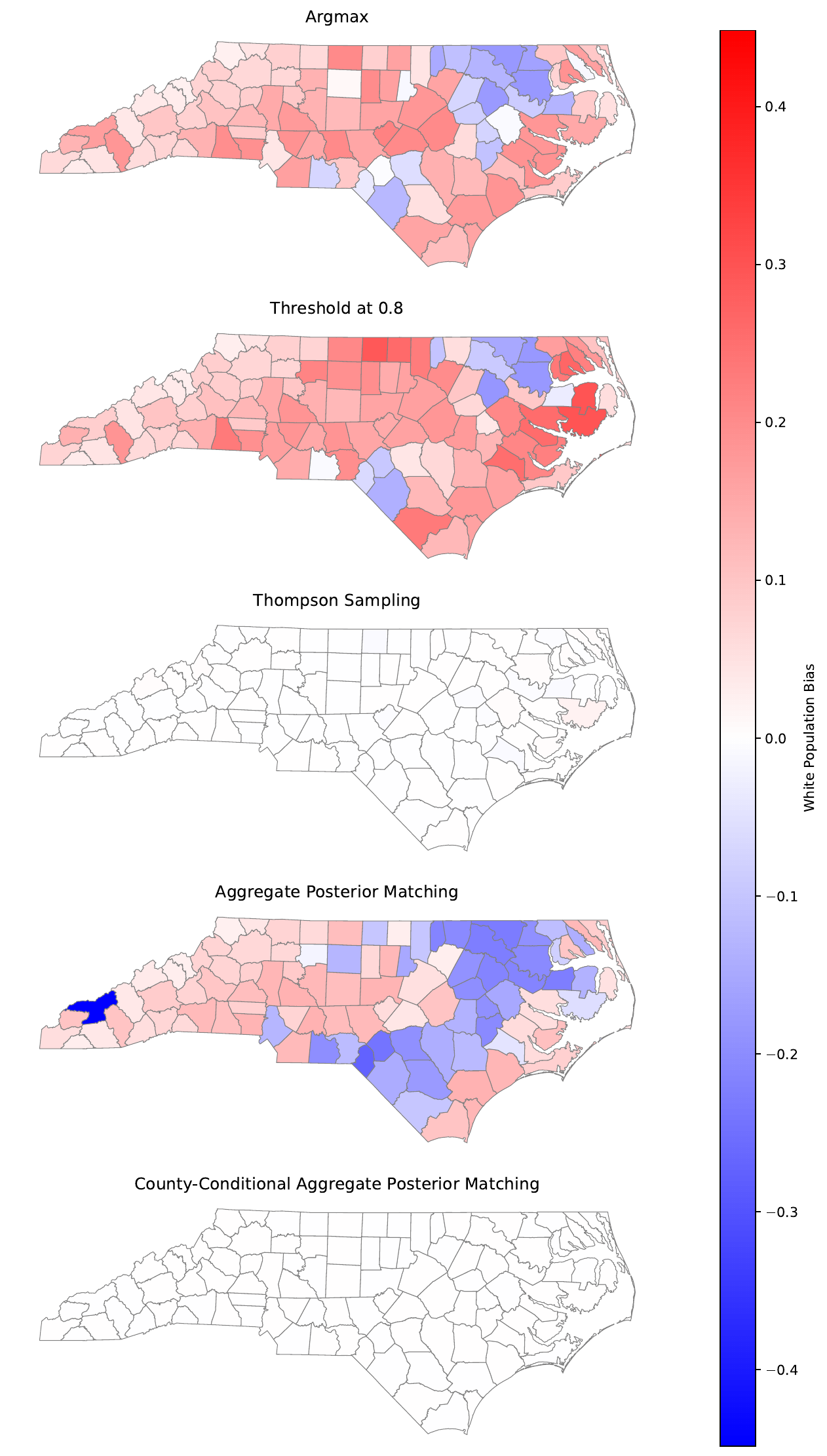}
	\caption{Raking.}
	\label{fig:biasmap-raking}
    \end{subfigure}
	\caption{Replication of \Cref{fig:maps} using the public data and predictions from \citet{greengard2024improved}. 
	}
	\label{fig:maps-appendix}
\end{figure}

\begin{table}[!h]
\begin{subtable}{\textwidth}
\centering
\begin{tabular}{|c|c|c|c|}
    \hline
    & \text{Accuracy with} & \text{Fidelity to} & \text{Fidelity to} \\
    & \text{Ground Truth} & \text{Ground Truth} & \text{Aggregate Posterior} \\
    \hline\hline
    Threshold at 0.8 (Among $53.4\%$ Not Dropped) & \textbf{0.904} & -0.420 & -0.398 \\ 
    Integer Program, $\gamma = 0.9$ & 0.765 & -0.048 & -0.046 \\
    Argmax & 0.774 & -0.210 & -0.213 \\ 
    Aggregate Posterior Matching & 0.751 & -0.061 & \textbf{-0.000} \\ 
    Data-Driven Threshold Matching Heuristic & 0.756 & -0.049 & -0.021 \\  
    Thompson Sampling & 0.693 & -0.061 & -0.001 \\ 
    Top-2 Sampling & 0.713 & -0.037 & -0.053 \\ 
    Ground Truth Marginal Matching & 0.759 & \textbf{-0.000} & -0.060 \\ \hline
\end{tabular}
\caption{BISG.} 
\label{tab:table-bisg} 
\end{subtable}

\begin{subtable}{\textwidth}
\centering
\begin{tabular}{|c|c|c|c|}
    \hline
    & \text{Accuracy with} & \text{Fidelity to} & \text{Fidelity to} \\
    & \text{Ground Truth} & \text{Ground Truth} & \text{Aggregate Posterior} \\
    \hline\hline
    Threshold at 0.8 (Among $45.9\%$ Not Dropped) & \textbf{0.913} & -0.386 & -0.450 \\ 
    Integer Program, $\gamma = 0.9$ & 0.755 & -0.085 & -0.077 \\ 
    Argmax & 0.768 & -0.082 & -0.189 \\ 
    Aggregate Posterior Matching & 0.730 & -0.144 & \textbf{-0.000} \\ 
    Data-Driven Threshold Matching Heuristic & 0.739 & -0.142 & -0.035 \\ 
    Thompson Sampling & 0.669 & -0.144 & -0.001 \\ 
    Top-2 Sampling & 0.698 & -0.104 & -0.079 \\ 
    Ground Truth Marginal Matching & 0.759 & \textbf{-0.000} & -0.144 \\ 
    \hline
\end{tabular}
\caption{Raking.}
\label{tab:table-raking} 
\end{subtable}
\caption{Replicating \Cref{tab:table-NC} with the data and predicted probabilities in \citet{greengard2024improved}.}
\label{tab:table-public}
\end{table}

\subsection{Replication on other states} \label{sec:states}

We replicate our results (\Cref{fig:votercounts-appendix,fig:pareto-appendix,tab:table-appendix} corresponding to \Cref{fig:votercounts,fig:pareto-NC,tab:table-NC}), with additional methods displayed, on similarly processed commercial voter file data from South Carolina ($N=\num{3494505}$) and Florida ($N=\num{13766639}$), which also report self-reported race. For further robustness, we also replicate on data from all three states joined together ($N=\num{23635780}$), as well as all individuals across in the United States who have self-reported race/ethnicity data in the commercial file (i.e., including other states with self-reported race and former residents of NC/SC/FL that have moved elsewhere) ($N=\num{39685870}$).

\begin{figure}[!h]
\centering
\begin{subfigure}{.45\textwidth}
\centering
\includegraphics[width=\linewidth]{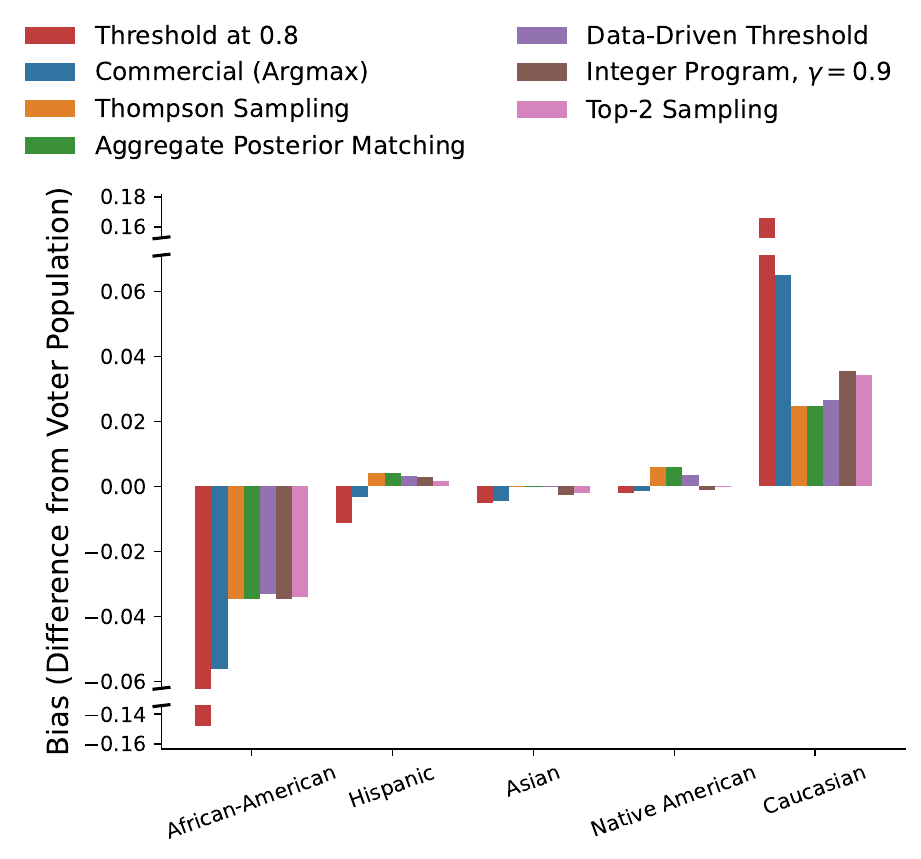}
\caption{South Carolina.}
\label{fig:votercounts-SC}   
\end{subfigure}
\hfill
\begin{subfigure}{.45\textwidth}
\centering
\includegraphics[width=\linewidth]{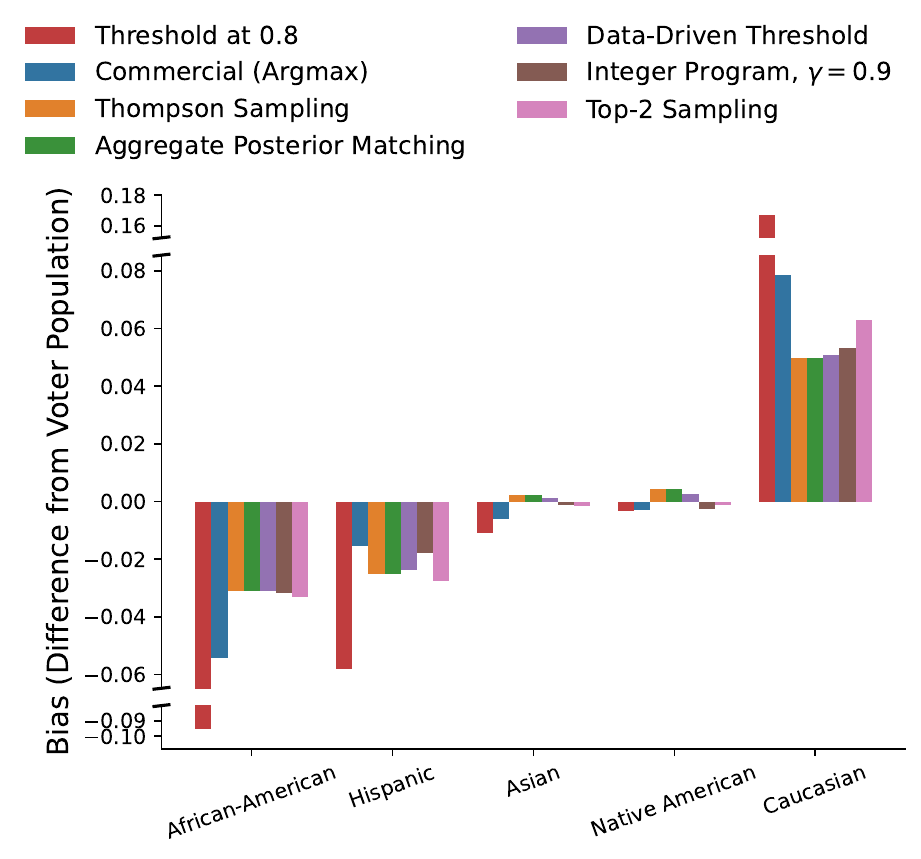}
\caption{Florida.}
\label{fig:votercounts-FL}    
\end{subfigure}
\begin{subfigure}{.45\textwidth}
\centering
\includegraphics[width=\linewidth]{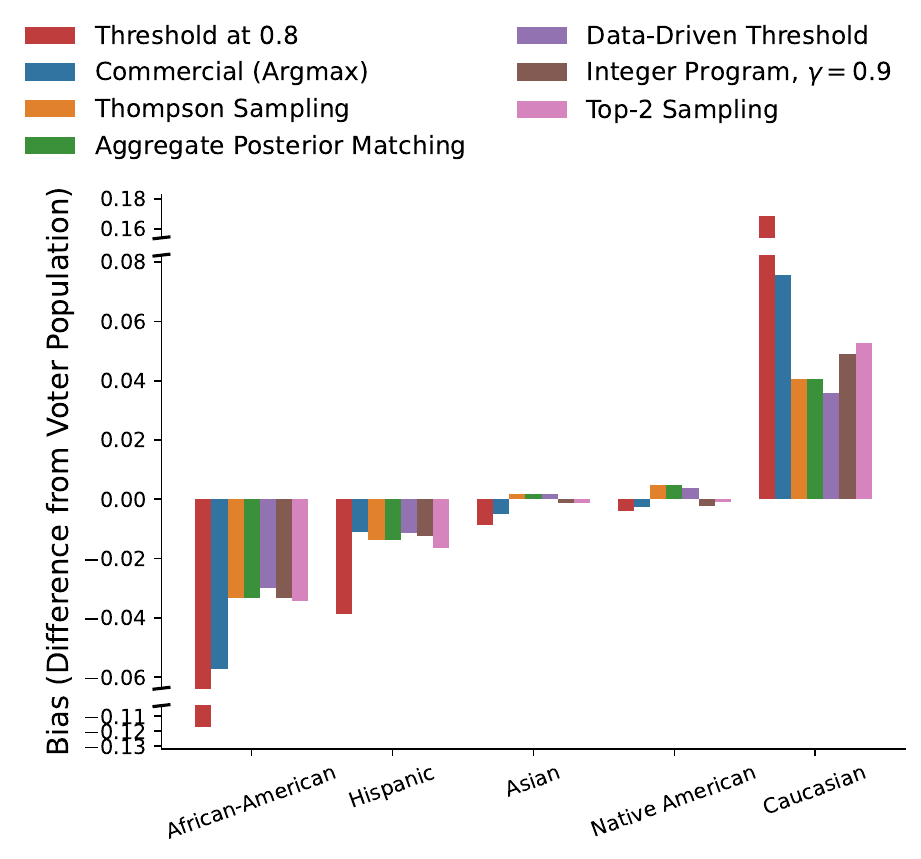}
\caption{North Carolina, South Carolina, and Florida combined.}
\label{fig:votercounts-3states}
\end{subfigure}
\hfill
\begin{subfigure}{.45\textwidth}
\centering
\includegraphics[width=\linewidth]{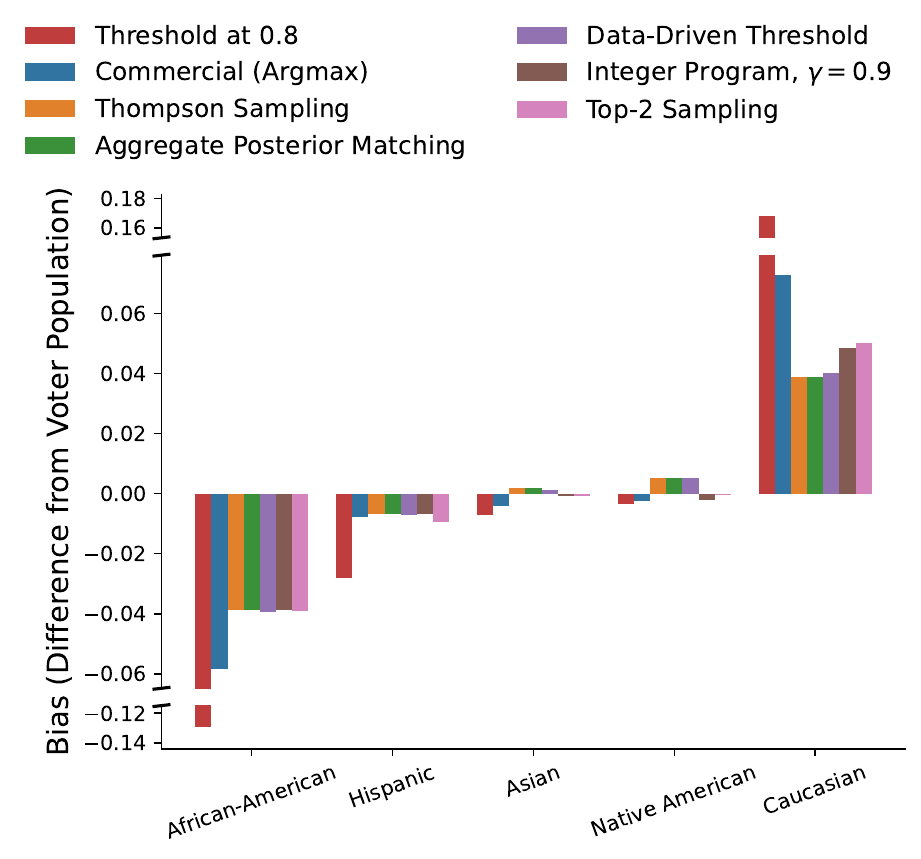}
\caption{All self-reported race data in the commercial voter file across the US.}
\label{fig:votercounts-US}   
\end{subfigure}
\caption{Replication of \Cref{fig:votercounts} with additional methods and data for other states from the commercial voter file.} 
\label{fig:votercounts-appendix}
\end{figure}

\begin{figure}[!h]
\centering
\begin{subfigure}{.45\textwidth}
\centering
\includegraphics[width=\linewidth]{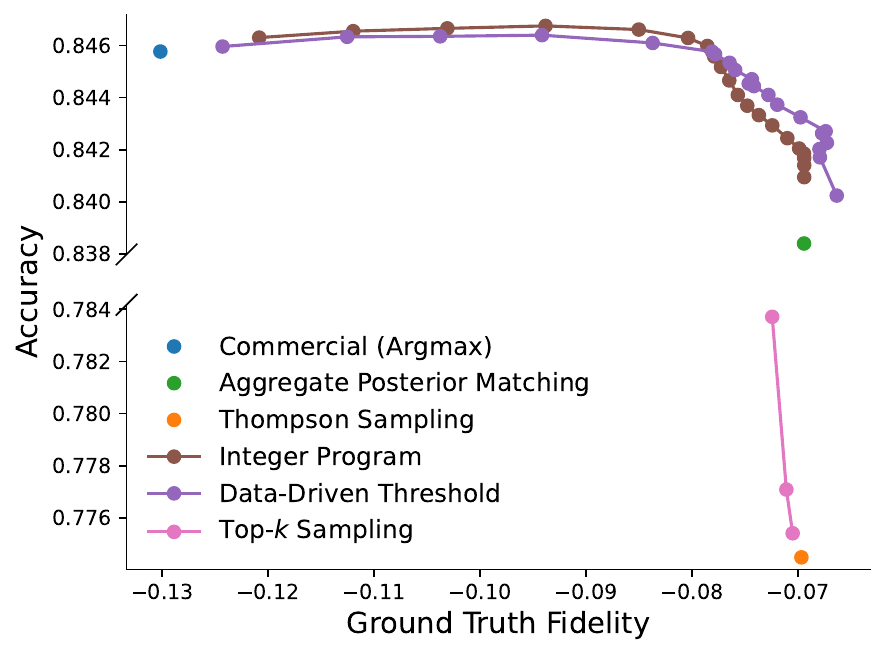}
\caption{South Carolina.}
\label{fig:pareto-SC}   
\end{subfigure}
\hfill
\begin{subfigure}{.45\textwidth}
\centering
\includegraphics[width=\linewidth]{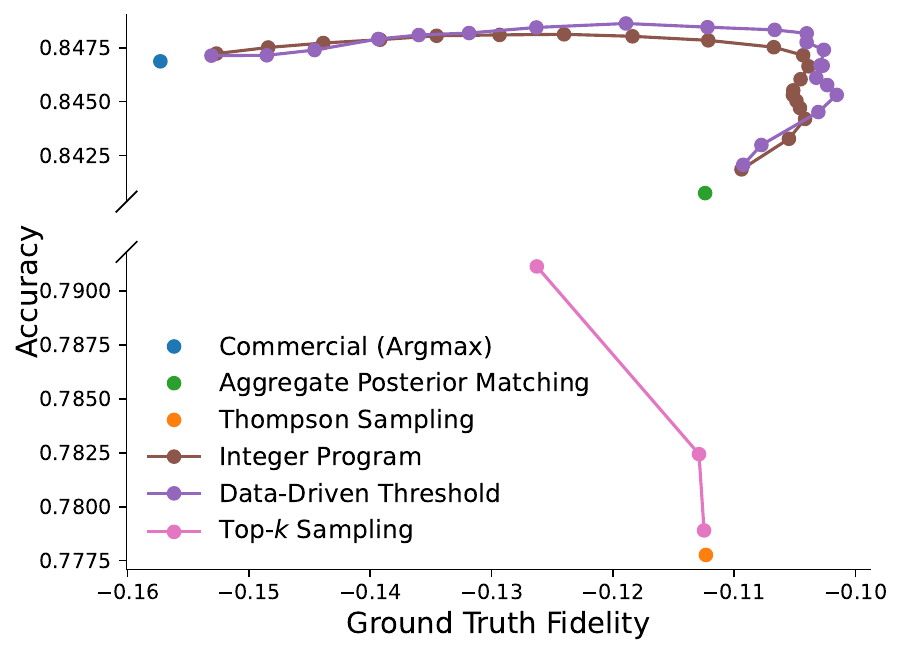}
\caption{Florida.}
\label{fig:pareto-FL}    
\end{subfigure}
\begin{subfigure}{.45\textwidth}
\centering
\includegraphics[width=\linewidth]{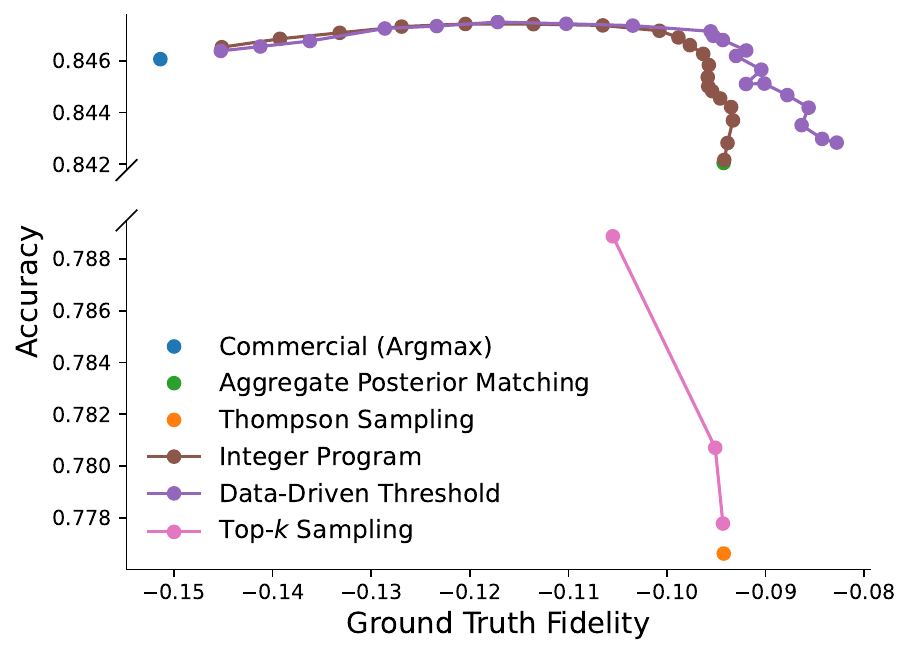}
\caption{North Carolina, South Carolina, and Florida combined.}
\label{fig:pareto-3states}\end{subfigure}
\hfill
\begin{subfigure}{.45\textwidth}
\centering
\includegraphics[width=\linewidth]{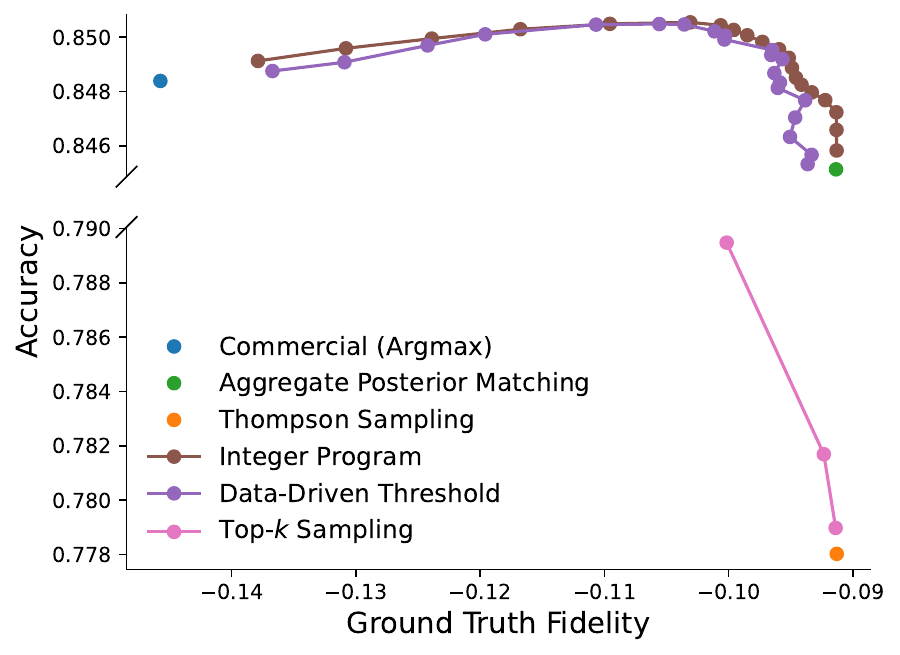}
\caption{All self-reported race data in the commercial voter file across the US.} \label{fig:pareto-US}   
\end{subfigure}
\caption{Replication of \Cref{fig:pareto-NC} with additional methods and data for other states from the commercial voter file.}
\label{fig:pareto-appendix}
\end{figure}

\begin{table}
\begin{subtable}{\textwidth}
\centering
\begin{tabular}{|c|c|c|c|}
    \hline
    & \text{Accuracy with} & \text{Fidelity to} & \text{Fidelity to} \\
    & \text{Ground Truth} & \text{Ground Truth} & \text{Aggregate Posterior} \\
    \hline\hline
    Threshold at 0.8 (Among $66.8\%$ Not Dropped) & 0.926 & -0.332 & -0.283 \\ 
    Integer Program, $\gamma = 0.9$ & 0.845 & -0.076 & -0.021 \\ 
    Commercial (Argmax) & 0.846 & -0.130 & -0.081 \\ 
    Aggregate Posterior Matching & 0.838 & -0.069 & -0.000 \\ 
    Data-Driven Threshold Matching Heuristic & 0.840 & -0.066 & -0.007 \\ 
    Thompson Sampling & 0.774 & -0.070 & -0.000 \\ 
    Top-2 Sampling & 0.784 & -0.072 & -0.020 \\ 
    Ground Truth Marginal Matching & 0.840 & -0.000 & -0.069 \\ 
    \hline
\end{tabular}
\caption{South Carolina.}
\label{tab:table-SC} 
\end{subtable}
\begin{subtable}{\textwidth}
\centering
\begin{tabular}{|c|c|c|c|}
    \hline
    & \text{Accuracy with} & \text{Fidelity to} & \text{Fidelity to} \\
    & \text{Ground Truth} & \text{Ground Truth} & \text{Aggregate Posterior} \\
    \hline\hline
    Threshold at 0.8 (Among $70.8\%$ Not Dropped) & 0.918 & -0.335 & -0.235 \\ 
    Integer Program, $\gamma = 0.9$ & 0.848 & -0.107 & -0.022 \\ 
    Commercial (Argmax) & 0.847 & -0.157 & -0.077 \\ 
    Aggregate Posterior Matching & 0.841 & -0.112 & -0.000 \\ 
    Data-Driven Threshold Matching Heuristic & 0.842 & -0.109 & -0.005 \\ 
    Thompson Sampling & 0.778 & -0.112 & -0.000 \\ 
    Top-2 Sampling & 0.791 & -0.126 & -0.027 \\ 
    Ground Truth Marginal Matching & 0.843 & -0.000 & -0.112 \\ \hline
\end{tabular}
\caption{Florida.}
\label{tab:table-FL} 
\end{subtable}
\begin{subtable}{\textwidth}
\centering
\begin{tabular}{|c|c|c|c|}
    \hline
    & \text{Accuracy with} & \text{Fidelity to} & \text{Fidelity to} \\
    & \text{Ground Truth} & \text{Ground Truth} & \text{Aggregate Posterior} \\
    \hline\hline
    Threshold at 0.8 (Among $69.4\%$ Not Dropped) & 0.921 & -0.337 & -0.256 \\ 
    Integer Program, $\gamma = 0.9$ & 0.847 & -0.098 & -0.020 \\ 
    Commercial (Argmax) & 0.846 & -0.151 & -0.076 \\ 
    Aggregate Posterior Matching & 0.842 & -0.094 & -0.000 \\ 
    Data-Driven Threshold Matching Heuristic & 0.843 & -0.083 & -0.012 \\ 
    Thompson Sampling & 0.777 & -0.094 & -0.000 \\ 
    Top-2 Sampling & 0.789 & -0.105 & -0.024 \\ 
    Ground Truth Marginal Matching & 0.842 & -0.000 & -0.094 \\ \hline
\end{tabular}
\caption{North Carolina, South Carolina, and Florida combined.}
\label{tab:table-3states} 
\end{subtable}
\begin{subtable}{\textwidth}
\centering
\begin{tabular}{|c|c|c|c|}
    \hline
    & \text{Accuracy with} & \text{Fidelity to} & \text{Fidelity to} \\
    & \text{Ground Truth} & \text{Ground Truth} & \text{Aggregate Posterior} \\
    \hline\hline
    Threshold at 0.8 (Among $68.8\%$ Not Dropped) & 0.924 & -0.336 & -0.258 \\ 
    Integer Program, $\gamma = 0.9$ & 0.850 & -0.097 & -0.020 \\ 
    Commercial (Argmax) & 0.848 & -0.146 & -0.068 \\ 
    Aggregate Posterior Matching & 0.845 & -0.091 & -0.000 \\ 
    Data-Driven Threshold Matching Heuristic & 0.845 & -0.094 & -0.003 \\ 
    Thompson Sampling & 0.778 & -0.091 & -0.000 \\ 
    Top-2 Sampling & 0.789 & -0.100 & -0.023 \\ 
    Ground Truth Marginal Matching & 0.845 & -0.000 & -0.091 \\ \hline
\end{tabular}
\caption{All self-reported race data in the commercial voter file across the US.}
\label{tab:table-US} 
\end{subtable}
\caption{Replicating \Cref{tab:table-NC} with additional methods and data for other states from the commercial voter file.}
\label{tab:table-appendix}
\end{table}

\subsection{Voter File Results Without \textit{Uncoded}} \label{sec:dropuncoded}

We now replicate our results with an alternative data processing choice: dropping the $\num{205683}$ entries that the commercial voter file labels as \textit{Uncoded}, as opposed to replacing them with the argmax choice, leaving $N=\num{6168953}$ data points. \Cref{tab:table-dropuncoded} and \Cref{fig:voterbias-dropuncoded} replicate \Cref{tab:table-NC} and \Cref{fig:voterbias}, respectively. Results are qualitatively identical when comparing across decision rules. Note that the undercounting of voters of color is higher (percentage-wise) after dropping the \textit{Uncoded} individuals, as these individuals tended to be (both in self-reported ground truth and imputed argmax) disproportionately voters of color. This follow the smae intuition seen in simpler threshold rules (e.g. probability at 0.8), where a preference for highly confident points leads to increased bias.

\begin{table}[tbh]

\begin{center}
\begin{tabular}{|c|c|c|c|}
\hline
 & \text{Accuracy with} & \text{Fidelity to} & \text{Fidelity to} \\
 & \text{Ground Truth} & \text{Ground Truth Distrib.} & \text{Aggregate Posterior Distrib.} \\
\hline\hline
    Threshold at 0.8 ($30.6\%$ Dropped) & \textbf{0.929} & -0.341 & -0.282 \\ 
    Commercial (Argmax) & 0.852 & -0.157 & -0.098 \\ 
    Thompson Sampling & 0.784 & -0.075 & -0.000 \\ 
    Aggregate Posterior Matching & 0.842 & -0.076 & -0.000 \\ 
    Data-Driven Threshold & 0.844 & -0.076 & -0.007 \\ 
    Integer Program, $\gamma = 0.9$ & 0.853 & -0.093 & -0.034 \\ 
    Top-2 Sampling & 0.795 & -0.080 & -0.024 \\ 
    Ground Truth Marginal Matching & 0.844 & \textbf{-0.000} & -0.075 \\ \hline
\end{tabular}

\caption{Replicating \Cref{tab:table-NC} with \textit{Uncoded} labels removed as opposed to replaced with the argmax continuous scores. 
}
\label{tab:table-dropuncoded} 
\end{center}
\end{table}
    \begin{figure}[!h]
    \centering
    \begin{subfigure}{.3\textwidth}
        \centering
        \includegraphics[width=.95\linewidth]{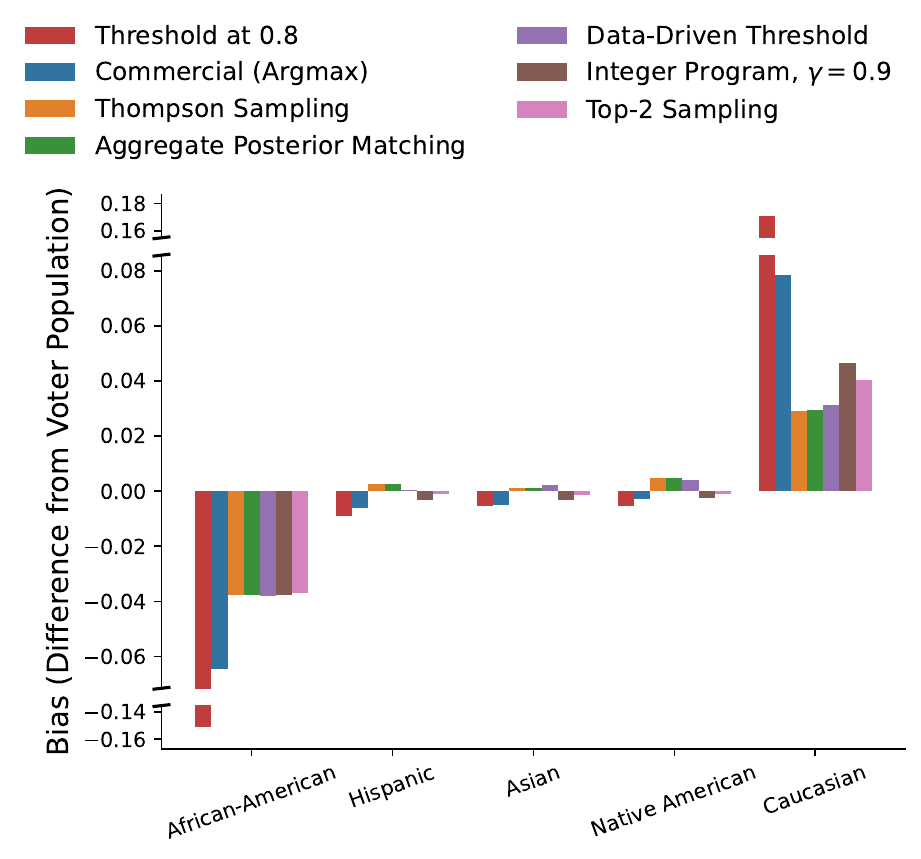}
        \caption{Replication of \Cref{fig:votercounts}}
        \label{fig:votercounts-dropuncoded}
    \end{subfigure}
    \hfill
    \begin{subfigure}{.3\textwidth}
        \centering
        \includegraphics[width=.95\linewidth]{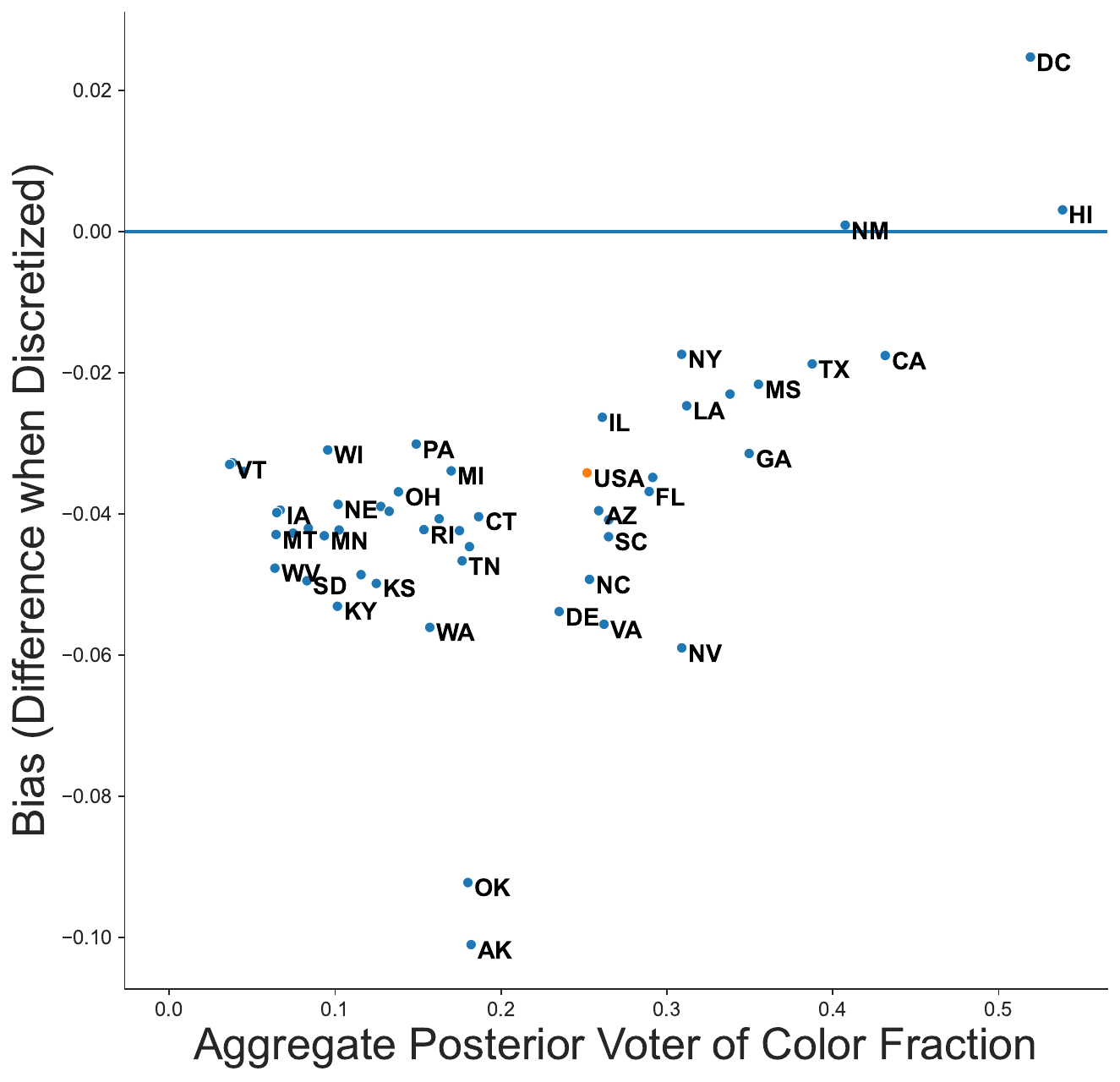}
        \caption{Replication of \Cref{fig:pocrep}}
        \label{fig:pocrep-dropuncoded}
    \end{subfigure}
    \hfill
    \begin{subfigure}{.3\textwidth}
        \centering
        \includegraphics[width=.95\linewidth]{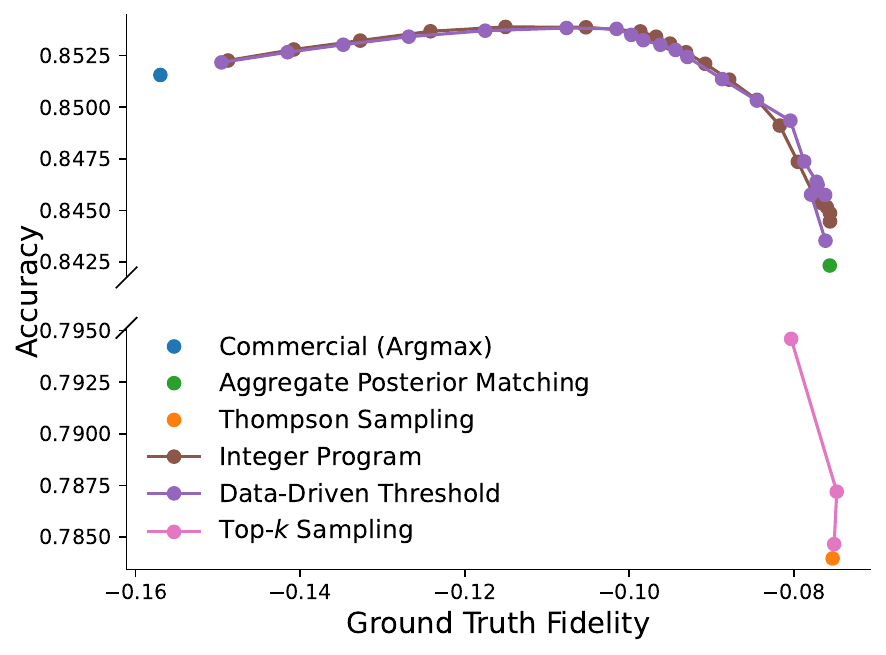}
        \caption{Replication of \Cref{fig:pareto-NC}}
        \label{fig:pareto-dropuncoded}
    \end{subfigure}    
    \caption{Replicating \Cref{fig:voterbias} with the \textit{Uncoded} category removed, $N=\num{6168953}$ registered voters. Note that in \Cref{fig:pocrep-dropuncoded}, many states across the board have lower counts for voters of color when data discretized when the \textit{Uncoded} category is removed, as such data points tend to be less white than average. }
    \label{fig:voterbias-dropuncoded}
    \end{figure}

\subsection{Additional Simplex Visualizations} \label{sec:simplex}

In \Cref{fig:simplex-appendix}, we provide additional visualizations for decision-making rules as a supplement to the figures in \Cref{fig:simplex-main}.  

\begin{figure}[!h]
    \centering
    \begin{subfigure}{.3\textwidth}
        \centering
        \includegraphics[width=.9\linewidth]{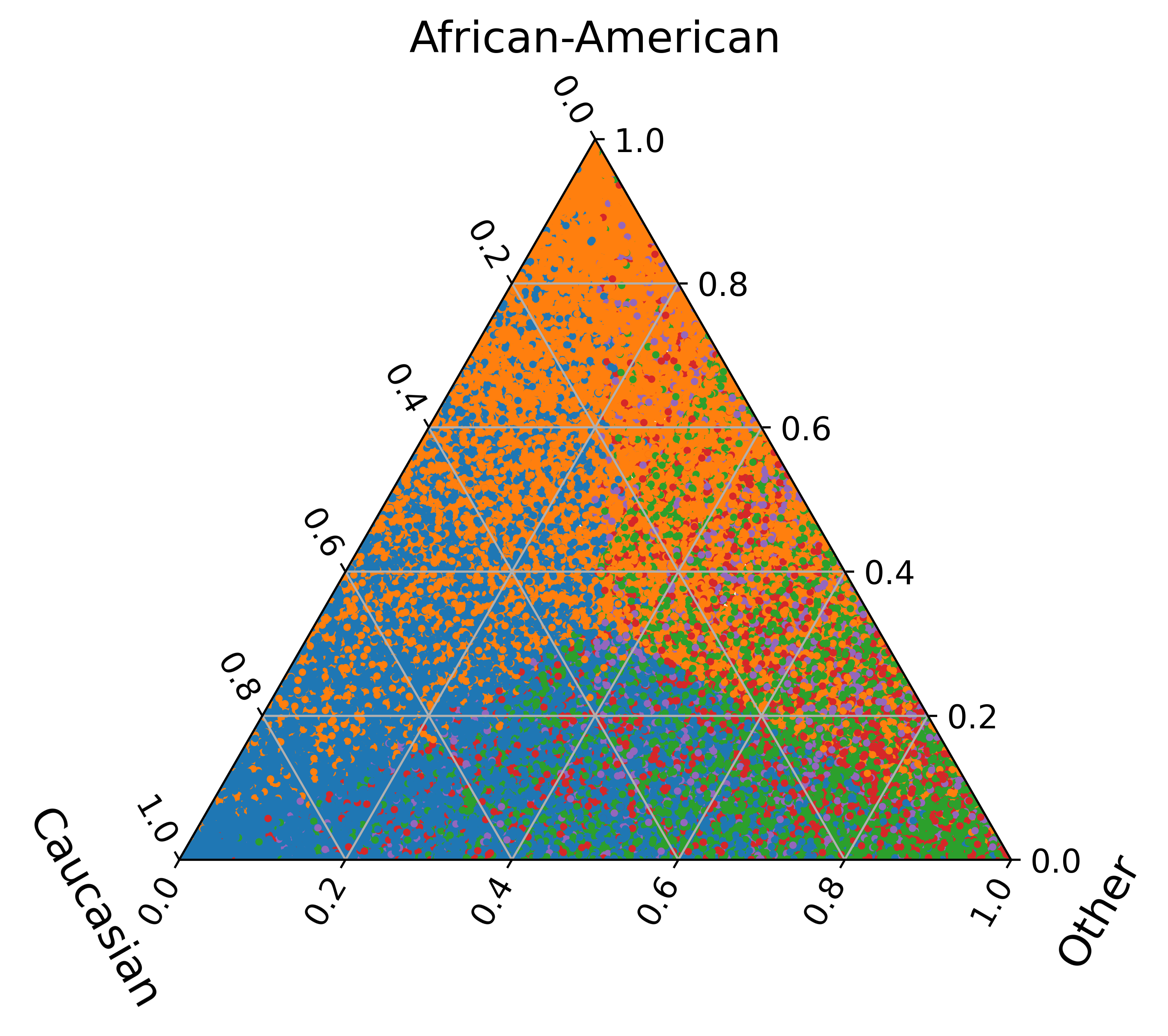}
        \caption{Top-2 Sampling, randomized according to rescaled probabilities.}
        \label{fig:simplex-top}
    \end{subfigure}        
    \hfill
    \begin{subfigure}{.3\textwidth}
        \centering
        \includegraphics[width=.9\linewidth]{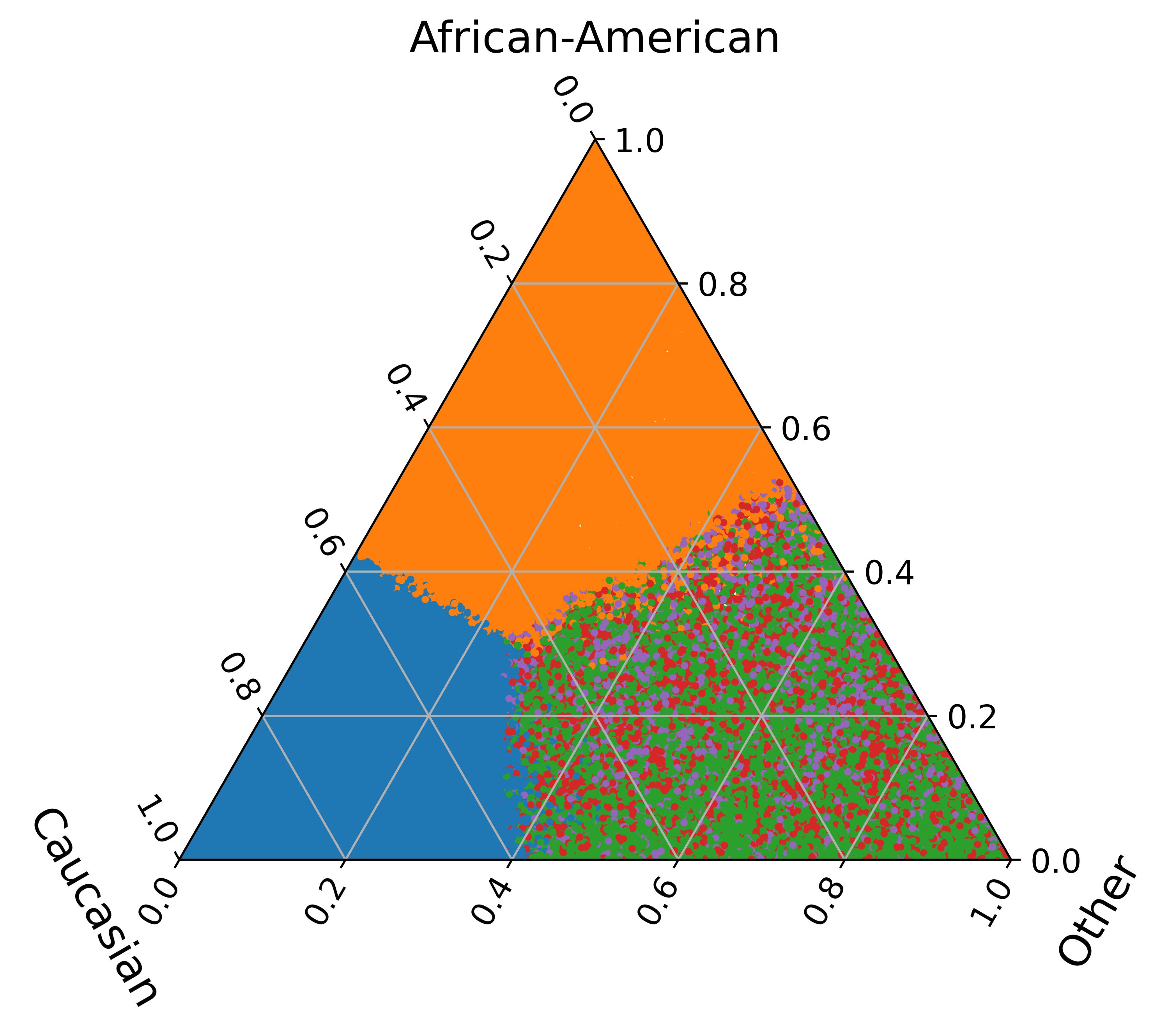}
        \caption{Integer program, with chosen parameter $\gamma=0.9$.}
        \label{fig:simplex-ip}
    \end{subfigure}
    \hfill
    \begin{subfigure}{.3\textwidth}
        \centering
        \includegraphics[width=.9\linewidth]{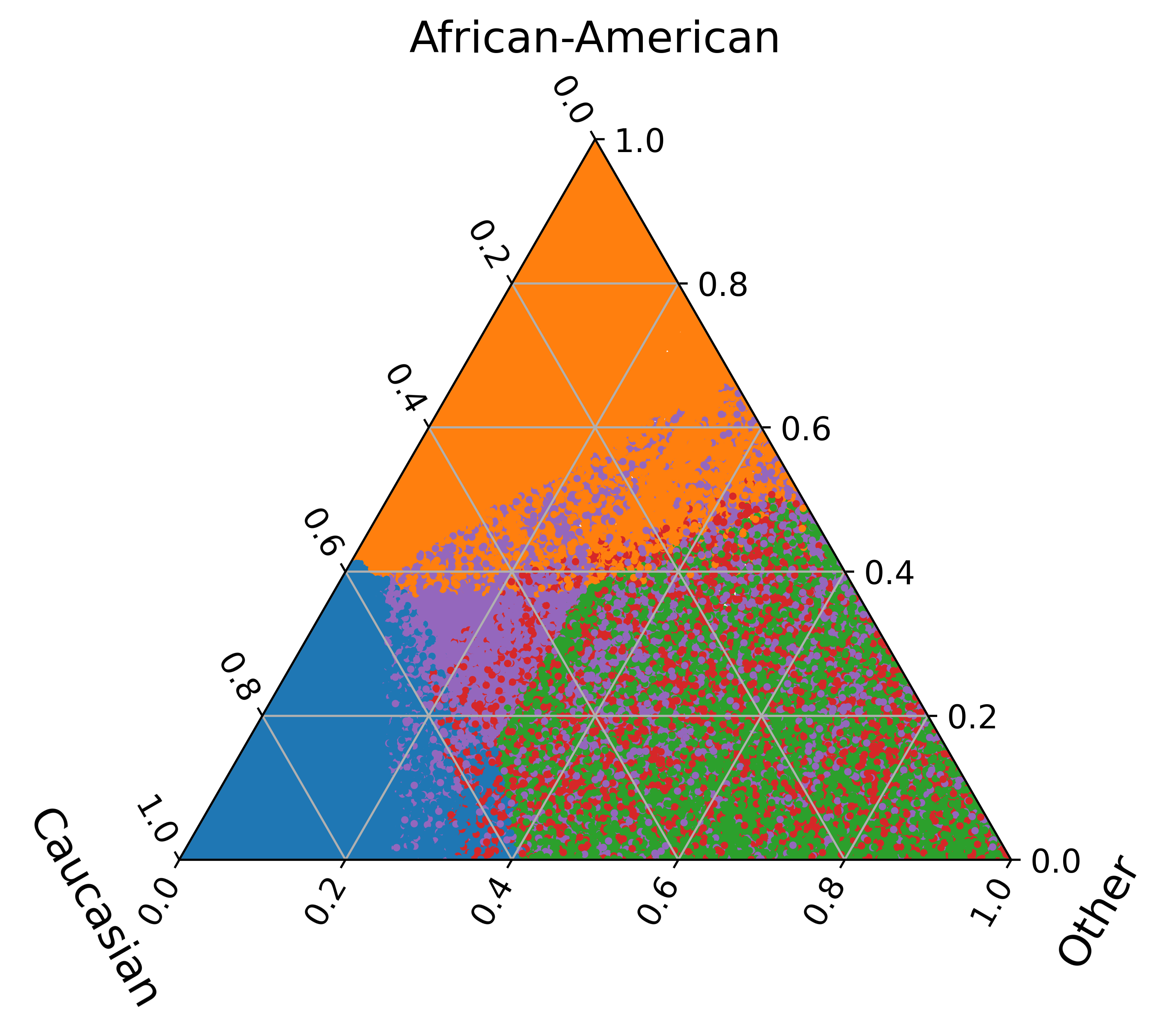}
        \caption{Data-driven threshold, approximating matching from 
       \Cref{fig:simplex-matching}.}
        \label{fig:simplex-svm}
    \end{subfigure}
\caption{Extending \Cref{fig:simplex-main}, classified labels onto a 3-dimensional probability simplex for other methods, with \textit{Hispanic}, \textit{Asian}, and \textit{Native American} probabilities aggregated into \textit{Other}.}
\label{fig:simplex-appendix}
\end{figure}

\end{document}